%% file: main_natbib.tex
\documentclass[a4paper,onecolumn,11pt, accepted=2026-04-28]{quantumarticle}
\pdfoutput=1

\usepackage[utf8]{inputenc}
\usepackage[english]{babel}
\usepackage[T1]{fontenc}
\usepackage{amsmath}
\usepackage[numbers,sort&compress]{natbib}
\usepackage{hyperref}
\usepackage{tikz}

\usepackage{physics}
\usepackage{amsmath}
\usepackage{amsfonts}
\usepackage{amssymb}
\usepackage{mathtools}

\usepackage{graphicx}
\definecolor{myGreen}{RGB}{18,159,87}
\definecolor{myGreenTwo}{RGB}{18,159,87}
\definecolor{quantumpurple}{RGB}{83,37,127}
\usepackage{pgfplots}
\usetikzlibrary{angles,
                graphs,
                graphs.standard, 
                fit,
                decorations.pathmorphing,
                decorations.pathreplacing,
                decorations.shapes,
                patterns, 
                patterns.meta,
                shadows,
                shapes.geometric,
                positioning,
                quantikz,
                decorations.pathreplacing,
                calligraphy}
\tikzset{snake it/.style={decorate, decoration=snake}}
\pgfplotsset{compat=1.17} 

\usepackage{url}
\usepackage{microtype}
\usepackage[capitalise]{cleveref}

\usepackage{algorithm2e}
\usepackage{algpseudocode}
\RestyleAlgo{ruled}
\SetKwComment{Comment}{/* }{ */}
\SetKwInOut{Input}{input}\SetKwInOut{Output}{output}

\usepackage{amsthm}
\usepackage{thmtools}
\usepackage{thm-restate}
\newtheorem{theorem}{Theorem}

\newtheorem{corollary}{Corollary}
\newtheorem{lemma}{Lemma}

\newtheorem*{result*}{Result}
\theoremstyle{definition}
\newtheorem{definition}{Definition}
\theoremstyle{remark}
\newtheorem{remark}{Remark}

\usepackage{paralist}
\usepackage{multirow}

\newcommand{\refcite}[1]{Ref.~\cite{#1}}
\newcommand{\V}{\mathcal{V}}
\newcommand{\E}{\mathcal{E}}

\newcommand{\poly}[1]{\textnormal{{\rm poly}}\left(#1\right)}
\newcommand{\cl}[1]{\textnormal{{\bf #1}}}
\newcommand{\clw}[2]{{\bf #1}{\rm -}{\rm #2}}
\newcommand{\Gate}[1]{\textnormal{\textsc{#1}}}
\renewcommand{\sc}[1]{\textnormal{\textsc{#1}}}

\input{tikz_spspgraph.tex}

\begin{document}

\title{The Complexity of Local Stoquastic Hamiltonians on 2D Lattices}

\author{Gabriel Waite}
\email{gabriel.waite@student.uts.edu.au}
\affiliation{Centre for Quantum Computation and Communication Technology,%
Centre for Quantum Software and Information, %
School of Computer Science, Faculty of Engineering \& Information Technology, University of Technology Sydney, NSW 2007, Australia
}

\author{Michael J. Bremner}
\email{michael.bremner@uts.edu.au}
\affiliation{Centre for Quantum Computation and Communication Technology,%
Centre for Quantum Software and Information, %
School of Computer Science, Faculty of Engineering \& Information Technology, University of Technology Sydney, NSW 2007, Australia
}
\maketitle

\begin{abstract}
    We show the \textsc{$2$-Local Stoquastic Hamiltonian} problem on a $2$D square qubit lattice is \textbf{StoqMA}-complete.
    We achieve this by extending the spatially sparse circuit construction of Oliveira and Terhal, as well as the perturbative gadgets of Bravyi, DiVincenzo, Oliveira, and Terhal.
    Our main contributions demonstrate \textbf{StoqMA} circuits can be made spatially sparse and that geometrical, stoquastic-preserving, perturbative gadgets can be constructed, without an increase to particle dimension.
\end{abstract}

\section{Introduction}
A fundamental challenge in quantum computing is to determine the computational complexity of ground-state energies of local Hamiltonians.
It is well-known that the \sc{Local Hamiltonian} problem is \cl{QMA}-complete, making it intractable for both classical and quantum computers~\cite{KSV02}.
Interesting variations of the problem have been extensively studied, including geometric restrictions~\cite{OT08, SV09, PM17, CM16}, sign restrictions~\cite{CM16,PM17}, specific Hamiltonian types~\cite{BBT06,BDOT06,BL08}, and commuting variants~\cite{BV05,Schuch11,AE11,AE14,AKV18,IJ23}.
These restrictions matter for modelling physically relevant many-body systems, especially geometrically local ones.
In this work we study the \sc{Local Stoquastic Hamiltonian} problem under 2D square-lattice constraints for qubit systems.

The class of Hamiltonians with non-positive off-diagonal elements in the computational basis, known as stoquastic Hamiltonians, are particularly significant as they circumvent the sign problem of Monte Carlo simulations~\cite{LGS+90,LY19,HRNE20}; as a result, Monte Carlo methods apply effectively~\cite{FMNR01,SAM10,Ohzeki17,BCGL22,BBH24}.
Additionally, stoquastic Hamiltonians can be used as a gateway method of `curing' the sign problem other, more general systems may exhibit~\cite{HBLSC95,BCGL22}.
To capture the problem of finding the ground-state energy of stoquastic Hamiltonians, the class \cl{StoqMA} was introduced~\cite{BBT06}, lying between \cl{MA} and \cl{QMA} in the complexity hierarchy.
Strong amplification procedures for this class requiring only a single copy of the proof state are not yet known~\cite{AGL21, Liu21, MW05}, making \cl{StoqMA} a compelling and technically challenging complexity class to study.
Important conjectures are established with respect to \cl{StoqMA}'s containment in \cl{MA} under certain amplification procedures~\cite{AGL21}; one such example is that \cl{StoqMA}\textsubscript{1} (\cl{StoqMA} with perfect completeness) is contained in \cl{MA}~\cite{BT09}.

Bravyi, Bessen, and Terhal~\cite{BBT06} demonstrated that the \sc{$k$-Local Stoquastic Hamiltonian} problem is \clw{StoqMA}{complete} for $k \geq 2$.
This was established via a circuit-to-Hamiltonian construction, similar to the original Feynman-Kitaev framework~\cite{KSV02}, showing that locality beyond two does not affect the problem's complexity.
Additionally, it was also shown that the \sc{$2$-Local Stoquastic Hamiltonian} problem is \clw{MA}{hard}, further supporting the connection between stoquasticity and classical complexity classes~\cite{BDOT06}.

In this work, we consider $2$-local stoquastic (qubit) Hamiltonians constrained to 2D regular lattices. 
This natural geometric constraint is often used for modelling physical systems, such as the Heisenberg model on a square lattice~\cite{SV09,PM17}.
We show the \clw{MA}{hardness} and \clw{StoqMA}{completeness} persist even under these geometric restrictions in qubit systems.
We establish this through perturbative gadget constructions (via Schrieffer-Wolff transformations) that allow for locality reductions, while preserving stoquasticity.
The existing gadget constructions of Ref.~\cite{OT08} enable subdivision of interactions, but do not preserve stoquasticity when applied verbatim.
We extend the techniques of Refs.~\cite{OT08,BDOT06,BBT06} to subdivide interaction edges while maintaining locality and term-wise stoquasticity of the resulting Hamiltonian.
We accomplish this by identifying an appropriate basis representation for $2$-local stoquastic Hamiltonians and constructing perturbative gadgets that respect this structure.
Our results hold for local term-wise stoquastic Hamiltonians on qubit systems.
This approach not only expands the scope of stoquastic gadget constructions but also provides new insights into the complexity of local stoquastic Hamiltonians on lattice geometries.

We identify a parent stoquastic Pauli Hamiltonian that is naturally \clw{StoqMA}{complete}.
This Hamiltonian is defined by a term-wise stoquastic construction where each term represents a $2$-local Pauli interaction.
A notable restriction of this class is the transverse-field Ising model~\cite{BH16, PM17}, making it a candidate for complexity reductions.
Although certain restrictions of this Hamiltonian are known to be contained in \cl{StoqMA}, it remains an open question whether these are \clw{StoqMA}{complete}.
Careful gadget constructions may resolve this question.

Our results suggest several directions for future investigation.
For instance, it is worth investigating whether specific stoquastic Pauli Hamiltonians on 2D lattices, such as the transverse field Ising model, are \clw{StoqMA}{complete}.
Another intriguing direction would be to explore the complexity of the antiferromagnetic Heisenberg model on bipartite lattices, a problem which is still only known to be contained in \cl{StoqMA}~\cite{BH16, PM17}.
Furthermore, we raise the question of whether the graphical structure of a local stoquastic Hamiltonian can be exploited to construct a guiding state~\cite{CFGHLGMW23}.
Can we reduce the degree of a planar graph representing a local stoquastic Hamiltonian to three while maintaining the problem's complexity? These are significant open questions that arise naturally from our work and point to a rich potential for future developments.

\subparagraph{Prior work.}
Recently, Raza, Eisert, and Grilo~\cite{REG24} studied the complexity of geometrically constrained stoquastic Hamiltonians, providing a direct embedding of quantum circuits without perturbative gadgets.
Their results include \clw{MA}{hardness} of high-dimensional qudit systems on $2$D lattices and even $1$D lines, as well as \clw{StoqMA}{completeness} for high-dimensional qudit systems on $2$D lattices.
A consequence of their approach is an increase in particle dimension, i.e., beyond the qubit regime, which stems from Toffoli gates being $3$-local in qubit systems.
However, this higher dimensionality allows for more direct applicability across a broader range of qudit dimensions and geometries, enhancing theoretical generality.
In contrast, our approach uses spatially sparse constructions and perturbative reductions to reduce to Hamiltonians of qubits on a square lattice, which has close ties to many physically relevant systems.
Moreover, our reduction technique draws on well-established methods from the literature, allowing application to and strengthening of a wider variety of problems, beyond the specific case considered here.
While Ref.~\cite{REG24} offers a compelling alternative perspective on the complexity of stoquastic Hamiltonians, their conclusions are complementary rather than contradictory to our results.
Both approaches contribute valuable insights, highlighting different trade-offs between generality and physical relevance.
An open problem at the intersection of these approaches is to classify local stoquastic qubit Hamiltonians on $1$D lines.

\subparagraph{Outline.} 
In \cref{sec:prelim}, we introduce the necessary background and definitions, including the definition of the class \cl{StoqMA}.
We then summarise the technical contributions of this work.
In \cref{sec:comp-stoq-ham-prob}, we review the appropriate Feynman-Kitaev clock construction necessary for stoquastic Hamiltonians, both for the classes \cl{MA} and \cl{StoqMA}, this establishes notation and intuition for the subsequent sections.
We also discuss and define a spatially sparse graph construction, then proceed to prove the \clw{StoqMA}{completeness} of the \sc{$6$-Local Stoquastic Hamiltonian} problem on spatially sparse graphs.
\cref{sec:stoq-pert-gadgets} outlines and summarises the preliminaries for the perturbative gadgets used in the reduction of the \sc{$6$-Local Stoquastic Hamiltonian} problem to the \sc{$2$-Local Stoquastic Hamiltonian} problem.
We then introduce new perturbative gadgets required for the geometric reduction in \cref{sec:stoq-geo-gadgets}.
Using the tools from the previous sections, we prove the \clw{StoqMA}{completeness} of the \sc{$2$-Local Stoquastic Hamiltonian} problem on $2$D lattices.
Finally, in \cref{sec:stoq-pauli-ham}, we discuss a parent stoquastic Pauli Hamiltonian that is naturally \clw{StoqMA}{complete} amongst other Pauli Hamiltonians that fall into the stoquastic regime.
The appendices contain various technical proofs concerning: gadget applications --- in \cref{app:parallel}, \cref{app:stoquastic_subdivision_gadget} and \cref{app:composition}, and circuit mappings --- in \cref{app:toffoli} and \cref{app:statistics}.

\section{Preliminaries and Technical Summary}\label{sec:prelim}
For brevity, we may sometimes omit identity terms and the operation $\otimes$ in tensor product strings; for example, $A\otimes I\otimes B \otimes  I = A_1B_3$.
Let the Pauli matrices be denoted as $X$, $Y$ and $Z$.
The $Z$-basis refers to the computational basis of qubits --- the eigenbasis of the Pauli-$Z$ operator, $\{\ket{0},\ket{1}\}$.
The $X$-basis refers to the eigenbasis of the Pauli-$X$ operator, $\{\ket{+},\ket{-}\}$.

A $k$-local Hamiltonian on $n$ qubits, $H = \sum_j H_j$, is a sum of local terms $H_j$ where the term acts non-trivially on a subset of at most $k$ qubits.
We assume $\norm{H_j} = O({\rm poly}(n))$ and each Hamiltonian term can be expressed using ${\rm poly}(n)$ bits.
Let the \emph{ground-state energy} of a Hamiltonian $H$ be denoted as $\lambda_0(H)$, which is the minimum eigenvalue of $H$.
A given local Hamiltonian admits an associated \emph{interaction (hyper)graph} $G = (\V,\E)$.
At each vertex of the graph there lies a two-dimensional Hilbert space, $\mathbb{C}^2$, representing a qubit.
Each (hyper)edge of the graph represents a local interaction term between qubits and will only contain those qubits acted on non-trivially by the Hamiltonian term.
\cref{fig:edge_base} demonstrates how we will visually represent a $2$-local interaction edge between two qubits, $u$ and $v$; where appropriate, the subscripts will be omitted.

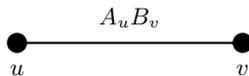
\begin{figure}[!ht]
    \centering
    \begin{tikzpicture}
        \pic{interaction};
    \end{tikzpicture}
    \caption{
        A pictorial representation of an interaction edge.
        The labels $u$/$v$ either represent single qubits or a set of qubits.
        The term $A_uB_v$ represents the local interaction between $u$ and $v$.
        }
    \label{fig:edge_base}
\end{figure}

We now present formal definitions of the relevant problem, local Hamiltonian family and complexity class for this work.

\begin{definition}[\sc{Local Hamiltonian}]
    Given a $k$-local Hamiltonian $H$ acting on an $n$ qubit system with parameters $a,b \in \mathbb{R}$ such that $b-a \geq 1/{\rm poly}(n)$, determine whether $\lambda_0(H) \leq a$ or $\lambda_0(H) > b$ promised one is true.
\end{definition}

Local stoquastic Hamiltonians are a subclass of local Hamiltonians, with each term having non-positive off-diagonal elements in the computational basis.

\begin{definition}[Stoquastic Hamiltonian]\label{def:stoq_ham}
    A Hamiltonian on $n$ qubits is said to be stoquastic if, in the computational basis, $\bra{x}H\ket{y} \leq 0$, for any $x\neq y$ where $\ket{x},\ket{y}\in (\mathbb{C}^2)^{\otimes n}$.
\end{definition}

We refer to a local stoquastic Hamiltonian as \emph{term-wise stoquastic} if each local term $H_j$ is stoquastic.
The \sc{Local Stoquastic Hamiltonian} problem is defined with respect to term-wise stoquastic Hamiltonians.
Unless otherwise stated, all results in this work refer to term-wise stoquastic Hamiltonians.
The Hermitian property of Hamiltonians implies that the off-diagonal and diagonal terms are real.
The \sc{Local Stoquastic Hamiltonian} problem is complete for the class \cl{StoqMA}.
The complexity class \cl{StoqMA} lies between \cl{MA} and \cl{QMA}, specifically,
\begin{equation*}
    \cl{P} \subseteq \cl{BPP} \subseteq \cl{MA} \subseteq \cl{StoqMA} \subseteq \cl{QMA}.
\end{equation*}
It is also known that the \sc{Local Stoquastic Hamiltonian} problem is hard for the class \cl{MA}~\cite{BDOT06}.
A formal definition of \cl{MA} follows in \cref{sec:comp-stoq-ham-prob}.

To define \cl{StoqMA}, we first introduce the concept of a stoquastic verification circuit.

\begin{definition}[Stoquastic Verification Circuit]\label{def:SVC}
    A stoquastic verification circuit is a tuple $S_n = (n,w,m,p,U)$ where $n$ is the number of input qubits, $w$ is the number of proof qubits, $m$ is the number of ancillae initialised in the $\ket{0}$ state and $p$ is the number of ancillae initialised in the $\ket{+}$ state.
    The circuit $U$ is a quantum circuit on $M\coloneqq n + w + m + p$ qubits, comprised of $T = O(\poly{n})$ gates from the set $\{X, \Gate{Cnot}, \Gate{Toffoli}\}$.
    The acceptance probability of a stoquastic verification circuit $S_n$, given some input string $x\in \varSigma^n$ and a proof state $\ket{\xi} \in \mathbb{C}^{2^w}$ is defined as:
    \begin{equation*}
        \Pr\left[S_n(x,\ket{\xi})\right]= \bra{\phi}U^\dagger \Pi_{\text{out}} U \ket{\phi},
    \end{equation*}
    where $\ket{\phi} = \ket{x,\xi,0^{m},+^{p}}$ and $\Pi_{\text{out}} = \ketbra{+}_1$ is a projector onto the output qubit.
\end{definition}

Note that $w,m,p = O(\poly{n})$.

\begin{definition}[\cl{StoqMA}($\alpha$,$\beta$)]\label{def:StoqMA}
    A promise problem $L = (L_{\textsc{yes}}, L_{\textsc{no}})$ belongs to the class \cl{StoqMA}($\alpha$,$\beta$) if there exists a polynomial-time generated stoquastic circuit family $\mathcal{S} = \{S_n : n \in \mathbb{N}\}$, where each stoquastic circuit $S_n$ acts on $n + w+m + p$ input qubits and produces one output qubit, such that:
    \begin{itemize}
        \item[] \textbf{Completeness}: For all $x\in L_{\textsc{yes}}$, $\exists \ket{\xi}\in(\mathbb{C}^2)^{\otimes w}$, such that, $ \Pr\left[S_{|x|}(x,\ket{\xi})=\mathtt{+} \right] \geq \alpha(|x|)$
        \item[] \textbf{Soundness}: For all $x\in L_{\textsc{no}}$, $\forall\ket{\xi}\in(\mathbb{C}^2)^{\otimes w}$, then, $ \Pr\left[S_{|x|}(x,\ket{\xi})=\mathtt{+} \right] \leq \beta(|x|)$
    \end{itemize}
    The term $\alpha$ refers to the completeness parameter and $\beta$ the soundness parameter, where $1/2 \leq \beta(|x|) < \alpha(|x|) \leq 1$ and satisfying $\alpha-\beta\geq\frac{1}{\poly{|x|}}$.
\end{definition}

Unlike \cl{QMA} and \cl{MA}, the completeness and soundness parameters cannot be amplified for \cl{StoqMA}.
It was conjectured by Aharanov, Grilo and Liu~\cite{AGL21} that 
\begin{equation*}
    \cl{StoqMA}(\alpha,\beta) \subseteq \cl{StoqMA}(1 - 2^{-l(n)},\frac{1}{2}+2^{-l(n)})
\end{equation*}
where $l(n)$ is some polynomial in the system size.
Interestingly, via a clever application of distribution testing, Liu~\cite{Liu21} was able to prove a method for soundness error-reduction, specifically
\begin{equation*}
    \cl{StoqMA}\left(\frac{1}{2} + \frac{\alpha}{2}, \frac{1}{2} + \frac{\beta}{2}\right) \subseteq \cl{StoqMA}\left(\frac{1}{2} + \frac{\alpha^r}{2}, \frac{1}{2} + \frac{\beta^r}{2}\right) 
\end{equation*}
where $r = \poly{n}$.
A slight downfall to this reduction is that it requires $r$ copies of the proof state.
We therefore assume the parameters are fixed as in \cref{def:StoqMA}.

\begin{remark}[Merlin's message]\label{rmk:Merlins_message}
    In the class \cl{StoqMA}, we describe the interaction between Merlin and Arthur by explicitly distinguishing the components each party contributes.
    This framing is consistent with the formal definition from Ref.~\cite{BBT06} and subsequent works~\cite{CM16,AGL21,Liu21,Gharibian24}.
    To concisely represent this interaction, we define a tuple $(\xi, S_{|x|})$, where $\xi$ is the proof state provided by Merlin and $S_{|x|}$ is the verification circuit controlled by Arthur.
    The verification process is structured as follows: $S_{|x|}$ takes as input the problem instance $x$, includes polynomially many $\ket{0}$- and $\ket{+}$-ancillae, and is described by a polynomial-sized stochastic verification circuit.
    Without loss of generality, we assume an even number of $\ket{+}$-ancillae in the circuit.
    This re-framing captures the essential aspects of the class \cl{StoqMA} while making the roles of Merlin and Arthur more intuitive.
\end{remark}

\subsection{Technical Summary}
The main contribution of this work establishes the \clw{StoqMA}{completeness} of the \sc{$2$-Local Stoquastic Hamiltonian} problem on $2$D lattices of qubit systems.
We follow the framework outlined by Oliveira and Terhal.
To start, we show that a \Gate{Swap} network can be used to map long-range \cl{StoqMA} circuits to nearest-neighbour circuits.
The number of gates increases by a factor $\Theta(n)$ using this approach and results in a polynomial overhead on the depth of the verification circuit.
Importantly, the circuit statistics are preserved, i.e., the completeness and soundness parameters are not affected.

We proceed by mapping general \cl{StoqMA} circuits to a spatially sparse circuit construction.
This entails a polynomial increase in the number of system qubits and a linear increase in the number of overall gates.
The reason is so that each qubit only interacts with a constant number of others in a geometrically local neighbourhood.
Surprisingly, this mapping also preserves the statistics of the circuit (see \cref{app:statistics}).
\cref{fig:stoqma_circuit_modifications} gives an overview of the required circuit modifications.

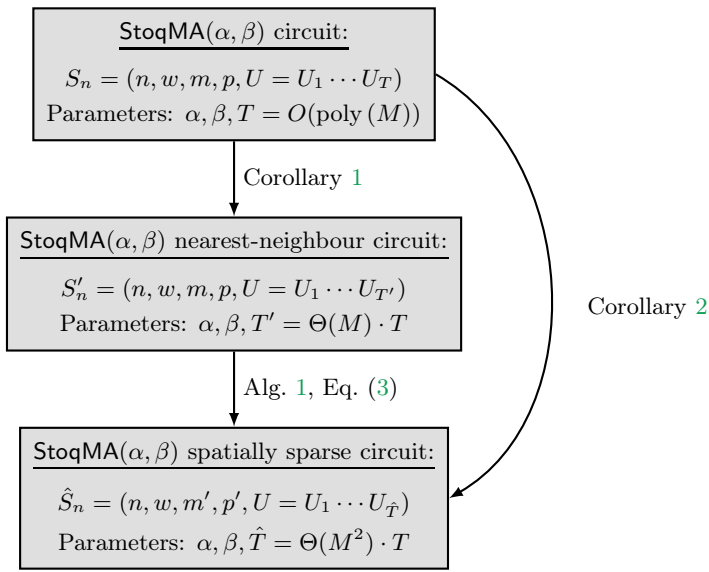
\begin{figure}[!ht]
    \centering
    \begin{tikzpicture}[
        node distance=10mm and 10mm,
            box/.style = {draw, minimum height=5mm, inner xsep=3mm, inner sep=5pt, align=center, thick},
            scale=1,transform shape
            ]
        \coordinate (orig) at (0,0);

        \node[fill=quantumpurple!25] (a0) [box, below=of orig] {\underline{$\cl{StoqMA}(\alpha, \beta)$ circuit:}  \\[0.2cm] $S_n = (n,w,m,p,U = U_1\cdots U_T)$ \\[0.1cm] Parameters: $\alpha, \beta, T = O(\poly{M})$ };
        \node[fill=quantumpurple!15] (a1) [box, below=of a0] {\underline{$\cl{StoqMA}(\alpha, \beta)$ nearest-neighbour circuit:}  \\[0.2cm] $S'_n = (n,w,m,p,U = U_1\cdots U_{T'})$ \\[0.1cm] Parameters: $\alpha, \beta, T' = \Theta(M)\cdot T$ };
        \node[fill=quantumpurple!25] (a2) [box, below=of a1] {\underline{$\cl{StoqMA}(\alpha, \beta)$ spatially sparse circuit:}  \\[0.2cm] $\hat{S}_n = (n,w,m',p',U = U_1\cdots U_{\hat{T}})$ \\[0.1cm] Parameters: $\alpha, \beta, \hat{T} = \Theta(M^2)\cdot T$};

        \draw[-latex,thick] (a0) -- (a1) node[midway, right] {\cref{cor:generic-StoqMA}};
        \draw[-latex,thick] (a1) -- (a2) node[midway, right] {Alg.~\ref{alg:sp_sp_circuit}, \cref{eq:spatially_sparse_gate_sequence}};
        \draw[-latex,thick] (a0.east) to [out=-30, in=30] (a2.east) node[midway,xshift=6cm, yshift=-5cm] {\cref{cor:generic-StoqMAtospsp}};
        \node at (-6.5,-5) {};
    \end{tikzpicture}
    \caption{
        Workflow of the required circuit modifications.
        We take generic (long-range) \cl{StoqMA} circuits to ones comprised of only nearest-neighbour gates.
        A subsequent mapping takes such circuits to the spatially sparse construction.
        Here, $M \coloneqq n+w+m+p$ is the total number of qubits in the circuit.
        Additionally, $m'>m$ and $p'>p$ are the number of ancilla qubits required for the spatially sparse construction.
        The important parameters at each stage are: the completeness parameter, the soundness parameter and the gate count.
    }
    \label{fig:stoqma_circuit_modifications}
\end{figure}

We then employ the (stoquastic) Feynman-Kitaev clock construction for the spatially sparse circuit, proving the \clw{StoqMA}{completeness} of the \sc{$6$-Local Stoquastic Hamiltonian} problem on spatially sparse graphs.
This mapping is essentially the same as the standard \clw{StoqMA}{hardness} proof but takes special care with how the clock qubits are oriented.

\begin{restatable*}[]{theorem}{thrmsixLSHPspsp}
    \label{thm:stoqma-complete_k-local_stoq_ham}
    The \sc{$6$-Local Stoquastic Hamiltonian} problem on a spatially sparse graph is \clw{StoqMA}{complete}.
\end{restatable*}

We then review the perturbative gadgets used in the reduction of the \sc{$6$-Local Stoquastic Hamiltonian} problem to the \sc{$2$-Local Stoquastic Hamiltonian} problem~\cite{BBT06}.
Motivated by the design of these gadgets, specifically the subdivision gadget, we construct a family of new perturbative gadgets required for the geometric reduction.
Using intuition from \refcite{OT08}, it is possible to reduce the \sc{$2$-Local Stoquastic Hamiltonian} problem to a degree-$4$ planar graph.
However, the previous geometric gadget constructions for general local Hamiltonians~\cite{OT08} rely on a Pauli-basis decomposition and do not preserve stoquasticity; consequently, such methods cannot be directly applied in our setting.
The technical challenge in this part is therefore ensuring that each gadget term is stoquastic, that no unwanted terms are introduced, and that second-order perturbative gadgets produce no undesired cross terms.
To achieve our results, we first decompose general $2$-local stoquastic Hamiltonians into a specific basis of single-qubit matrices.
Specifically, since we consider computational basis stoquastic Hamiltonians, we define a basis denote as $\rho^{\mu}$ for $\mu \in \{0,1,2,3\}$ where $\rho^0 = \ketbra{0}, \rho^1 = \ketbra{0}{1}, \rho^2 = \ketbra{1}{0}, \rho^3 = \ketbra{1}$.
Using this basis, we can express any local stoquastic Hamiltonian as a sum of diagonal terms and non-positive off-diagonal terms produced by tensored combinations of $\rho^{\mu}$ matrices.
Within this decomposition, the construction of perturbative gadgets becomes analogous to that of Ref.~\cite{OT08}, but with careful attention to stoquasticity.

We introduce three additional stoquastic-preserving perturbative gadgets: the \sc{Cross}, the \sc{Fork}, and the \sc{Triangle} gadgets.
Additionally, we decompose a general $2$-local stoquastic Hamiltonian into a sum over a specific basis of $2$-local stoquastic interactions.
We then employ the Schrieffer-Wolff transformation to analyse the effect of the perturbative gadgets and prove that the constructed qubit Hamiltonians are term-wise stoquastic and $2$-local.
Importantly, the reduction from a spatially sparse graph to a planar graph is carried out using only a constant number of perturbative gadgets.

\newpage 
\begin{restatable*}{theorem}{planargraph}
    Given a $2$-local stoquastic Hamiltonian on a spatially sparse graph, $H$, there exists a $2$-local stoquastic Hamiltonian on a degree-$4$ planar graph with a straight-line drawing in the plane that approximates $H$.
\end{restatable*}

A subsequent embedding procedure from the planar graph to the $2$D square lattice is then discussed.
Since this mapping is efficient and only requires an additional constant number of perturbative gadgets, we establish the \clw{StoqMA}{completeness} of the \sc{$2$-Local Stoquastic Hamiltonian} problem on $2$D lattices.
Our results hold for local term-wise stoquastic Hamiltonians on qubit systems, i.e., where each local term is stoquastic as per \cref{def:stoq_ham}.

\begin{restatable*}[]{theorem}{thrmTwoLSHsquarelattice}
    \label{thrm:2LSH_square_lattice}
    The \sc{$2$-Local Stoquastic Hamiltonian} problem on a $2$D square lattice is \clw{StoqMA}{complete}.
\end{restatable*}

The contributions of this work are summarised in \cref{fig:flow-diagram}.

\begin{figure}[!ht]
    \centering
    \pgfdeclarelayer{background layer}
    \pgfdeclarelayer{foreground layer}
    \pgfsetlayers{background layer,main,foreground layer}
    \begin{tikzpicture}[
        node distance=10mm and 10mm,
            box/.style = {draw=none, minimum height=5mm, inner xsep=3mm, inner sep=5pt, align=center, thick},
             scale=0.6,transform shape
            ]
        \coordinate (orig) at (0,0);
        \coordinate (shiftorig) at (0,-8);
    
        \node[fill=gray!25,label={[draw=black, thick,fill=white,inner sep=3pt]below:Ref.~\cite{BDOT06}}] (a0) [box, below=of orig] {\sc{6-Loc. Stoq. Ham.} \\[0.1cm] \clw{MA}{hard}};
        \node[fill=gray!25,label={[draw=black, thick,fill=white,inner sep=3pt]below:Ref.~\cite{BDOT06}}] (a1) [box, right=of a0] {\sc{2-Loc. Stoq. Ham.} \\[0.1cm] \clw{MA}{hard}};
    
        \node[fill=gray!25,label={[draw=black, thick,fill=white,inner sep=3pt]below:Ref.~\cite{BBT06}}] (b0) [box, below right=of a0] {\sc{6-Loc. Stoq. Ham.} \\[0.1cm] \clw{StoqMA}{complete}};
        \node[fill=gray!25,label={[draw=black, thick,fill=white,inner sep=3pt]below:Ref.~\cite{BBT06}}] (b1) [box, right=of b0] {\sc{2-Loc. Stoq. Ham.} \\[0.1cm] \clw{StoqMA}{complete}};
    
        \node[fill=quantumpurple!35,label={[draw=black,thick,fill=white,inner sep=3pt]below:This work.}] (c0) [box, above right=of a0] {\sc{6-Loc. Stoq. Ham.} \\[0.1cm] Spatially Sparse \\[0.1cm] \clw{MA}{hard}};
        \node[fill=quantumpurple!35,label={[draw=black,thick,fill=white,inner sep=3pt]below:This work.}] (c1) [box, right=of c0] {\sc{2-Loc. Stoq. Ham.} \\[0.1cm] Spatially Sparse \\[0.1cm] \clw{MA}{hard}};
    
        \node[fill=quantumpurple!35,label={[draw=black,thick,fill=white,inner sep=3pt]below:This work.}] (d0) [box, below right=of b0] {\sc{6-Loc. Stoq. Ham.} \\[0.1cm] Spatially Sparse \\[0.1cm] \clw{StoqMA}{complete}};
        \node[fill=quantumpurple!35,label={[draw=black,thick,fill=white,inner sep=3pt]below:This work.}] (d1) [box, right =of d0, xshift=2cm] {\sc{2-Loc. Stoq. Ham.} \\[0.1cm] Square Lattice \\[0.1cm] \clw{StoqMA}{complete}};
        \begin{scope}
            \node[fill=quantumpurple!35,label={[draw=black,thick,fill=white,inner sep=3pt]below:This work.}] (ld0) [box, below=of shiftorig] {\sc{6-Loc. Stoq. Ham.} \\[0.1cm] Spatially Sparse \\[0.1cm] \clw{StoqMA}{complete}};
            \node[fill=quantumpurple!15,label={[draw=black,thick,fill=white,inner sep=3pt]below:This work.}] (ld1) [box, right=of ld0] {\sc{2-Loc. Stoq. Ham.} \\[0.1cm] Spatially Sparse \\[0.1cm] \clw{StoqMA}{complete}};
            \node[fill=quantumpurple!15,label={[draw=black,thick,fill=white,inner sep=3pt]below:This work.}] (ld2) [box, right=of ld1] {\sc{2-Loc. Stoq. Ham.} \\[0.1cm] Planar Graph \\[0.1cm] \clw{StoqMA}{complete}};
            \node[fill=quantumpurple!35,label={[draw=black,thick,fill=white,inner sep=3pt]below:This work.}] (ld3) [box, right=of ld2] {\sc{2-Loc. Stoq. Ham.} \\[0.1cm] Square Lattice \\[0.1cm] \clw{StoqMA}{complete}};
            \node[fill=quantumpurple!15,label={[draw=black,thick,fill=white,inner sep=3pt]below:This work.}] (ld4) [box, right=of ld3] {\sc{2-Loc. Stoq. Ham.} \\[0.1cm] Triangular Lattice \\[0.1cm] \clw{StoqMA}{complete}};

            \draw [decorate, decoration = {calligraphic brace,raise=0pt,aspect=0.65,amplitude=0.85cm}, line width=0.5mm] (-2,-9)--(22,-9);
        \end{scope}
        \begin{pgfonlayer}{background layer}
            \draw[-latex, very thick, gray] (a0) -- (a1);
            \draw[-latex, very thick, gray] (a0.south) -- (b0.west);
            \draw[-latex, very thick, gray] (b0) -- (b1);
            \draw[-latex, very thick] (c0) -- (c1);
            \draw[-latex, very thick] (d0) --node[fill=white,rotate=0,inner sep=-2.5pt,outer sep=0]{//} (d1);
            \draw[-latex, very thick] (ld0) -- (ld1);
            \draw[-latex, very thick] (ld1) -- (ld2);
            \draw[-latex, very thick] (ld2) -- (ld3);
            \draw[-latex, very thick] (ld3) -- (ld4);
            \draw[-latex, very thick] (a0.north) -- (c0.west);
            \draw[-latex, very thick] (b0.south) -- (d0.west);
        \end{pgfonlayer}
    \end{tikzpicture}
    \caption{A flow diagram of the complexity of the \sc{Local Stoquastic Hamiltonian} problem.
            Arrows represent modifications/reductions to the problem.
            Grey boxes represent the results of prior work.
            Purple boxes represent the results of this work.
            Results below the horizontal brace are intermediate steps in the reduction to the \sc{$2$-Local Stoquastic Hamiltonian} problem on a $2$D square lattice.
            }
    \label{fig:flow-diagram}
\end{figure}
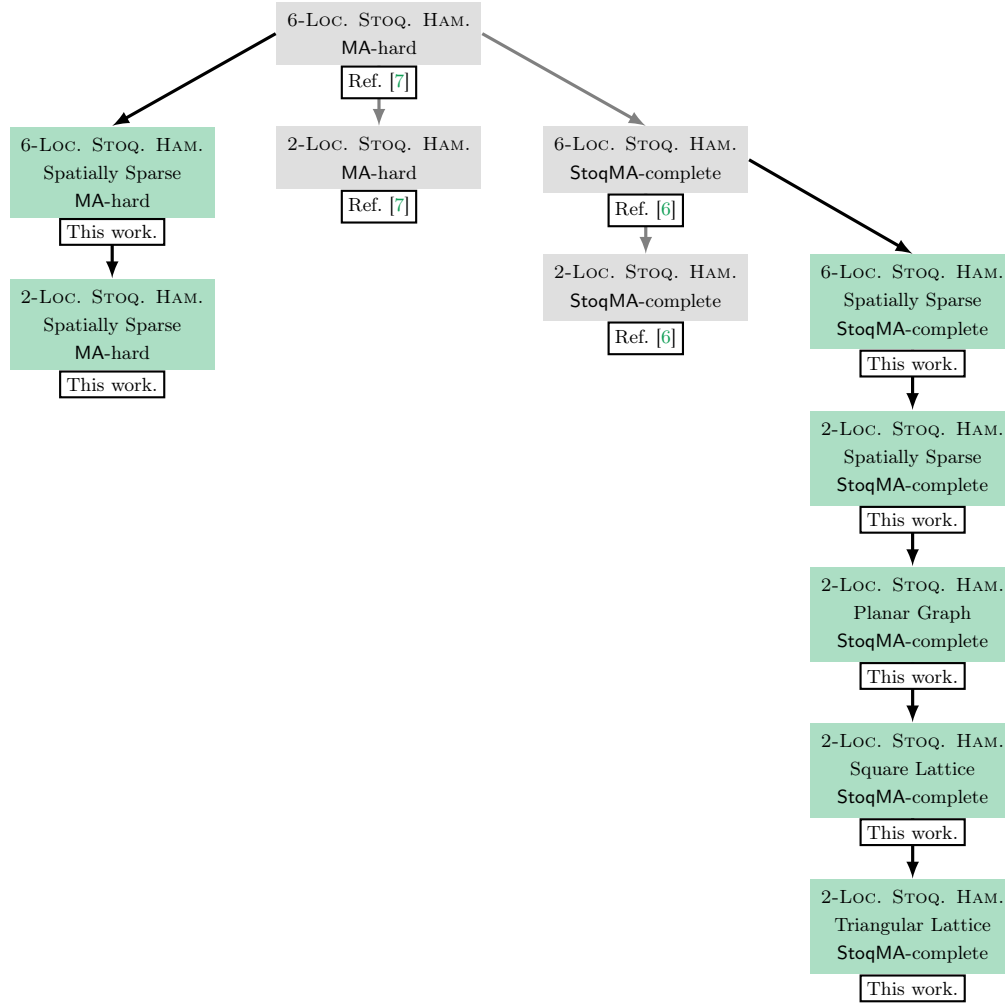

\section{The Complexity of the Stoquastic Hamiltonian Problem}\label{sec:comp-stoq-ham-prob}
Being a restriction of the standard \sc{Local Hamiltonian} problem implies the \sc{Local Stoquastic Hamiltonian} problem is naturally in \cl{QMA}.
Furthermore, it is also at least \cl{NP}-hard since all classical local Hamiltonians are stoquastic.
The main challenge is determining whether the problem is complete for any complexity class.
Bravyi \emph{et al}.~\cite{BDOT06} considered a semi-classical variant of \cl{MA} called \cl{MA}\textsubscript{q}, which subsequently led to the \clw{MA}{hardness} of the \sc{$6$-Local Stoquastic Hamiltonian} problem.
The result followed from the conclusion that \cl{MA}\textsubscript{q} = \cl{MA}.
A later work from Bravyi \emph{et al}.~\cite{BBT06} considered a more `quantum'  variant of \cl{MA} called \cl{StoqMA}, where it followed that the \sc{$2$-Local Stoquastic Hamiltonian} problem was \clw{StoqMA}{complete}.
Analysing the results concerning the \clw{MA}{hardness} is fruitful for the \clw{StoqMA}{hardness} results that follow.

\begin{definition}[Semi-Classical Verification Circuit]
    A semi-classical verification circuit is a tuple $F_n = (n,w,m,p,U)$ where $n$ is the number of input qubits, $w$ is the number of proof qubits, $m$ is the number of ancillae initialised in the $\ket{0}$ state and $p$ is the number of ancillae initialised in the $\ket{+}$ state.
    The circuit $U$ is a quantum circuit on $M\coloneqq n + w + m + p$ qubits, comprised of $T = O(\poly{n})$ gates from the set $\{X, \Gate{Cnot}, \Gate{Toffoli}\}$.
    The acceptance probability of a semi-classical verification circuit $F_n$, given some input string $x\in \varSigma^n$ and a proof state $\ket{\xi} \in \mathbb{C}^{2^w}$ is defined as:
    \begin{equation*}
        \Pr\left[F_n(x,\ket{\xi})\right]= \bra{\phi}U^\dagger \Pi_{\text{out}} U \ket{\phi},
    \end{equation*}
    where $\ket{\phi} = \ket{x,\xi,0^{m},+^{p}}$ and $\Pi_{\text{out}} = \ketbra{1}_1$ is a projector onto the output qubit.
\end{definition}

Note that this definition differs from \cref{def:SVC} in that the output qubit is measured in the $Z$-basis.

\begin{definition}[\cl{MA}\textsubscript{q}~\cite{BDOT06}]
    A promise problem $L = (L_{\textsc{yes}}, L_{\textsc{no}})$ belongs to the class \cl{MA}\textsubscript{q} if there exists a polynomial-time generated stoquastic circuit family $\mathcal{F} = \{F_n : n \in \mathbb{N}\}$, where each semi-classical circuit $F_n$ acts on $n + w+m + p$ input qubits and produces one output qubit, such that:
    \begin{itemize}
        \item[] \textbf{Completeness}: For all $x\in L_{\textsc{yes}}$, $\exists \ket{\xi}\in(\mathbb{C}^2)^{\otimes w}$, such that, $ \Pr\left[F_{|x|}(x,\ket{\xi})=\mathtt{1} \right] \geq 2/3$
        \item[] \textbf{Soundness}: For all $x\in L_{\textsc{no}}$, $\forall\ket{\xi}\in(\mathbb{C}^2)^{\otimes w}$, then, $ \Pr\left[F_{|x|}(x,\ket{\xi})=\mathtt{1} \right] \leq 1/3$
    \end{itemize}
\end{definition}

Note that $w,m,p = O(\poly{n})$.
The purpose of considering a semi-classical variant was to promote the \cl{BPP} verification circuit to a classically reversible quantum circuit.
Restricting Merlin to sending only classical proof states allows for the conclusion that \cl{MA}\textsubscript{q} = \cl{MA}~\cite[Lemma 2]{BDOT06}.
A key difference between \cl{MA}\textsubscript{q} and \cl{StoqMA} is that Arthur measures only in the $Z$-basis for \cl{MA}\textsubscript{q} and only in the $X$-basis for \cl{StoqMA}.
The $X$-basis measurement for \cl{StoqMA} circuits makes amplification of the completeness and soundness difficult using current techniques.
\cl{MA}\textsubscript{q}, on the other hand, does admit amplification, and it is known that \cl{MA}\textsubscript{$1$} = \cl{MA}~\cite{ZF87}; hence the same applies for \cl{MA}\textsubscript{q}.

A direct application of Feynman-Kitaev circuit-to-Hamiltonian construction on \cl{MA}\textsubscript{q} circuits was sufficient to show the \clw{MA}{hardness} of the \sc{$6$-Local Stoquastic Hamiltonian} problem.
An alternate method is required to prove \clw{StoqMA}{hardness} of the problem due to the inability to apply the circuit-to-Hamiltonian construction exactly for \cl{StoqMA} circuits.
This is a consequence of the fact that we do not have robust methods to amplify the completeness and soundness parameters.
In this section, we recap both the \cl{MA}\textsubscript{q}-hardness and \clw{StoqMA}{hardness} results.
We then show how the spatially sparse construction of \refcite{OT08} can be used to prove the \cl{MA}\textsubscript{q}-hardness and \clw{StoqMA}{hardness} of $6$-local stoquastic Hamiltonians.
We cover both ideas since a large portion of the work for the \clw{StoqMA}{hardness} result is in the \clw{MA}{hardness} proof.
The completeness of the results follows directly from Ref.~\cite{BBT06}.

\subsection{MA-hardness of the Local Stoquastic Hamiltonian Problem}
We recap the original \clw{MA}{hardness} proof of $6$-local stoquastic Hamiltonians due to Bravyi, DiVincenzo, Oliveira and Terhal~\cite{BDOT06}, then outline the important steps of the spatially sparse extension.
The process of extending the Feynman-Kitaev construction to a spatially sparse lattice is well-known~\cite{OT08,PM17}.
Our goal is to demonstrate how this idea can apply to local stoquastic Hamiltonians.

\begin{theorem}[\cite{BDOT06}]\label{thrm:ma-hard_6-local_stoq_ham}
    The \sc{$6$-Local Stoquastic Hamiltonian} problem is \textnormal{\clw{MA}{hard}}.
\end{theorem}

\begin{proof}
    We employ the Feynman-Kitaev circuit-to-Hamiltonian construction to prove the problem is hard for the class \cl{MA}.
    Let $F_{|x|}$ be Arthur's semi-classical verification circuit.
    Recall \cl{MA}\textsubscript{q} = \cl{MA}.
    Let the input to the circuit be an $N \coloneqq n + w+ m + p$ qubit register comprised of four parts: the input state $\ket{x}$ of $n$ qubits, the proof state $\ket{\xi}$ of $w$ qubits, the \textit{ancilla} register of $m$ qubits initialised to $\ket{0}$ and the \textit{coin} register of $p$ qubits initialised to $\ket{+}$.
    Let $F_{|x|}$ comprise a sequence of $T$ Toffoli gates denoted as $R_T, \dots, R_1$.

    Define a Hamiltonian $H = H_{\text{in}} + H_{\text{out}} + H_{\text{prop}} + H_{\text{clock}}$ acting on a register of $T$ \textit{clock} qubits labelled as $c_1, \dots, c_T$ and the $N$-qubit input register.
    Let the output measured qubit be denoted $q$; for this instance, Arthur can measure using only the $Z$-basis.
    Each Hamiltonian term is defined as a penalising Hamiltonian and must be stoquastic.
        
    \begin{align}
        H_{\text{in}} &\coloneqq \left(\sum_{i=1}^{n} \ketbra{\bar{x}_i}_i + \sum_{j=1}^{m} \ketbra{1}_{\textit{anc},j} + \sum_{i=1}^{p} \ketbra{-}_{\textit{coin},i}\right)\otimes\ketbra{0}_{c_1}, \notag\\
        H_{\text{out}} &\coloneqq \ketbra{0}_q \otimes \ketbra{1}_{c_T}, \label{eq:6LH_hard_proof_H-out}\\
        H_{\text{clock}} &\coloneqq \sum_{t=1}^{T-1} \ketbra{01}_{c_{t},c_{t+1}}, \label{eq:6LH_hard_proof_H-clock}\\
        H_{\text{prop}} &\coloneqq \sum_{t=1}^{T} H_{\text{prop}}(t).\notag
    \end{align}

    We define the propagation Hamiltonian terms in the following way:
    \begin{align*}
        H_{\text{prop}}(1) &= \ketbra{00}_{c_1,c_2} + \ketbra{10}_{c_1,c_2} - R_1\otimes(\ketbra{10}{00}_{c_1,c_2} + \ketbra{00}{10}_{c_1,c_2}), \\
        H_{\text{prop}}(t) &= \ketbra{100}_{c_{t-1},c_t,c_{t+1}} + \ketbra{110}_{c_{t-1},c_t,c_{t+1}} \notag\\
                            &\qquad - R_t\otimes(\ketbra{110}{100}_{c_{t-1},c_t,c_{t+1}} + \ketbra{100}{110}_{c_{t-1},c_t,c_{t+1}}), \quad 1<t<T \\
        H_{\text{prop}}(T) &= \ketbra{10}_{c_{T-1},c_T} + \ketbra{11}_{c_{T-1},c_T} - R_T\otimes(\ketbra{11}{10}_{c_{T-1},c_T} + \ketbra{10}{11}_{c_{T-1},c_T}).
    \end{align*}

    Note that $H_{\text{in}}$, $H_{\text{out}}$ and $H_{\text{clock}}$ are all $2$-local Hamiltonians.
    The terms $H_{\text{prop}}(1)$ and $H_{\text{prop}}(T)$ are $5$-local and $H_{\text{prop}}(t)$ terms are $6$-local.
    It is straightforward to show each Hamiltonian term is stoquastic.
    Notice that $\ketbra{-} = \frac{1}{2}( I - X)$, $\ketbra{1} = \frac{1}{2}( I + Z)$ and $H_{\text{out}}$, $H_{\text{clock}}$ are diagonal; hence $H_{\text{in}}$, $H_{\text{out}}$ and $H_{\text{clock}}$ are all $2$-local \emph{stoquastic} Hamiltonians.
    The terms $R_t\otimes(\dots)$ in $H_{\text{prop}}(t)$ will have off-diagonal elements that are strictly positive.
    Therefore, each $H_{\text{prop}}(t)$ term is stoquastic.
    To conclude, we simply leverage the original arguments from \refcite{KSV02} to show that in the \textsc{yes} case, there exists a proof state such that the Hamiltonian $H$ has eigenvalues at most $\epsilon/(T+1)$.
    In the \textsc{no} case, all eigenvalues are at least $c(1-\sqrt{\epsilon})/T^3$.
\end{proof}

\subsection{Circuit Modifications}
Converting a standard \cl{MA}\textsubscript{q} or \cl{StoqMA} circuit to one which is spatially sparse in design requires two steps.
The first is to replace all long-range gates with ones that are nearest-neighbour.
Then, we map the nearest-neighbour circuit to one with a large number of ancillae qubits and some extra gates.
The general purpose of the spatially sparse modification is to make it so that each qubit is only acted on by a constant number of gates.
Instead of analysing both \cl{MA}\textsubscript{q} and \cl{StoqMA} circuits separately, we will focus on the \cl{StoqMA} case.
The same logic applies to \cl{MA}\textsubscript{q} circuits.
The workflow of these reductions is given in \cref{fig:stoqma_circuit_modifications}.

\subparagraph{Nearest-neighbour circuits.} The \emph{range} of a gate we define as the maximum distance between any two qubits the gate acts on.
The distance metric is the number of registers between the two registers the qubits belong to.
The worst case is when a gate acts on the first and last qubits in the circuit register, giving a range of $M-1 = \Theta(M)$.
To replace a long-range gate with a nearest-neighbour one, we employ a \Gate{Swap} network procedure.
The \Gate{Swap} network is a sequence of \Gate{Cnot} gates that swap the qubits of the long-range gate to be nearest-neighbour.
See \cref{app:toffoli} for more information.

\begin{restatable}{proposition}{proptoffdecomp}\label{prop:toff-decomp}
    Let $\Gate{Toffoli}_r[a,b;c]$ be a Toffoli gate with range $r>2$.
$\Gate{Toffoli}_r[a,b;c]$ can be exactly expressed using $\Theta(r)$ nearest-neighbour \Gate{Cnot} gates and a single nearest-neighbour $\Gate{Toffoli}_3[b-1,b;b+1]$.
\end{restatable}

Without loss of generality, we can assume all gates in \cl{StoqMA} (and also \cl{MA}\textsubscript{q}) circuits are Toffoli gates.
The consequence of this proposition is the following corollary.
Essentially, we can take any \cl{StoqMA} circuit and decompose it into a nearest-neighbour one using only nearest-neighbour gates at the cost of the number of gates increasing by a factor of $\Theta(M)$ in the worst case.

\begin{restatable}{corollary}{corgenericStoqMA}
    \label{cor:generic-StoqMA}
    Given any long-range $\cl{StoqMA}(\alpha,\beta)$ circuit with $T$ gates on $M$ qubits, there exists a nearest-neighbour \cl{StoqMA}($\alpha,\beta$) circuit with $\Theta(T\cdot M)$ gates on $M$ qubits.
\end{restatable}

The key conclusion is that the completeness and soundness parameters of the original circuit are preserved in the transformed circuit (\cref{app:statistics}).
This is crucial for the subsequent spatially sparse construction.
The only important differences are the range and number of gates in the circuit.
As a final note, we emphasise that a \Gate{Swap} gate is just a sequence of \Gate{Cnot} gates and hence can be constructed in the \cl{StoqMA} circuit framework\footnote{It is likely that gate optimisation can be employed to reduce the overall gate overhead and circuit depth; however for our purposes, this is not necessary.}.

\subparagraph{Spatially sparse circuits.}
This construction is based on the idea of \emph{rounds} of gate executions.
Each round comprises of the application of one non-trivial gate in the verification sequence.
After the round has been executed, a series of \Gate{Swap} gates are employed between rows of qubits.
Each row of qubits is used for one (original) gate execution; if there are $T$ gates in the sequence, there are $T$ rows of qubits.
For an input register of $M$ qubits, each row is $M$ qubits `wide', giving a total of $P=T\cdot M$ qubits in the new construction.
Note that we still only require one copy of the proof state, initialised on the first row, and so the $T-1$ remaining rows are all ancillae; roughly speaking, $P = n\cdot O(\poly{n})$.
The \Gate{Swap} sequence is conducted between rows from `right to left'.
The time with respect to the gate sequence execution is taken from row-$1$ to row-$T$ (vertically).
Time with respect to \emph{all} gates employed (gate sequence and \Gate{Swap} sequence) follows a snake-like pattern; this will be characterised using a time cursor.

\begin{figure}[!ht]
    \centering
    \begin{tikzpicture}
        \pic[scale=0.8]{sp-sp-graph};
    \end{tikzpicture}
    \caption{A visual representation of the modified Feynman-Kitaev construction.
Each gate in the verification sequence is applied to a row of qubits in succession.
After each round of gate applications, a \Gate{Swap} gate sequence is applied between rows of qubits from right to left.
The time cursor is shown in green.
The small circles represent qubits.
The dashed boxes represent one of the non-trivial gates in the verification sequence.
Qubits that have no dashed box are assumed to be acted on trivially.}
    \label{fig:swap_sequence}
\end{figure}
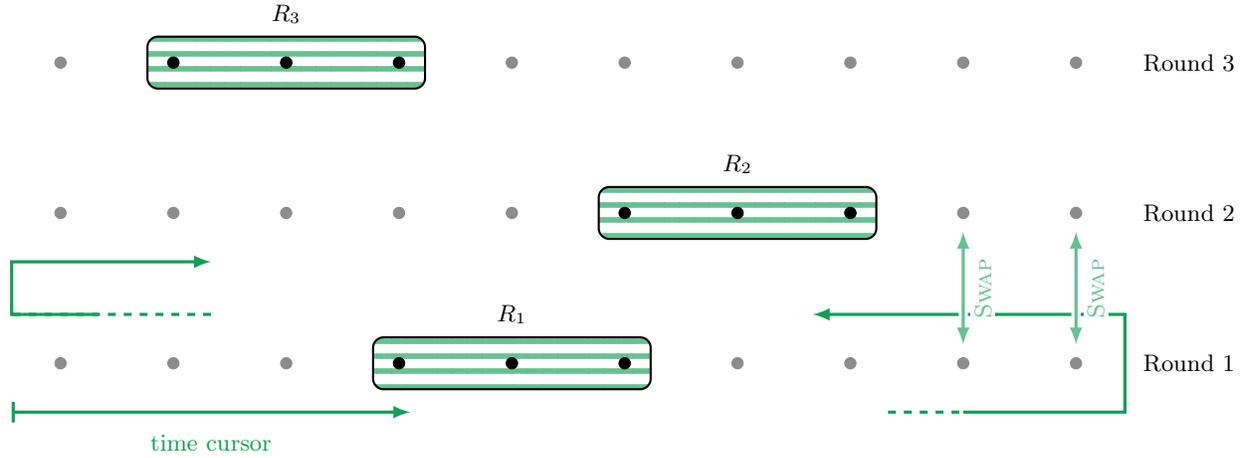

\newpage
\begin{definition}[Spatially Sparse Graph~\cite{OT08}]
    A spatially sparse graph is defined such that
    \begin{enumerate}[(i)]
        \item every vertex participates in $O(1)$ edges,
        \item there is a straight-line drawing in the plane such that every edge overlaps with $O(1)$ other edges, and the length of every edge is $O(1)$.
    \end{enumerate}
\end{definition}

Assuming that the original circuit \cl{StoqMA} circuit has been mapped to one comprised of only nearest-neighbour gates, we can now apply the spatially sparse construction.
The pseudocode in Alg.~\ref{alg:sp_sp_circuit} demonstrates how the spatially sparse construction operates.
Additionally, \cref{eq:spatially_sparse_gate_sequence} shows the modification to the gate sequence.
Introducing \Gate{Swap} gates modifies the original sequence of $T$ gates to a sequence of $T + M(T-1)$ gates.
Label the $M$ qubits in row-$j$ from left to right as $j_1, j_2, \dots, j_M$.
The new gate sequence follows
\begin{equation}\label{eq:spatially_sparse_gate_sequence}
    R_1R_2\dots R_T ~~\mapsto~~ R_1\;\left(\prod_{j=2}^T \big(\prod_{q=M}^1 \Gate{Swap}_{{j-1}_q, j_q} \big)\;R_j\right).
\end{equation}
Notice from previous arguments that \cref{eq:spatially_sparse_gate_sequence} is still a valid gate sequence even for the stoquastic verifier.
To do this mapping, we need additional $\Theta(T\cdot M)$ gates and $\Theta(T)$ ancillae qubits.

\begin{algorithm}[!ht]
    \SetKwData{Left}{left}\SetKwData{This}{this}\SetKwData{Up}{up}
    \SetKwFunction{Union}{Union}\SetKwFunction{FindCompress}{FindCompress}
    \SetKwInOut{Input}{input}\SetKwInOut{Output}{output}
    \Input{$T$ rows of $M$ qubits labelled (from left to right) as $j_1, \dots, j_M$ where $j\in[T]$, and a verification circuit of $T$ gates labelled $R_1,\dots, R_T$}
    \BlankLine
    \Output{Executed circuit $S_{|x|}$ on the $TM$ qubit register}
    \BlankLine
    \BlankLine
    Apply the gate $R_T$ to the appropriate qubits in row-$T$\\
    Apply the Identity gate to the remaining qubits in row-$T$\\
    \For{$j = 2$ \KwTo $j = T$}{
        Apply the gate $R_j$ to the appropriate qubits in row-$j$\\
        Apply the Identity gate to the remaining qubits in row-$j$\\
        Let the qubits in row-$j$ be labelled $j_1, \dots, j_n$ from left to right\\
        \For{$i=M$ \KwTo $i=1$}{
            Apply a \Gate{Swap} gate between $(j-1)_i$ and $j_i$
        }
    }
    \caption{Spatially Sparse Circuit}\label{alg:sp_sp_circuit}
\end{algorithm}\DecMargin{1em}

\begin{restatable}{corollary}{corgenericStoqMAtospsp}
    \label{cor:generic-StoqMAtospsp}
    Given any $\cl{StoqMA}(\alpha,\beta)$ circuit with $T$ gates on $M$ qubits, there exists a spatially sparse $\cl{StoqMA}(\alpha,\beta)$ circuit with $\Theta(T\cdot M) + \Theta(T\cdot M^2)$ gates on $T + \Theta(T\cdot M)$ qubits.
\end{restatable}

Similar to the previous mapping, the circuit statistics are preserved under this modification (\cref{app:statistics}).

\begin{remark}
    Let $n$ be the number of input $\ket{x}$ bits, then the input register size is $M = m(n)$, and the number of gates is $T = t(n)$, for some polynomials $m$ and $t$.
    The total number of gates in the spatially sparse circuit is $S = O(1)\cdot t(n) \cdot m(n)^2 = s(n)$ and the total number of qubits in the spatially sparse circuit is $P = O(1)\cdot m(n)\cdot t(n) = p(n)$, where $s(n)$ and $p(n)$ are polynomials.
\end{remark}

Note that in the subsequent Feynman-Kitaev construction, for each \Gate{Swap} gate, there would be an additional three clock qubits to mediate the equivalent \Gate{Cnot} decomposition; since this does not change the complexity or the proof for that matter, we will ignore this detail.
We can therefore combine the circuit-to-Hamiltonian construction for the $6$-local stoquastic Hamiltonian from \refcite{BDOT06} (\cref{thrm:ma-hard_6-local_stoq_ham}) with the spatially sparse construction of \refcite{OT08} to prove the following theorem.

\begin{theorem}
    \label{thrm:ma-hard-6-local-stoq-ham-sp-sp-graph}
    The \sc{$6$-Local Stoquastic Hamiltonian} problem on a spatially sparse graph is \clw{MA}{hard}.
\end{theorem}

\begin{proof}
    We employ the Feynman-Kitaev circuit-to-Hamiltonian construction to prove the problem is hard for the class \cl{MA}.
    Let $F_{|x|}$ be Arthur's semi-classical verification circuit.
    Recall \cl{MA}\textsubscript{q} = \cl{MA}.
    Let the input to the circuit be an $N \coloneqq n + w+ m + p$ qubit register comprised of four parts: the input state $\ket{x}$ of $n$ qubits, the proof state $\ket{\xi}$ of $w$ qubits, the \textit{ancilla} register of $m$ qubits initialised to $\ket{0}$ and the \textit{coin} register of $p$ qubits initialised to $\ket{+}$.
    Let $F_{|x|}$ comprise a sequence of $T$ nearest-neighbour Toffoli, and \textsc{Swap} gates denoted as $R_T, \dots, R_1$.

    Define a Hamiltonian $H = H_{\text{in}} + H_{\text{out}} + H_{\text{prop}} + H_{\text{clock}}$ acting on a register comprised of $T$ rows of $N$ qubits and $\hat{T}=(2T-1)N$ clock qubits labelled $c_1, \dots, c_{\hat{T}}$.
    There is one clock qubit for each operation in the gate sequence.
    Let $\tilde{F}_{|x|}$ represent a modified version of $F_{|x|}$ according to \cref{eq:spatially_sparse_gate_sequence}.
    $\tilde{F}_{|x|}$ is denoted as $\tilde{R}_{\hat{T}},\dots, \tilde{R}_1$.
    Let the output measured qubit be denoted $q$ where $q=TN$, i.e., the rightmost qubit on the final row.
    Arthur can only measure in the $Z$-basis.
    A given qubit, $l$, is acted on by circuit gates in two intervals: \begin{inparaenum}[(i)] \item By $R_j$ or the Identity gate, \item by the \Gate{Swap} gate.
    \end{inparaenum}Let $Q_x$ be the set of qubits that contain $\ket{x}$.
    Separate the first row of qubits into three columns respective of the input to the circuit.
    Let the column where the $\ket{+}$-\textit{ancilla} lie all be initialised to $\ket{+}$, denote this set of $Tp$ qubits as $Q_+$.
    Let the column where the $\ket{0}$-\textit{ancilla} lie all be initialised to $\ket{0}$ and all other qubits in rows $>1$ for the \textit{proof} column be also initialised to $\ket{0}$; this is a set of $Tm + (T-1)(n+w)$ qubits denoted as $Q_0$.
    Note that $\abs{Q_x \cup Q_+ \cup Q_0} = n + Tp + Tm + (T-1)(n+w) = TN - w$.
    
    Each Hamiltonian term is defined as a penalising Hamiltonian and must be stoquastic.
    \begin{align*}
        H_{\text{in}} &= \left(\sum_{j\in Q_x} \ketbra{\bar{x}_j}_j
        + \sum_{j \in Q_0} \ketbra{1}_{\textit{anc}, j} 
        + \sum_{j \in Q_+} \ketbra{-}_{\textit{coin}, j}\right)\otimes \ketbra{100}_{c_{t_{j}-1},c_{t_j},c_{t_{j}+1}}, \notag\\
        H_{\text{out}} &= \ketbra{0}_q \otimes \ketbra{1}_{c_{\hat{T}}},\notag\\
        H_{\text{clock}} &= \sum_{t=1}^{\hat{T}-1} \ketbra{01}_{c_{t},c_{t+1}},\notag\\
        H_{\text{prop}} &= \sum_{t=1}^{\hat{T}} H_{\text{prop}}(t).
    \end{align*}

    The Hamiltonian terms $H_{\text{out}}$ and $H_{\text{clock}}$ are left unchanged from \cref{eq:6LH_hard_proof_H-out} and \cref{eq:6LH_hard_proof_H-clock} respectively.
    The term $H_{\text{in}}$ now involves extra clock qubit checks.
    Following the arguments of \refcite{OT08}, the role of $H_{\text{in}}$ is to make sure that the state of the input qubits are appropriately set before the gates act on the qubits.
    The form of the propagation Hamiltonian terms are also unchanged:
    \begin{align*}
        H_{\text{prop}}(1) &= \ketbra{00}_{c_1,c_2} + \ketbra{10}_{c_1,c_2} - \tilde{R}_1\otimes(\ketbra{10}{00}_{c_1,c_2} + \ketbra{00}{10}_{c_1,c_2}), \\
        H_{\text{prop}}(t) &= \ketbra{100}_{c_{t-1},c_t,c_{t+1}} + \ketbra{110}_{c_{t-1},c_t,c_{t+1}} \notag\\
                            &\qquad - \tilde{R}_t\otimes(\ketbra{110}{100}_{c_{t-1},c_t,c_{t+1}} + \ketbra{100}{110}_{c_{t-1},c_t,c_{t+1}}), \quad 1<t<\hat{T}\\
        H_{\text{prop}}(\hat{T}) &= \ketbra{10}_{c_{\hat{T}-1},c_{\hat{T}}} + \ketbra{11}_{c_{\hat{T}-1},c_{\hat{T}}} - \tilde{R}_{\hat{T}}\otimes(\ketbra{11}{10}_{c_{\hat{T}-1},c_{\hat{T}}} + \ketbra{10}{11}_{c_{\hat{T}-1},c_{\hat{T}}}).
    \end{align*}

    Finally, the spatially sparse interaction graph occurs from the above construction and the format of \cref{fig:swap_sequence}.
    The snake-like time arrow over the qubits in the rows represents a string of clock qubits following the gate sequence seen in \cref{eq:spatially_sparse_gate_sequence}.
    Each Hamiltonian term above only acts in a local neighbourhood about each qubit.
    Moreover, each qubit only interacts with a set of qubits in its neighbourhood.
    Therefore, the interaction graph is spatially sparse.

    Additionally, we know each Hamiltonian term is stoquastic.
    Only $H_{\text{in}}$ and $H_{\text{prop}}(t)$ (for $\tilde{R}_t = \Gate{Swap}$) need to be proven to be stoquastic since the remaining terms are known to be stoquastic from \cref{thrm:ma-hard_6-local_stoq_ham}.
    The clock terms of $H_{\text{in}}$ are diagonal hence $H_{\text{in}}$ is stoquastic.
    The terms $\tilde{R}_t\otimes(\dots)$ in $H_{\text{prop}}(t)$ will have off-diagonal elements that are strictly positive.
    Therefore, each $H_{\text{prop}}(t)$ term is stoquastic even if $\tilde{R}_t = \Gate{Swap}$.
    To conclude, in the \textsc{yes} case, if Arthur's circuit accepts with probability at least $1-\epsilon$ then there exists a proof state such that the Hamiltonian $H$ has eigenvalues at most $\epsilon/(\hat{T}+1)$ and in the \textsc{no} case, having Arthur reject with probability at most $\epsilon$, all eigenvalues are at least $c(1-\epsilon -\sqrt{\epsilon})/\hat{T}^3$ for some constant $c$~\cite[Lemma 1]{OT08}.
\end{proof}

\subsection{StoqMA-hardness of the Local Stoquastic Hamiltonian Problem}
We will now revisit the \clw{StoqMA}{hardness} proof for $6$-local stoquastic Hamiltonians as presented in \refcite{BBT06}.
This proof employs a modified version of the Feynman-Kitaev circuit-to-Hamiltonian construction.
The approach involves defining a $6$-local stoquastic penalising Hamiltonian, similar to the method used in \cref{thrm:ma-hard_6-local_stoq_ham}, and then adding a measurement term perturbatively to ensure the eigenvalue promises are met.
Note that we do not cover the proof that the \sc{$k$-Local Stoquastic Hamiltonian} problem is in \cl{StoqMA}; for further details on this, refer to \refcite{BBT06}.
From this point forward, we assume this fact is established.

\begin{theorem}[\cite{BBT06}]\label{thrm:stoqma-hard-6-local-stoq-ham}
    The \sc{$6$-Local Stoquastic Hamiltonian} problem is \clw{StoqMA}{complete}.
\end{theorem}

\begin{proof}
    Let $S_{|x|}$ be Arthur's stoquastic verification circuit.
    Let the input to the circuit be an $M= n+w+m+p$ qubit register comprised of four parts: the input state $\ket{x}$ of $n$ qubits, the proof state $\ket{\xi}$ of $w$ qubits, the \textit{ancilla} register of $m$ qubits initialised to $\ket{0}$ and the \textit{coin} register of $p$ qubits initialised to $\ket{+}$.
    Let $S_{|x|}$ be comprised of a sequence of $T$ gates, $R_T,\dots, R_1$, from the set $\{X, \Gate{Cnot}, \Gate{Toffoli}\}$.
    Without loss of generality, we can assume all gates in the circuit are Toffoli gates.

    Define a Hamiltonian $H = H_{\text{in}} + H_{\text{prop}} + H_{\text{clock}}$ acting on a register of $T$ clock qubits labelled as $c_1, \dots, c_{T}$ and the $M$ qubit register.
    Let the output measured qubit be denoted $q$; for this instance, Arthur can measure using only the $X$-basis.
    Each Hamiltonian term is defined as a penalising Hamiltonian and must be stoquastic.
    Define the following history state
    \begin{equation*}
        \ket{\eta(x,\xi)} = \frac{1}{\sqrt{T+1}}\sum_{t=0}^{T} R_t\dots R_0 \ket{x,\xi,0^m,+^p}\ket{1^{t}0^{T-t}}.
    \end{equation*}
    We therefore have,
    \begin{align*}
        H_{\text{in}} &= \left(\sum_{i=1}^{n} \ketbra{\bar{x}_i}_i + \sum_{i=1}^{m} \ketbra{1}_i + \sum_{i=1}^{p} \ketbra{-}_i\right)\otimes \ketbra{0}_{c_1}, \\
        H_{\text{clock}} &= \sum_{t=1}^{T-1} \ketbra{01}_{c_{t},c_{t+1}}, \\
        H_{\text{prop}} &= \sum_{t=1}^{T} H_{\text{prop}}(t).
    \end{align*}
    We define the propagation Hamiltonian terms in the following way:
    \begin{align*}
        H_{\text{prop}}(1) &= \ketbra{00}_{c_1,c_2} + \ketbra{10}_{c_1,c_2} - R_1\otimes(\ketbra{10}{00}_{c_1,c_2} + \ketbra{00}{10}_{c_1,c_2}), \\
        H_{\text{prop}}(t) &= \ketbra{100}_{c_{t-1},c_t,c_{t+1}} + \ketbra{110}_{c_{t-1},c_t,c_{t+1}} \notag\\
                            &\qquad - R_t\otimes(\ketbra{110}{100}_{c_{t-1},c_t,c_{t+1}} + \ketbra{100}{110}_{c_{t-1},c_t,c_{t+1}}), \quad 1<t<T \\
        H_{\text{prop}}(T) &= \ketbra{10}_{c_{T-1},c_T} + \ketbra{11}_{c_{T-1},c_T} - R_T\otimes(\ketbra{11}{10}_{c_{T-1},c_T} + \ketbra{10}{11}_{c_{T-1},c_T}).
    \end{align*}
    The history state defined above is a zero-energy eigenstate of $H_{\text{in}} + H_{\text{clock}} + H_{\text{prop}}$.
    It is known that the spectral gap of this Hamiltonian is $\Delta = \Omega(T^{-3})$~\cite[Lemma 5]{BBT06}.
    Furthermore, it is trivial to check that the above Hamiltonian terms are stoquastic and that $H$ is $6$-local.
        
    The next term to define is the output Hamiltonian, 
    \begin{equation*}
        H_{\text{out}} = \ketbra{-}_q \otimes \ketbra{1}_{c_{T}}.
    \end{equation*}
    Clearly $H_{\text{out}}$ is stoquastic.
    We define a new Hamiltonian $H' = H + \delta H_{\text{out}}$ where $0 < \delta \ll \Delta$.
    We treat $\delta H_{\text{out}}$ as a perturbation term.
    This gives the eigenvalue as 
    \begin{equation*}
        \lambda(H') = \delta \min_{\ket{\xi}} \bra{\xi} H_{\text{out}} \ket{\xi} + O(\delta^2).
    \end{equation*}
    For a sufficiently small $\delta$ the second-order terms can be ignored.
    To conclude, we leverage arguments analogous to \refcite{KSV02} to show in the \textsc{yes} case, there exists a proof state such that the Hamiltonian $H$ has eigenvalues at most $\delta(1-\alpha)/(T+1)$.
    In the \textsc{no} case, all eigenvalues are at least $\delta(1-\beta)/(T+1)$.
    The perturbation analogous applies for $\delta \ll \Delta$ hence let $\delta \ll T^{-3}$.
\end{proof}

As we shall discuss in the next section, the degree of locality has no influence on the complexity for $k\geq 2$~\cite[Theorem 8]{BDOT06}.
We will now show that a simple extension to the above proof shows that $6$-local stoquastic Hamiltonians are \clw{StoqMA}{complete} on a spatially sparse graph.
This proof is similar to the proof of \cref{thrm:ma-hard-6-local-stoq-ham-sp-sp-graph}.

\thrmsixLSHPspsp

\begin{proof}
    Let $S_{|x|}$ be Arthur's stoquastic verification circuit.
    Let the input to the circuit be an $M= n+w+m+p$ qubit register comprised of four parts: the input state $\ket{x}$ of $n$ qubits, the proof state $\ket{\xi}$ of $w$ qubits, the \textit{ancilla} register of $m$ qubits initialised to $\ket{0}$ and the \textit{coin} register of $p$ qubits initialised to $\ket{+}$.
    Let $S_{|x|}$ be comprised of a sequence of $T$ nearest-neighbour gates, $R_T,\dots, R_1$, from the set $\{X, \Gate{Cnot}, \Gate{Toffoli}\}$.
    The gates are either Toffoli gates or \Gate{Swap} gates (three \Gate{Cnot} gates).

    Define a Hamiltonian $H = H_{\text{in}} + H_{\text{prop}} + H_{\text{clock}}$ acting on a register comprised of $T$ rows of $M$ qubits and $\hat{T}=(2T-1)M + 2$ clock qubits labelled $c_0, \dots, c_{\hat{T}+1}$.
    Note that the $0$-th and $(\hat{T}+1)$-th clock qubits are always set to $1$ and $0$, respectively.
    There is one clock qubit for each operation in the gate sequence.
    Let $\hat{S}_{|x|}$ represent a modified version of $S_{|x|}$ according to \cref{eq:spatially_sparse_gate_sequence}.
    $\hat{S}_{|x|}$ is denoted as $\hat{R}_{\hat{T}} \dots \hat{R}_1$.
    Let the output measured qubit be denoted $q$ where $q=TM$, i.e., the rightmost qubit on the final row.
    Arthur can only measure in the $X$-basis.
    A given qubit, $l$ is acted on by circuit gates in two intervals: \begin{inparaenum}[(i)] \item By $\hat{R}_j$ or the Identity gate, \item by the \Gate{Swap} gate.
    \end{inparaenum} Let $Q_x$ be the set of qubits that contain $\ket{x}$.
    Separate the first row of qubits into four columns respective of the input to the circuit.
    Let the column where the $\ket{+}$-\textit{ancilla} lie all be initialised to $\ket{+}$, denote this set of $Tp$ qubits as $Q_+$.
    Let the column where the $\ket{0}$-\textit{ancilla} lie all be initialised to $\ket{0}$ and all other qubits in rows $>1$ for the \textit{proof} column be also initialised to $\ket{0}$; this is a set of $Tm + (T-1)(n+w)$ qubits denoted as $Q_0$.
    Note that $\abs{Q_x \cup Q_+ \cup Q_0} = n + Tp + Tm + (T-1)(n+w) = TM - w$.
    
    Each Hamiltonian term is defined as a penalising Hamiltonian and must be stoquastic.
    Define the following history state
    \begin{equation*}
        \ket{\eta(x,\xi)} = \frac{1}{\sqrt{\hat{T}+1}}\sum_{t=0}^{\hat{T}} \hat{R}_t\dots \hat{R}_0 \ket{x,\xi,0^m,+^p}\ket{0^{(T-1)(n+w)},+^{(T-1)p}}\ket{1^{t}0^{\hat{T}-t}}.
    \end{equation*}
    We also have,
    \begin{align*}
        H_{\text{in}} &= \left(\sum_{j\in Q_x} \ketbra{\bar{x}_j}_j
        + \sum_{j \in Q_0} \ketbra{1}_{\textit{anc}, j} 
        + \sum_{j \in Q_+} \ketbra{-}_{\textit{coin}, j}\right)\otimes \ketbra{100}_{c_{t_{j}-1},c_{t_j},c_{t_{j}+1}}, \notag\\
        H_{\text{clock}} &= \sum_{t=1}^{\hat{T}-1} \ketbra{01}_{c_{t},c_{t-1}},  \notag\\
        H_{\text{prop}} &= \sum_{t=1}^{\hat{T}} H_{\text{prop}}(t).\notag
    \end{align*}

    The term $H_{\text{in}}$ now involves extra clock qubit checks.
    Following the arguments of \refcite{OT08}, the role of $H_{\text{in}}$ is to make sure that the state of the input qubits are appropriately set before the gates act on the qubits.
    By demanding that all additional qubits be initialised to either $\ket{0}$ or $\ket{+}$ (depending on their column), we can see that $H_{\text{in}}$ acts as expected.
    A short calculation shows that $H_{\text{in}}\ket{\eta} = 0$~\cite{OT08},
    \begin{align*}
        H_{\text{in}}^{(x)}\ket{\eta} &\propto \sum_{j\in Q_x} \ketbra{\bar{x}_j}_j\ket{x,\xi,0^m,+^p}\ket{0^{(T-1)(n+w)},+^{(T-1)p}} \ket{100\dots 0} = 0,\\
        H_{\text{in}}^{(0)}\ket{\eta} &\propto \sum_{j\in Q_0} \ketbra{1}_{\textit{anc},j} \ket{x,\xi,0^m,+^p}\ket{0^{(T-1)(n+w)},+^{(T-1)p}} \ket{100\dots 0} = 0,\\
        H_{\text{in}}^{(+)}\ket{\eta} &\propto \sum_{j\in Q_+} \ketbra{-}_{\textit{coin},j} \ket{x,\xi,0^m,+^p}\ket{0^{(T-1)(n+w)},+^{(T-1)p}} \ket{100\dots 0} = 0.
    \end{align*}
    Let $\ket{\xi_t}  = \hat{R}_t \ket{\xi_{t-1}}$ where $\ket{\xi_0} = \ket{x,\xi,0^m,+^p}\ket{0^{(T-1)(m+w)},+^{(T-1)p}}$.
    The form of the propagation Hamiltonian terms are also unchanged; hence
    \begin{align*}
        H_{\text{prop}}(1) &= \ketbra{00}_{c_1,c_2} + \ketbra{10}_{c_1,c_2} - \hat{R}_1\otimes(\ketbra{10}{00}_{c_1,c_2} + \ketbra{00}{10}_{c_1,c_2}), \\
        H_{\text{prop}}(t) &= \ketbra{100}_{c_{t-1},c_t,c_{t+1}} + \ketbra{110}_{c_{t-1},c_t,c_{t+1}} \notag\\
                            &\qquad - \hat{R}_t\otimes(\ketbra{110}{100}_{c_{t-1},c_t,c_{t+1}} + \ketbra{100}{110}_{c_{t-1},c_t,c_{t+1}}), \quad 1<t<\hat{T} \\
        H_{\text{prop}}(\hat{T}) &= \ketbra{10}_{c_{\hat{T}-1},c_{\hat{T}}} + \ketbra{11}_{c_{\hat{T}-1},c_{\hat{T}}} - \hat{R}_S\otimes(\ketbra{11}{10}_{c_{\hat{T}-1},c_{\hat{T}}} + \ketbra{10}{11}_{c_{\hat{T}-1},c_{\hat{T}}}).
    \end{align*}
    We conclude that the spectral of $H \coloneqq H_{\text{in}} + H_{\text{clock}} + H_{\text{prop}}$ is $\Delta = \Omega(\hat{T}^{-3})$~\cite[Lemma 5]{BBT06}.
    Furthermore, it is trivial to check that the above Hamiltonian terms are stoquastic and that $H$ is $6$-local.

    Define the output Hamiltonian,
    \begin{equation*}
        H_{\text{out}} = \ketbra{-}_q \otimes \ketbra{1}_{c_{\hat{T}}}.
    \end{equation*}
    Clearly $H_{\text{out}}$ is stoquastic.
    We define a new Hamiltonian $H' = H + \delta H_{\text{out}}$ where $0 < \delta \ll \Delta$.
    We treat $\delta H_{\text{out}}$ as a perturbation term.
    This gives the eigenvalue as
    \begin{equation*}
        \lambda(H') = \delta \min_{\ket{\xi}} \bra{\xi} H_{\text{out}} \ket{\xi} + O(\delta^2).
    \end{equation*}
    For a sufficiently small $\delta$ the second-order terms can be ignored.
    To conclude, we leverage arguments analogous to \refcite{KSV02} (cf. \cref{thrm:stoqma-hard-6-local-stoq-ham}) to show in the \textsc{yes} case, there exists a proof state such that the Hamiltonian $H$ has eigenvalues at most $\delta(1-\alpha)/(\hat{T}+1)$.
    In the \textsc{no} case, all eigenvalues are at least $\delta(1-\beta)/(\hat{T}+1)$.
    The perturbation analogous applies for $\delta \ll \Delta$ hence let $\delta \ll \hat{T}^{-3}$.
\end{proof}

To summarise, we have shown that simple modifications to the Feynman-Kitaev clock construction is sufficient to prove the \clw{StoqMA}{hardness} of the \sc{$6$-Local Stoquastic Hamiltonian} problem on a spatially sparse graph.
We have restated the original proofs for clarity and to set the notation for the new proofs that follow.

\section{Stoquastic Perturbation Gadgets}\label{sec:stoq-pert-gadgets}
In this section, we recap the work of Bravyi \emph{et al}.~\cite{BDOT06} on perturbation gadgets for stoquastic Hamiltonians.
Similar to the original work by Kempe, Kitaev, and Regev~\cite{KKR06}, we demonstrate how $k$-local stoquastic Hamiltonians can be reduced to $3$-local Hamiltonians.
Subsequently, we apply a gadget technique to further reduce these $3$-local stoquastic Hamiltonians to $2$-local stoquastic ones, thereby establishing that the \sc{$2$-Local Stoquastic Hamiltonian} problem is \clw{StoqMA}{complete}.

\refcite{BDOT06} use the self-energy method to prove the perturbative effects, a technique originally developed in \refcite{KKR06} for the \sc{Local Hamiltonian} problem.
In contrast, we employ the Schrieffer-Wolff transformation, formally described for many-body Hamiltonians in \refcite{BDL11} and applied to stoquastic Hamiltonians in \refcite{BDLT06}.
Both methods are equivalent in the sense that they utilise perturbation theory to derive an effective Hamiltonian whose low-energy spectrum closely approximates that of the original Hamiltonian.
Readers familiar with these ideas can skip to \cref{sec:stoq-geo-gadgets}.

\subsection{Crash Course in Perturbation Gadgets}
The idea of perturbation gadgets is to introduce a mediator qubit in a system with the effect of simulating the low-energy spectrum of a target Hamiltonian via local interactions between the mediator qubit and the system qubits.
An intuitive example is the simulation of ferromagnetic spin interactions using antiferromagnetic ones~\cite{TheHamiltonianJungle}.
The two main methods used in the literature to study this effect are the self-energy method and the Schrieffer-Wolff transformation.
The self-energy method is a perturbative method that uses the Dyson equation~\cite{FWK03} to calculate the effective Hamiltonian.
The Schrieffer-Wolff transformation uses a unitary transformation to decouple the low-energy subspace from the high-energy subspace, yielding an effective Hamiltonian on the low-energy sector.

The usual format of reduction proofs is to start from a known problem that is complete for some complexity class.
However, when dealing with perturbation gadgets, it is common practice to use the notion of simulability instead.
Informally speaking, saying one Hamiltonian can `\emph{simulate}' another Hamiltonian gives the same conclusions as an appropriate reduction.
Moreover, a reduction from problem $A$ to problem $B$ means that $B$ is at least as hard as $A$.
Analogously, if Hamiltonian $H_{{\rm sim}}$ can simulate Hamiltonian $H_{{\rm targ.}}$, then $H_{{\rm sim}}$ is at least as hard as $H_{{\rm targ.}}$, cf.\cref{rmk:simulation_means_reduction}.
In this regard, we interchange the word `simulate' with `reduction'.
We commonly refer to $H_{{\rm targ.}}$ as the \emph{target} Hamiltonian, which acts on a $2^n$-dimensional Hilbert space $\mathcal{H} = \mathcal{L}_- \oplus \mathcal{L}_+$; $\mathcal{L}_-$ refers to the low-energy eigenspace and $\mathcal{L}_+$ the high-energy eigenspace.
The reduction aims to show that a Hamiltonian $H_{{\rm sim}}$, acting on a larger Hilbert space $\widetilde{\mathcal{H}} = \widetilde{\mathcal{L}}_- \oplus \widetilde{\mathcal{L}}_+$, can be constructed to have a low-energy subspace that approximates that of the target Hamiltonian.
To formalise this, we say there exists an isometry $\widetilde{\mathcal{E}}: \mathcal{H} \rightarrow \mathcal{H}_{{\rm sim}}$ such that $\text{Im}(\widetilde{\mathcal{E}}) \coloneqq \mathcal{L}_-(H_{{\rm targ.}})$ and $\norm{\widetilde{\mathcal{E}}^\dagger H_{{\rm sim}} \widetilde{\mathcal{E}}} \approx \norm{H_{{\rm targ.}}}$.

\begin{definition}[\cite{PM17}]\label{def:Hamiltonian_simulation}
    Let $H$ be a Hamiltonian acting on a $2^n$-dimensional Hilbert space $\mathcal{H} = \mathcal{L}_- \oplus \mathcal{L}_+$.
    Let $H_{{\rm sim}}$ be a Hamiltonian acting on a $2^m$-dimensional Hilbert space, with $m>n$ and where $\widetilde{\mathcal{H}} = \widetilde{\mathcal{L}}_- \oplus \widetilde{\mathcal{L}}_+$.
    Let $\mathcal{E}: \mathcal{H}\rightarrow \widetilde {\mathcal{H}}$ be an isometry.
    We say that $H_{{\rm sim}}$ is an effective Hamiltonian, or a $(\eta, \epsilon)$-simulator, for $H$ if these exists an isometry $\widetilde{\mathcal{E}}: \mathcal{H} \rightarrow \mathcal{H}_{{\rm sim}}$ such that
    \begin{enumerate}[(i)]
        \item $\text{Im}(\widetilde {\mathcal{E}}) \coloneqq \mathcal{L}_-(H)$.
    \label{item:sim_1}
        \item $\norm{H - \widetilde {\mathcal{E}}^\dagger H_{{\rm sim}} \widetilde {\mathcal{E}}} \leq \epsilon$.
        \item $\norm{\mathcal{E} - \widetilde {\mathcal{E}}} \leq \eta$.
    \end{enumerate}
\end{definition}

We note that it is not strictly necessary for the Hilbert space in \cref{def:Hamiltonian_simulation} to admit the format $\mathcal{L}_- \oplus \mathcal{L}_+$; however, this is appropriate and convenient for our purposes.
Important lemmas can be gleaned from \cref{def:Hamiltonian_simulation} essential for the reduction proofs.
We will only need the following lemma from \refcite{BH16} for this work.

\begin{lemma}[Eigenvalue Simulation~\cite{BH16}]\label{lma:eigenvalue_simulation}
    Let $(H_{{\rm sim}},\mathcal{E})$ be an $(\eta, \epsilon)$-simulator for $H$.
    Let $\lambda_j(H)$ denote the $j$-th smallest eigenvalue of $H$.
    Then
    \begin{equation*}
        \abs{\lambda_j(H) - \lambda_j(H_{{\rm sim}})} \leq \epsilon.
    \end{equation*}
\end{lemma}

We do not attempt to calculate precise error bounds, and so we assume $\eta, \epsilon = O(1/\poly{n})$ (\cref{lma:second_order_reduction} is essentially an existence lemma).

To set up the Schrieffer-Wolff transformation, we consider a simulator Hamiltonian of the form $H_{{\rm sim}} = \Delta H_0 + V$, where $H_0$ is the unperturbed part and $V$ is a perturbation term.
The unperturbed Hamiltonian $H_0$ induces a \emph{split} Hilbert space $\mathcal{H} = \mathcal{L}_- \oplus \mathcal{L}_+$: there is a spectral gap $\Delta \gg 1$ between the two subspaces, with $(H_0)_- = 0$ and $\lambda((H_0)_{++}) \geq 1$.
The perturbation satisfies $\norm{V} < \Delta/2$ to prevent mixing between the subspaces.
We define projectors $\Pi_-$ and $\Pi_+$ onto the low- and high-energy subspaces respectively, and write $O_{\pm\mp} = \Pi_\pm O \Pi_\mp$ and $O_{\pm\pm} = \Pi_\pm O \Pi_\pm = O_\pm$ for any operator $O$.

The Schrieffer-Wolff transformation is then a unitary $e^S$, where $S$ is anti-Hermitian, chosen so that the transformed Hamiltonian $e^{-S} H_{{\rm sim}} e^S$ is block-diagonal with respect to $\Pi_\pm$.
The effective Hamiltonian is $H_{{\rm eff.}} = (e^{-S} H_{{\rm sim}} e^S)_{-}$, approximated via a truncated series.
For the purposes of this work, we only need to go to second order.
There are lemmas that completely specify the form of the effective Hamiltonian for first-, second- and third-order terms~\cite{BH16}.
A simplified version of the second-order lemma is as follows.

\begin{lemma}[Second-order Reduction~\cite{BH16}]\label{lma:second_order_reduction}
    Let $H_{{\rm sim}} = \Delta H_0 + \sqrt{\Delta}\;V_{{\rm main}} + V_{{\rm extra}}$ be chosen such that $\lambda((H_0)_{++})\geq 1$, $(H_0)_{-}, (H_0)_{-+} = 0$, $(V_{{\rm extra}})_{-+}, (V_{{\rm main}})_{-} = 0$ and
    \begin{equation*}
        \norm{\bar{H}_{{\rm targ.}} - (V_{{\rm extra}})_{-} + \Delta^{-1}(V_{{\rm main}})_{-+}(H_0^{-1})_{++}(V_{{\rm main}})_{+-}} \leq \epsilon/2.
    \end{equation*}
    For appropriate choices of $\norm{V_{{\rm main}}}, \norm{V_{{\rm extra}}}$ and $\Delta$, $H_{{\rm sim}}$ is an $(\eta, \epsilon)$-simulator for $H_{{\rm targ.}}$.
\end{lemma}

The term $\bar{H}_{{\rm targ.}} = \mathcal{E}^\dagger H_{{\rm targ.}}\mathcal{E}$ is the logical encoding of the target Hamiltonian.
Using \cref{lma:second_order_reduction}, the general recipe for the effective Hamiltonian up to second order is
\begin{equation}\label{eq:Heff_form}
    H_{{\rm eff.}} = V_{-} - \Delta^{-1}V_{-+}(H_0^{-1})_{++}V_{+-}.
\end{equation}

The perturbation techniques and reductions that follow are designed with the goal of decreasing the degree of locality for a general $k$-local stoquastic Hamiltonian.
The difficulty here, compared to previous work~\refcite{KKR06}, is that each term and resulting Hamiltonian must also be stoquastic, requiring careful consideration throughout.
The goal is to reduce $O(1)$-local stoquastic Hamiltonians to $2$-local stoquastic Hamiltonians.
The first step uses the \emph{subdivision gadget}, which takes a $O(1)$-local stoquastic Hamiltonian to a $3$-local one.
To reach $2$-local terms, there is an intermediate reduction to \emph{special} $3$-local stoquastic Hamiltonians~\cite{BDOT06}, after which the $3$-to-$2$-local reduction proceeds analogously to \refcite{KKR06}.
The same results for $6$-local stoquastic Hamiltonians then apply to the $2$-local case, since the perturbative reductions are polynomial-time transformations that preserve polynomial norms.

\begin{remark}\label{rmk:simulation_means_reduction}
    Let $H_B$ be a $(\eta,\epsilon)$ simulator for $H_A$.
    Then $H_B$ is at least as hard as $H_A$.
\end{remark}

\begin{proof}
    Take an instance of $H_A$ to be defined as $x\coloneqq(H_A, a,b)$ such that $b-a \geq 1/\poly{n}$.
    Since $H_B$ is a $(\eta,\epsilon)$ simulator for $H_A$, we set the parameters $b' = b - \epsilon$ and $a' = a + \epsilon$.
    Setting $\epsilon < (b-a)/2$ ensures that $b'-a' = b-a - 2\epsilon \geq 1/\poly{n}$, meeting the criterion for a valid instance of $H_B$.
    It is straightforward to verify that $\lambda(H_B) \leq a'$ if $\lambda(H_A) \leq a$ (\textsc{yes} case), and the converse holds for the \textsc{no} case.
\end{proof}

As a final note, perturbation gadgets can be applied in \emph{parallel} and in \emph{series}.
In the parallel case, the gadgets are applied to each term of the Hamiltonian simultaneously; in the series case, they are applied sequentially.
\cref{app:parallel} and \cref{app:composition} provide detailed treatments of these two cases: the former shows that no gadget cross-terms arise at second order, and the latter establishes how the composition law influences the approximation error.

In what follows we typically denote the target Hamiltonian as $H$ and the simulator Hamiltonian as $\widetilde{H}$, with the understanding that $H$ is the original Hamiltonian we wish to simulate and $\widetilde{H}$ is the new Hamiltonian constructed.

\subsection{Stoquastic Subdivision Gadget}
A general $k$-local stoquastic Hamiltonian can be decomposed as a sum of a product of 1-qubit matrices.
Define the $1$-qubit matrices, denoted $\rho^\mu$, as
\begin{equation*}
    \rho^0 = \ketbra{0}, \quad \rho^1 = \ketbra{0}{1}, \quad \rho^2 = \ketbra{1}{0}, \quad \rho^3 = \ketbra{1}.
\end{equation*}
A general local Hamiltonian term, $H_j$ has support on a set of at most $k$ qubits; denote this set $Q_j = \{j_1,\dots, j_k\}$.
Let $\Omega \coloneqq \{Q_1, \dots, Q_m\}$ denote the $m$ such sets defining the full Hamiltonian.
The stoquastic $k$-local Hamiltonian term $H_j = H_j(Q_j)$ is then,
\begin{equation*}
    H_j(Q_j) = \sum_{\mu_1, \dots, \mu_k} h_{j_1,\dots j_k}^{\mu_1, \dots, \mu_k} \rho_{j_1}^{\mu_1} \otimes \dots \otimes \rho_{j_k}^{\mu_k}.
\end{equation*}
For example, a $2$-local term, $H = Z_1Z_2$ can be written as
\begin{equation*}
    Z_1Z_2 = h^{0_10_2}_{12} \rho^0_1\rho^0_2 + h^{0_11_2}_{12} \rho^0_1\rho^3_2 + h^{1_10_2}_{12} \rho^3_1\rho^0_2 + h^{1_11_2}_{12} \rho^3_1\rho^3_2.
\end{equation*}
Define $\boldsymbol{\mu} = (\mu_1, \dots, \mu_k)$ and $\boldsymbol{j} = (j_1, \dots, j_k)$.
If each $h_{\boldsymbol{j}}^{\boldsymbol{\mu}}$ term is non-negative then a general stoquastic Hamiltonian can be expressed as~\cite{BDOT06}
\begin{equation}\label{eq:general_stoq_rho}
    H = K - \sum_{Q \in \Omega} \sum_{\boldsymbol{\mu}} h_{\boldsymbol{j}}^{\boldsymbol{\mu}} \rho_{j_1}^{\mu_1} \dots\rho_{j_k}^{\mu_k}.
\end{equation}
The term $K$ represents a constant energy shift that renders the diagonal elements of $H$ non-positive.
Take each subset $Q_j$ and partition it into two disjoint subsets, $\sigma_j$ and $\tau_j$.
The partitioning is such that $\sigma_j \cup \tau_j = Q_j$, $\sigma_j \cap \tau_j = \emptyset$ and $\abs{\sigma_j}$ is close to $\abs{\tau_j}$, i.e., their sizes are near $\lceil k/2 \rceil$.
The Hamiltonian can be rewritten as
\begin{equation}\label{eq:stoq_ham_CD_form}
    H = K - \sum_{a=1}^M \left(C_a \otimes D_a + C_a^\dagger \otimes D_a^\dagger \right).
\end{equation}
The quantity $M = 4^k \binom{n}{k}$ is an upper bound on the number of possible terms for a given $k$-local interaction; for example, a dense matrix where all elements are non-zero admits a $\rho$-matrix decomposition of $4^k = 2^{2k}$ terms.
Furthermore, $\binom{n}{k}$ assumes all interactions are strictly $k$-local.
Both $C$ and $D$ have non-negative elements and act on disjoint sets of at most $\lceil k/2 \rceil$ qubits.
Informally, $C$ and $D$ are collections of a string of $\lceil k/2 \rceil$ $\rho$-matrices where $C$ acts on the partition $\sigma$ and $D$ on $\tau$.
For example, in a $6$-local interaction we could denote $C_1 = h^{000}_{123}\rho^0_1\rho^0_2\rho^0_3$ and $D_1 = h^{000}_{456}\rho^0_4\rho^0_5\rho^0_6$.
It is easy to see $C^\dagger C$ and $D^\dagger D$ are diagonal.

We can now introduce a perturbation gadget to simulate the low-energy spectrum of \cref{eq:stoq_ham_CD_form}.
Let the target Hamiltonian be \cref{eq:stoq_ham_CD_form} and define the perturbed Hamiltonian $\widetilde{H}=\Delta H_0+V$ with,
\begin{align}
    H_0 &= \sum_{a=1}^M \ketbra{1}_a, \notag\\
    V &= \sqrt{\Delta}\; V_{{\rm main}} + V_{{\rm extra}}, \notag\\
    V_{{\rm main}} &= -\sum_{a=1}^M \left(C_a + D_a^\dagger \right)\otimes S_a^+ + \left(C_a^\dagger + D_a \right)\otimes S_a^-, \label{eq:OG_subdivision_Vmain}\\
    V_{{\rm extra}} &= \sum_{a=1}^M \left(C_a^\dagger \otimes C_a + D_a \otimes D_a^\dagger \right).\notag
\end{align}
Let $S^+ = (S^-)^\dagger = \ketbra{1}{0}$\footnote{$S^\pm$ could be represented as the $\rho$-matrices $\rho^2/\rho^1$, but we leave them as $S^\pm$ for clarity.}.
Both $H_0$ and $V$ are local (term-wise) stoquastic Hamiltonians.

Using \cref{lma:second_order_reduction} and \cref{eq:Heff_form}, it can be shown the effective Hamiltonian simulates the low-energy spectrum of \cref{eq:stoq_ham_CD_form} (up to an overall constant) to second-order, cf. \cref{app:stoquastic_subdivision_gadget}.
With $O(\log k)$ applications of this subdivision gadget, a $k$-local stoquastic Hamiltonian can be reduced to a $3$-local one.
The repeated application of such gadgets leaves the resultant Hamiltonian with $3$-local and $2$-local terms.
We group $2$-local contributions into a term $\Gamma$.
All remaining $3$-local terms are subsets of triples, $B_j = \{j_1, j_2, j_3\}$ where $B_j \in \Omega_3$.
We can, therefore assume a general $3$-local stoquastic Hamiltonian is of the form
\begin{equation*}
    H = \Gamma - \sum_{B\in\Omega_3} \sum_{\nu_i \in \pm} h_{\boldsymbol{j}}^{\boldsymbol{\nu}} S_{j_1}^{\nu_1}S_{j_2}^{\nu_2}S_{j_3}^{\nu_3},
\end{equation*}
where $\Omega_3$ is the set of all triples of qubits in the system and $h_{\boldsymbol{j}}^{\boldsymbol{\nu}}$ are non-negative.

\subsubsection{Special 3-Local Stoquastic Hamiltonians}
In order to perform the $3$-to-$2$-local reduction of stoquastic Hamiltonians, it is convenient to reduce a general $3$-local stoquastic Hamiltonian to a special form,
\begin{equation*}
    H = \Gamma - \sum_{B\in\Omega_3} h_{\boldsymbol{j}} X_{j_1}X_{j_2}X_{j_3}.
\end{equation*}

A general $3$-local is reduced to the special $3$-local using perturbation theory.
The specific types of terms that occur for combinations of $S^+$ and $S^-$ are well approximated by $XXX$ interactions.
The structure of this reduction involves several steps, some of which are useful for the $3$-to-$2$-local reduction.
The reduction uses a fourth-order correction, so we will not cover the details here.
The interested reader is referred to \refcite{BDOT06} for the full details.
The general point is an intermediate step from general $3$-local stoquastic Hamiltonians to special $3$-local stoquastic Hamiltonians.
The special form is useful for the $3$-to-$2$-local reduction.

\subsection{3-Local to 2-Local Reduction}
As mentioned, the subdivision gadget can only reduce a $k$-local Hamiltonian to a $3$-local one.
We must employ subsequent ideas to reduce the locality beyond $3$.
Thankfully, the bulk of this work was originally conducted in \refcite{KKR06}.
In this section, we recap the ideas of \refcite{BDOT06} to show how this process works for stoquastic Hamiltonians.

Define non-negative operators $O$ as being proportional to $X$, then an interaction term of a special $3$-local stoquastic Hamiltonian can be expressed as 
\begin{equation}\label{eq:special_3LSH_O_form}
    H' = \Gamma - 6\;O_1O_2O_3.
\end{equation}
Let the target Hamiltonian be \cref{eq:special_3LSH_O_form} and define the perturbed Hamiltonian $\widetilde{H}=\Delta H_0+V$ with,
\begin{align*}
    H_0 &= -\frac{1}{4} \left(Z_1Z_2 + Z_2Z_3 + Z_1Z_3 - 3\; I \right), \\
    V &= \Delta^{2/3}\; V_{{\rm main}} + V_{{\rm extra}}, \\
    V_{{\rm main}} &= -\sum_{j=1}^3 O_j\otimes S^+_j + O_j^\dagger \otimes S^-_j, \\
    V_{{\rm extra}} &= \Gamma.
\end{align*}
Trivially, $\Pi_- = \ketbra{000} + \ketbra{111}$ and $\Pi_+ =  I - \Pi_-$.

Using the third-order reduction format of~\cite[Lemma 6]{BH16}, it can be shown that the effective Hamiltonian is 
\begin{equation*}
    H_{{\rm eff.}} = K + \Gamma\otimes  I_c - 6\;O_1O_2O_3\otimes X_c,
\end{equation*}
where $ I_c$ and $X_c$ act on the two-dimensional low-energy subspace of the mediator qubits.
With a small calculation, it can be shown that the effective Hamiltonian simulates the low-energy spectrum of the target Hamiltonian up to an overall energy shift.

All such gadgets allow for the conclusion that the complexity of the \sc{$k$-Local Stoquastic Hamiltonian} problem is preserved for $k\geq 2$.
By showing the $k$-local stoquastic Hamiltonian on a spatially sparse graph is \clw{StoqMA}{complete}, we can conclude the same is true for $2$-local case~\cite[Theorem 8]{BDOT06}.
The next section covers geometric gadgets that allow for the reduction of a general spatially sparse $2$-local stoquastic Hamiltonian to a $2$-local stoquastic Hamiltonian on a planar graph.

\section{Geometrical Stoquastic Perturbation Gadgets}\label{sec:stoq-geo-gadgets}
We now present a series of gadgets, inspired by \refcite{OT08}, specific to stoquastic Hamiltonians for reducing a local stoquastic Hamiltonian on a spatially sparse graph to one on a planar graph.
Our gadgets preserve the stoquasticity and $2$-locality of the Hamiltonian --- specifically, we constructed new Hamiltonians that are $2$-local term-wise stoquastic and act on qubit systems.
The main gadgets required for this reduction are:
\begin{enumerate}
    \item the \textsl{Subdivision} gadget,
    \item the \textsl{Cross} gadget,
    \item the \textsl{Fork} gadget,
    \item the \textsl{Triangle} gadget.
\end{enumerate}
Each gadget serves a specific purpose.
The \textsl{Fork} and \textsl{Triangle} gadgets are used to reduce the degree of vertices.
The \textsl{Cross} gadget's role is to planarise the interaction graph.
We have already seen a subdivision gadget.
However, we emphasise that a $2$-local stoquastic interaction can also be subdivided --- proving useful for the \textsl{Triangle} and other gadget identities.
Each gadgets' analysis is analogous to the generic subdivision gadget and so we will not cover a preliminary sketch of the original ideas from \refcite{OT08}.

In the general $2$-local case the Hamiltonian can be expressed as a summation of $2$-local Pauli terms.
The degree of each vertex could then be characterised by its \emph{Pauli-degree}, i.e., the number of Pauli operators emanating from said vertex.
The task was to construct a series of gadgets that could \begin{inparaenum}[(a)]
    \item reduce the Pauli-degree of each vertex to three, and
    \item reduce the graph to a planar one.
\end{inparaenum} We do not have the luxury of a Pauli decomposition and hence must resort to a different way of constructing gadgets.
The elements of the subdivision gadget serve as a basis for those to come.
We can leverage the construction of \cref{eq:OG_subdivision_Vmain} to determine how other gadgets should act.
We must always ensure the unperturbed Hamiltonian, perturbation term and perturbed Hamiltonian are stoquastic.

Using \cref{eq:general_stoq_rho}, we can express a general $2$-local stoquastic Hamiltonian as
\begin{equation*}
    H = K - \sum_{\{u,v\}\in \E(G)} \sum_{\boldsymbol{\mu}} h_{uv}^{\boldsymbol{\mu}} \rho_{u}^{\mu_u}\rho_{v}^{\mu_v}.
\end{equation*}
This is a somewhat cumbersome notation.
We instead will use $P$ to represent a general $\rho$-matrix, analogous to using $P$ to represent a general Pauli operator.
The notation $P_{u}$ describes one $\rho$-matrix acting on vertex $u$; it is convenient to think of $P_u = h_u^{\mu_u}\rho_{u}^{\mu_u}$\footnote{With this thought process $h_{uv}^{\boldsymbol{\mu}} = h_u^{\mu_u}h_v^{\mu_v}$.}.
For example, an edge $P_uP_v$ describes two $\rho$-matrices acting on vertices $u$ and $v$ and where $P_u$ does not necessarily equal $P_v$.
A general interaction edge for these stoquastic Hamiltonians is expressed as $P_uP_v + P_u^\dagger P_v^\dagger$ --- this is taken from \cref{eq:stoq_ham_CD_form}.
For brevity we say $\chi_{uv} = P_u + P_v^\dagger$ and $\chi_{uv}^\dagger = P_u^\dagger + P_v$.
Note that the interaction edge $P_uP_v + P_u^\dagger P_v^\dagger$ cannot be split into two components in general since $P_uP_v$ may not be Hermitian.
To the best of our knowledge, these are the only geometrically inspired gadgets for general local stoquastic Hamiltonians.

As we proceed with the gadgets, we will use diagrammatic representations to illustrate the interactions.
The left-hand side of each diagram represents the target Hamiltonian interaction, and the right-hand side represents the simulator Hamiltonian interaction.
Note that we often omit the negative signs in the diagrams for clarity.

\subsection{The Subdivision Gadget}
The purpose of this gadget is to show that a $2$-local stoquastic interaction between two system qubits, $u$ and $v$ can be simulated by $2$-local stoquastic interactions between the system qubits and a mediator qubit, $c$.
Let the target Hamiltonian be $H = \Gamma - (P_uP_v + P_u^\dagger P_v^\dagger)$ and define the perturbed Hamiltonian $\widetilde{H}=\Delta H_0+V$ with,
\begin{align}
    H_0 &= \ketbra{1}{1}_c, \notag\\
    V &= \sqrt{\Delta}~V_{{\rm main}} + V_{{\rm extra}}, \notag\\
    V_{{\rm main}} &= -\left(\chi_{uv}S^+_c + \chi_{uv}^\dagger S^-_c\right), \label{eq:subdivision_Vmain}\\
    V_{{\rm extra}} &= \Gamma + G. \notag
\end{align}
Note that the form of \cref{eq:subdivision_Vmain} is slightly misleading in that the diagrammatic representation of the interaction edge takes a slightly different form.
The diagrammatic representation of the interaction edge is shown in \cref{fig:subdivision_gadget}.

\begin{figure}[!ht]
    \centering
    \begin{tikzpicture}
        \pic{subdiv-gadget};
    \end{tikzpicture}
    \caption{The \textsl{Subdivision} gadget.}
    \label{fig:subdivision_gadget}
\end{figure}
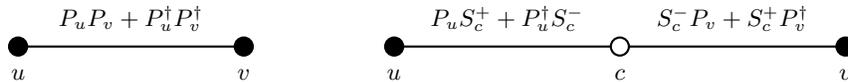

The term $G$ is diagonal and will be defined shortly.
The Hilbert space is split with a spectral gap of $\Delta$ and $\Pi_- = \ketbra{0}{0}_c$ and $\Pi_+ = \ketbra{1}{1}_c$ being the projectors onto the low-energy and high-energy subspaces respectively.
The unperturbed Hamiltonian is block diagonal with $(H_0)_{--} = 0$ and $(H_0)_{++} = 1$.
The perturbation term follows,
\begin{align*}
    V_{-+} &= - \sqrt{\Delta} \left(\chi^\dagger_{uv} \right)\ketbra{1}{0}_c, \\
    V_{+-} &= - \sqrt{\Delta} \left(\chi_{uv} \right)\ketbra{0}{1}_c, \\
    V_{-} &= (\Gamma + G)\ketbra{0}{0}_c.
\end{align*}
To second-order, it can be shown the effective Hamiltonian is
\begin{align*}
    H_{{\rm eff.}} &= \left(\Gamma + G - \chi_{uv}\chi_{uv}^\dagger\right) \ketbra{0}_c, \notag\\
                    &= \left(\Gamma - (P_uP_v + P_u^\dagger P_v^\dagger)\right)\ketbra{0}_c + \left(G - (P_uP_u^\dagger + P_v^\dagger P_v)\right) \ketbra{0}_c.
\end{align*}
Clearly if $G = P_uP_u^\dagger + P_v^\dagger P_v$ then we recover $H_{{\rm eff.}} = H\otimes\ketbra{0}_c$.

\subsection{The Cross Gadget}
The \textsl{Cross} gadget is used to remove a non-planar section of an interaction graph.
This gadget, however, creates additional edges, $E$.
The target Hamiltonian will be $H = \Gamma - E - (P_uP_v + P_u^\dagger P_v^\dagger + P_wP_s + P_w^\dagger P_s^\dagger)$ and the perturbation Hamiltonian will be $\widetilde{H}=\Delta H_0+V$ with,

\begin{align*}
    H_0 &= \ketbra{1}{1}_c, \\
    V &= \sqrt{\Delta}~V_{{\rm main}} + V_{{\rm extra}}, \notag\\
    V_{{\rm main}} &= -\left((\chi_{uv} + \chi_{sw})S^+_c + (\chi_{uv}^\dagger+\chi_{sw}^\dagger) S^-_c\right), \\
    V_{{\rm extra}} &= \Gamma + G.
\end{align*}

\begin{figure}[!ht]
    \centering
    \begin{tikzpicture}
        \pic{cross-gadget};
    \end{tikzpicture}
    \caption{The \textsl{Cross} gadget.
            The dashed lines represent the additional edges.
            }
    \label{fig:cross_gadget}
\end{figure}
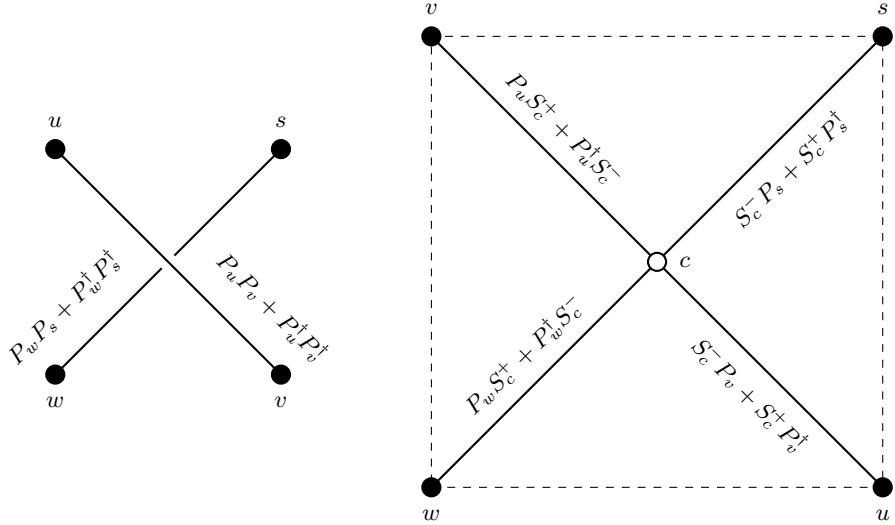

The diagrammatic representation of the interaction edge is shown in \cref{fig:cross_gadget}.
Using a similar analysis to that of the \textsl{Subdivision} gadget, it can be shown the effective Hamiltonian is
\begin{eqnarray}
    H_{{\rm eff.}} &=& \left(\Gamma -(P_uP_v + P_u^\dagger P_v^\dagger + P_wP_s + P_w^\dagger P_s^\dagger) - E\right)\ketbra{0}_c \notag\\
    && + \left(G - (P_uP_u^\dagger + P_v^\dagger P_v + P_wP_w^\dagger + P_s^\dagger P_s)\right) \ketbra{0}_c.
\end{eqnarray}
The extra edges (dashed lines) are contained in $E$ and lie between $us$, $sv$, $vw$ and $wu$.
The term $G$ is diagonal and is defined as before.

\subsection{The Fork Gadget}
The \textsl{Fork} gadget is used as one way to reduce the degree of a given vertex.
This gadget also creates an additional edge, $E$.
The target Hamiltonian will be $H = \Gamma - E - (P_uP_v + P_u^\dagger P_v^\dagger + P_vP_w + P_v^\dagger P_w^\dagger)$ and the perturbation Hamiltonian will be $\widetilde{H}=\Delta H_0+V$ with,

\begin{align*}
    H_0 &= \ketbra{1}{1}_c, \\
    V &= \sqrt{\Delta}~V_{{\rm main}} + V_{{\rm extra}}, \notag\\
    V_{{\rm main}} &= -\left((P_u + P_w + P_v^\dagger)S^+_c + (P_u^\dagger + P_w^\dagger + P_v) S^-_c\right), \\
    V_{{\rm extra}} &= \Gamma + G.
\end{align*}

\begin{figure}[!ht]
    \centering
    \begin{tikzpicture}
        \pic{fork-gadget};
    \end{tikzpicture}
    \caption{The \textsl{Fork} gadget.
The dashed line represents the additional edge.}
    \label{fig:fork_gadget}
\end{figure}
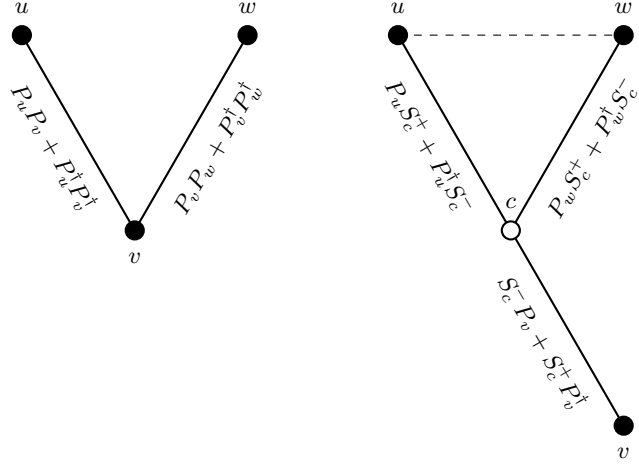

Again, a similar analysis entails and the effective Hamiltonian is
\begin{eqnarray}
    H_{{\rm eff.}} &=& \left(\Gamma - (P_uP_v + P_u^\dagger P_v^\dagger + P_vP_w + P_v^\dagger P_w^\dagger) - E\right)\ketbra{0}_c \notag\\
    && + \left(G - (P_u^\dagger P_u + P_vP_v^\dagger + P_w^\dagger P_w)\right) \ketbra{0}_c.
\end{eqnarray}
The additional edge is between $uw$ and is contained in $E$.
The term $G$ is diagonal and analogously defined as before.
It is important to note here that not any two arbitrary edges can be forked in this manner.
The edges must adhere to the specific structure shown in \cref{fig:fork_gadget}.
Moreover, the two edges in question must be of the form:
\begin{align*}
    P_uQ_v &+ P_u^\dagger Q_v^\dagger\,, & Q_vR_w &+ Q_v^\dagger R_w^\dagger\,,
\end{align*}
where we have chosen $P$, $Q$ and $R$ to represent the $\rho$-matrices making it clear where the connection lies.

\subsection{The Triangle Gadget}
The \textsl{Triangle} gadget is a hybrid of the \textsl{Subdivision} and \textsl{Fork} gadgets to reduce the degree of one qubit in a triad.
No additional edges between the target qubits are created.
The target Hamiltonian will be the same as the \textsl{Fork} gadget.
The \textsl{Triangle} gadget will first subdivide the edge $uv$ using a mediator qubit $c_1$ and $vw$ using a mediator qubit $c_2$.
Then proceed to fork between $c_1v$ and $vc_2$, cf. \cref{fig:triangle_gadget}.

\begin{figure}[!ht]
    \centering
    \begin{tikzpicture}
        \pic[scale=0.8]{triangle-gadget};
    \end{tikzpicture}
    \caption{The \textsl{Triangle} gadget.
The diagram shows the steps required to construct the gadget.}
    \label{fig:triangle_gadget}
\end{figure}
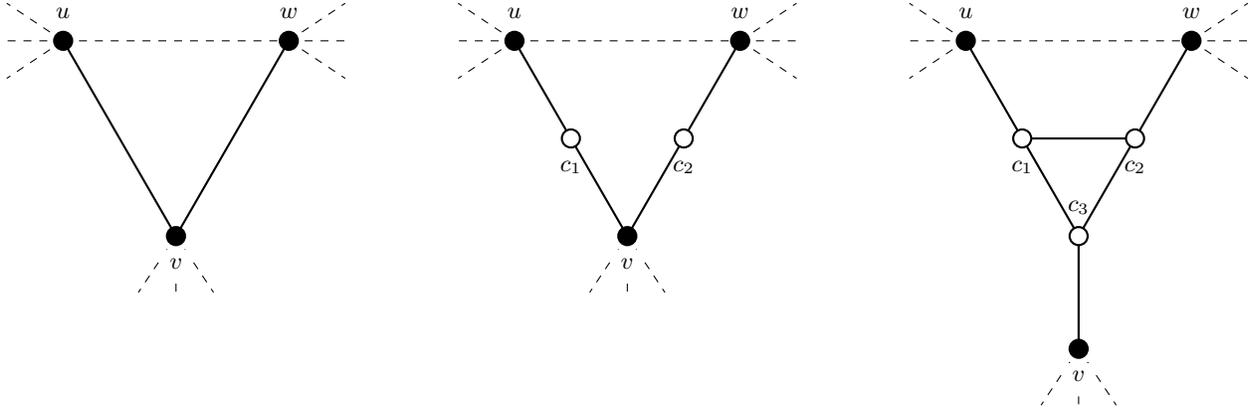

Notice that, with a small calculation, the conditions that permit the use of the \textsl{Fork} gadget are imparted onto the \textsl{Triangle} gadget.

\subsection{Further Gadget Combinations}
In this subsection, we briefly outline some useful gadget combinations that allow us to reduce the degree of a vertex and better `planarise' the interaction graph.
We do not provide an analysis as the gadgets are either parallel or series applications of gadgets already discussed and, hence, are efficient constructions.
An important gadget combination involves `localising' a vertex of high degree; this is a way of preparing said vertex for repeated application of the \textsl{Triangle} gadget.
Localising a vertex allows for specific edges to be grouped and the relevant degree of the vertex in question to be rapidly reduced, i.e., in a parallel swoop.
\cref{fig:localise_vertex} gives a diagrammatic overview of this.
The vertex \emph{type} will be discussed later.
Looped edges can be subdivided and subsequently forked using the \textsl{Triangle} gadget.
\cref{fig:forking_loops} shows how this works; with a little thought, it is clear which circumstances allow this to work.
Clearly, this also reduces the degree of the vertices in question.
Recall that two arbitrary edges cannot be forked, only specific pairs of edges.

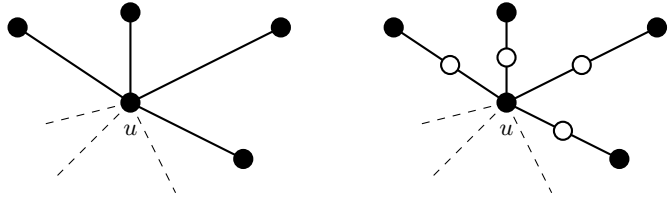
\begin{figure}[!ht]
    \centering
    \begin{tikzpicture}
        \pic{localisation};
    \end{tikzpicture}
    \caption{Localising a high degree vertex of the same type.}
    \label{fig:localise_vertex}
\end{figure}

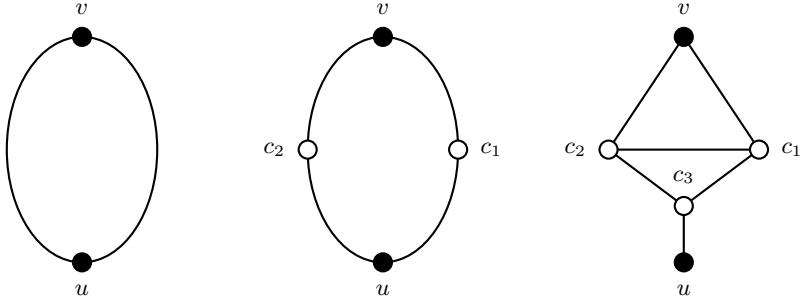
\begin{figure}[!ht]
    \centering
    \begin{tikzpicture}
        \pic{forking-loops};
    \end{tikzpicture}
    \caption{Forking a loop interaction to decrease the vertex degrees.}
    \label{fig:forking_loops}
\end{figure}

The number of edge crossings can also be reduced by subdividing that specific edge, \cref{fig:subdividing_cross_edge} demonstrates this.
Furthermore, the localisation of a single crossing can be achieved by four subdivision gadget applications.
\cref{fig:localise_crossing} shows this, and it is clear that an application of the \textsl{Cross} gadget here negates the additional edges between the system (black circle) qubits.

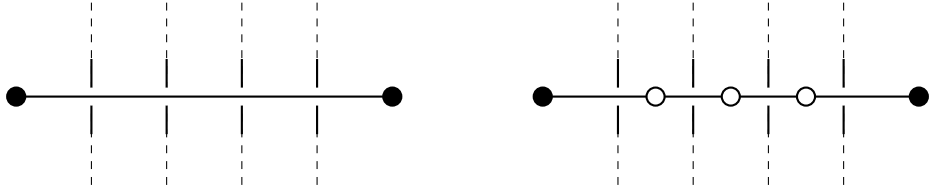
\begin{figure}[!ht]
    \centering
    \begin{tikzpicture}
        \pic{sub-cross-edge};
    \end{tikzpicture}
    \caption{Subdividing an edge to reduce the number of edge crossings.}
    \label{fig:subdividing_cross_edge}
\end{figure}

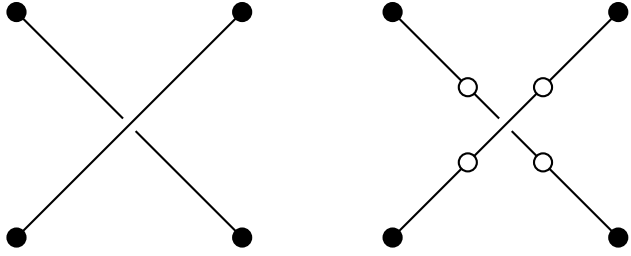
\begin{figure}[!ht]
    \centering
    \begin{tikzpicture}
        \pic{localise-cross};
    \end{tikzpicture}
    \caption{Localising a single crossing.}
    \label{fig:localise_crossing}
\end{figure}

It can be argued that a set of edges between two vertices make up a larger, parent stoquastic Hamiltonian term.
For example, a Hamiltonian term could be made of three $\rho$-matrix edges --- a $h^{00}\rho^0\rho^0$ edge, a $h^{33}\rho^3\rho^3$ edge and a $h^{20}\rho^2\rho^0 + h^{10}\rho^1\rho^0$ edge.

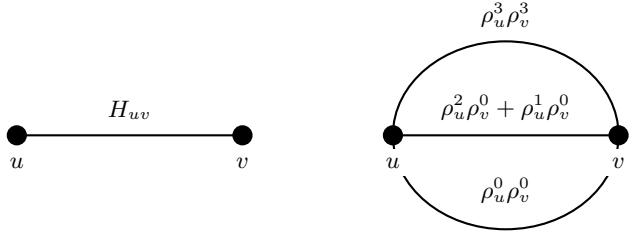
\begin{figure}[!ht]
    \centering
    \begin{tikzpicture}
        \pic{rho-decomp};
    \end{tikzpicture}
    \caption{An example interaction parent edge decomposed into three $\rho$-matrix edges.}
    \label{fig:example_rho_decomposition}
\end{figure}

We could therefore refer to the leftmost case of \cref{fig:example_rho_decomposition} as a \emph{parent edge} and the rightmost as \emph{$\rho$-matrix edges}.
The gadgets must be applied to $\rho$-matrix edges, and the degree of a vertex is defined in terms of the number of $\rho$-matrix edges.
The point of the parent edge is to demonstrate that it can always be decomposed into $O(1)$ $\rho$-matrix edges.
A spatially sparse graph has degree-$O(1)$ vertices when described using parent edges; the $\rho$-matrix equivalent also has degree-$O(1)$ vertices.
It turns out there are four unique $\rho$-matrix edges a given parent edge can be decomposed into.

In the worst-case situation there will be ten $\rho$-matrix edges that make up a parent edge.
The ten possible edges are: 
\begin{equation}\label{eq:rho_matrix_edges}
    \begin{array}{cccc}
        \rho^0_u\rho^0_v &\quad\qquad \rho^3_u\rho^0_v &\quad\qquad \rho^1_u\rho^0_v + \rho^2_u\rho^0_v &\quad\qquad  \\
        \rho^0_u\rho^3_v &\quad\qquad \rho^3_u\rho^3_v &\quad\qquad \rho^1_u\rho^2_v + \rho^2_u\rho^1_v &\quad\qquad  \rho^2_u\rho^2_v + \rho^1_u\rho^1_v \\
        \rho^0_u\rho^1_v + \rho^0_u\rho^2_v &\quad\qquad \rho^3_u\rho^1_v + \rho^3_u\rho^2_v &\quad\qquad \rho^1_u\rho^3_v + \rho^2_u\rho^3_v &\quad\qquad   
    \end{array}
\end{equation}
These are grouped into four classes.
For one vertex, in the worst-case, it is possible to reduce the degree to four.
It is clear that the terms in \cref{eq:rho_matrix_edges} allow for a minimum of degree-$3$ in the ideal worst-case situation, but this can result in a neighbouring worst-case vertex having degree-$4$.
Therefore, the worst-case situation, is a degree-$4$ vertex for all.

\begin{figure}[!ht]
    \centering
    \begin{tikzpicture}
        \pic{legal-fork};
    \end{tikzpicture}
    \caption{An example procedure showing a degree reduction of a vertex.}
    \label{fig:legal_fork_gadget}
\end{figure}
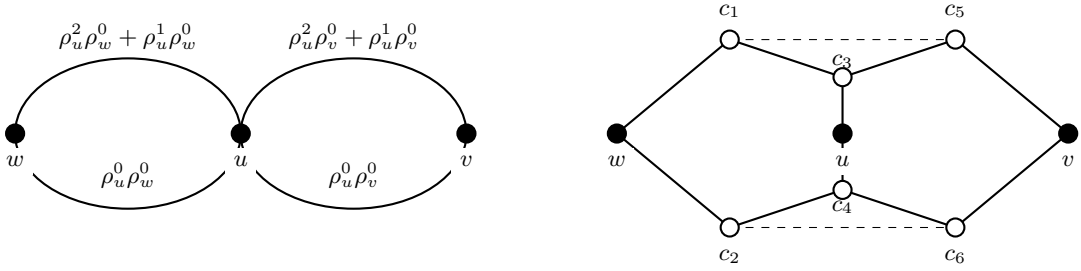

An example of a dual application of the \textsl{Fork} gadget is shown in \cref{fig:legal_fork_gadget}.
This reduces the degree of vertex $u$ by two.
It turns out that a further \textsl{Fork} gadget could be employed between $w$, $c_1$ and $c_2$ and also $v$, $c_5$ and $c_6$.
This is justified since in the original interaction edges, we had $\rho^0_w$ and $\rho^0_v$ terms.
For the example shown, the degree of each target vertex can be reduced to two.

\subsection{Stoquastic Hamiltonians on Lattice Geometries}
\planargraph

\begin{proof}
    The proof of this Theorem is analogous to that of~\cite[Lemma 2]{OT08}.
    Given a $2$-local stoquastic Hamiltonian on a spatially sparse graph, we know that each vertex is degree-$O(1)$ (with respect to parent edges).
    Decompose each parent edge into the appropriate $\rho$-matrix edges, leaving each vertex as degree-$O(1)$.
    This process can be done efficiently.
    We employ the localisation procedure by subdividing each edge.
    In doing so, we are preparing the vertices' degrees to be reduced.
    This can be done in parallel.

    Once the localisation has been done, we can group like-edges, i.e., into four groups.
    For each group, we apply the \textsl{Fork} gadget in parallel $O(\log d)$ times, which can be done efficiently.
    This reduces the degree of each group to $\lceil d_{\text{group}}/2 \rceil$.
    Repeating the \textsl{Triangle} process $O(1)$ times allows for each group to be at most degree-$1$.
    This leaves each vertex as degree-$4$.
    Up to this point in the process, we have required $O(1)$ reduction steps.

    We then proceed to reduce the number of crossings using the \textsl{Subdivision} and \textsl{Cross} gadgets.
    For $c=O(1)$ crossings, we require $O(\log c)$ repetitions of this procedure.
    This is done in parallel and can be done efficiently.
    The resulting graph is planar with each vertex of degree at most $4$ and has a straight-line drawing in the plane.
\end{proof}

\thrmTwoLSHsquarelattice

\begin{proof}
    A planar graph $G = (V,E)$ with vertices of degree at most $4$, all edges being straight-lines of length $O(1)$ and having $\Omega(1)$ angular separation with other edges can be represented on a $2$D square lattice, $\Lambda = (\V,\E)$.
    We embed the planar graph in the square lattice such that
    \begin{enumerate}
        \item each vertex $u\in V$ is mapped to a lattice site $\phi(u)\in \V$ within the boundary \\ $\partial\Lambda = [-O(\abs{V}), O(\abs{V})]^2$,
        \item each edge $\{u,v\}\in E$ is mapped to a lattice path $\phi(\{u,v\})\in \E$ of length $O(1)$ between $\phi(u)$ and $\phi(v)$ without crossing any other lattice path.
    \end{enumerate}
    The \textsl{Subdivision} gadget can be used to map the edges $\{u,v\}$ to lattice paths $\phi(\{u,v\})$.
    We must ensure that the lattice paths remain close to the original edges they are associated with.
    For a sufficiently fine grid, the paths will not cross outside a constant-sized square about each vertex.
    \cref{fig:square_lattice} demonstrates how a planar graph in this regime can be embedded into a square lattice.
    Since the edges are $O(1)$-length, the \textsl{Subdivision} gadget need only be used $O(1)$ times per edge to get the lattice path.
    Therefore, there is an efficient embedding of a $2$-local stoquastic Hamiltonian with a planar interaction graph to a $2$-local stoquastic Hamiltonian with a square-lattice interaction graph.
    Since the perturbation gadgets defined above are used, we can conclude that the low-energy subspaces of each Hamiltonian are close.
\end{proof}

\begin{figure}[!ht]
    \centering
    \begin{tikzpicture}
        \pic{planar-to-square};
    \end{tikzpicture}
    \caption{A planar graph of degree at most $4$ embedded on a $2$D square lattice.
                The grey squares represent small regions where paths are rerouted to avoid crossing.}
    \label{fig:square_lattice}
\end{figure}
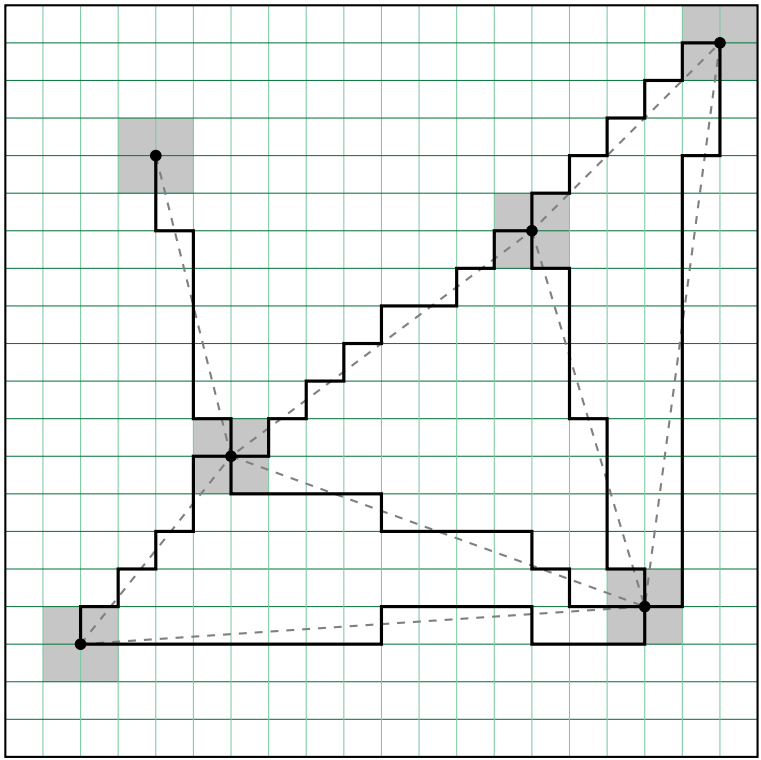

\cref{thrm:2LSH_square_lattice} can also be applied to the triangular lattice.
The proof is exactly the same and justified since the vertices in a triangular lattice are degree $6$.
Hence the following corollary is straightforwardly concluded.

\begin{corollary}
    The \sc{$2$-Local Stoquastic Hamiltonian} problem on a $2$D triangular lattice is \clw{StoqMA}{complete}.
\end{corollary}

Unlike work that considers a variety of lattice geometries~\cite{PM17}, the stoquastic Hamiltonians considered here are restricted to the square and triangular lattices due to the minimum planar graph degree being four.
If it were possible to reduce the degree of a vertex to three, then the Hamiltonian would be \clw{StoqMA}{complete} on a broader range of geometries.
We leave this as an open problem.

Our findings have shown that the complexity of the \sc{Local Stoquastic Hamiltonian} problem is retained even when restricted to lattice geometries such as the square and triangular lattices.
We have shown how to construct stoquastic perturbation gadgets, requiring a constant number of such, to reduce a local stoquastic Hamiltonian on a spatially sparse lattice to a planar graph where each vertex is of degree at most four.
A subsequent efficient embedding of the planar graph into a square lattice reveals that the \sc{Local Stoquastic Hamiltonian} problem remains \clw{StoqMA}{complete}.
In the following section, we discuss a more restrictive class of stoquastic Hamiltonians, namely those constructed from combinations of Pauli operators.
This is a logical step to make as there is a wealth of gadgets already known for Pauli interactions~\cite{BL08,PM17,CM16} and therefore, we can study their complexity using such.

\section{Stoquastic Pauli Hamiltonians}\label{sec:stoq-pauli-ham}
It is natural to try to consider a Pauli decomposition of a stoquastic Hamiltonian in hopes of applying the original techniques of \refcite{OT08}.
However, it is not hard to see that certain Pauli terms are not stoquastic.
This causes a significant problem when using perturbative gadgets --- they must be stoquastic for local stoquastic Hamiltonians.
We can, therefore, ask the question: \emph{what kinds of Pauli Hamiltonians are stoquastic}? 
Interestingly, it has been shown to be \clw{coNP}{hard} to decide whether a given Hamiltonian is stoquastic and $\varSigma^{\cl{p}}_2$-hard to decide whether a given Hamiltonian is stoquastic under single-qubit unitary transformations~\cite{IPMKT22}.
However, the problem of deciding whether a given Hamiltonian is stoquastic is not of interest here.
Recall a term-wise stoquastic Hamiltonian is one where each local Hamiltonian term is stoquastic.
This definition is less general than that of Ref.~\cite[Definition 2]{IPMKT22} since we are specifically interested only in the computational basis stoquasticity (\cref{def:stoq_ham}).
The types of local Hamiltonians we have considered thus far are term-wise stoquastic.

We now consider and formally define \emph{Pauli-term-wise} stoquastic Hamiltonians.
It is well-known the Pauli matrices form a basis for $2\times 2$ Hermitian matrices.
Given a $2$-local stoquastic Hamiltonian term, if we were to decompose the term into a sum of $2$-local Pauli operators, the resulting sub-Hamiltonian terms would not be term-wise stoquastic in general.
For example, consider the following stoquastic term as a sum of Pauli operators,
\begin{equation*}
    \begin{bmatrix}
        0 & 0 & -3 & 0 \\
        0 & 0 & 0 & -1 \\
        -3 & 0 & 0 & 0 \\
        0 & -1 & 0 & 0
    \end{bmatrix} = -2X I - XZ.
\end{equation*}
The term $XZ$ is not stoquastic and can never be made so\footnote{With respect to the computational basis.}.
If a stoquastic Hamiltonian \emph{can} be decomposed into a sum of stoquastic Pauli terms, then we say the Hamiltonian is \emph{Pauli-term-wise stoquastic}.

\begin{definition}\label{def:pauli-term-wise_stoq_ham}
    A local Hamiltonian is \emph{Pauli-term-wise stoquastic} if each local Hamiltonian term is stoquastic and each term is a single local Pauli interaction.
\end{definition}

The set of allowed single $2$-local Pauli terms that leave a Hamiltonian term stoquastic are $\mathcal{S}_2 \coloneqq \{ I, -X,Z,$ $ZZ, -XX\}$; this idea can be generalised to higher degrees of locality.
Of course, the terms comprised of Pauli-$X$ matrices must have an associated non-negative coefficient.
If we slightly alter \cref{def:pauli-term-wise_stoq_ham} to allow for a grouped sum of Pauli operators, then such terms as $-(XX+YY)$ (representing a `single' Pauli term) are stoquastic.
We could now define a \emph{grouped Pauli-term-wise stoquastic} Hamiltonian in the same way as \cref{def:pauli-term-wise_stoq_ham}; however, the scope of possible terms here exceeds that of $\mathcal{S}_2$ and so it makes more sense to classify such Hamiltonians as term-wise stoquastic.
Example interaction terms that live in this family include the antiferromagnetic $XY$-interaction and Heisenberg interactions, i.e., $-(XX+YY)$ and $-(XX+YY+ZZ)$.

Pondering which known $2$-local Hamiltonians admit a special case of \cref{def:pauli-term-wise_stoq_ham} --- the transverse field Ising model is one such example where it and the antiferromagnetic variant are \clw{StoqMA}{complete}~\cite{BH16, PM17}.
Taking a closer look at $\mathcal{S}_2$, it bares resemblance with the $ZZXX$-Hamiltonian presented by Biamonte and Love~\cite{BL08}, except there is now a condition on the $XX$ and $X$ terms.
We denote this as the (\textsc{x}\textsuperscript{$-$}--~\textsc{z}/\textsc{x}\textsuperscript{$-$}--~\textsc{z})-Hamiltonian --- the parent $2$-local Pauli-term-wise stoquastic Hamiltonian:

\begin{equation}\label{eq:parent_2LSPH}
    H = \sum_{\{u,v\}\in \E(G)} J_{uv}^{^{(\leq 0)}} X_uX_v + L_{uv}Z_uZ_v + \sum_{u\in \V(G)} f_u^{^{(\leq 0)}} X_u + h_uZ_u.
\end{equation}
The coefficients $L,h\in\mathbb{R}$ and $J,f \in\mathbb{R}^-_0$.
Up to an overall constant, the parent $2$-local stoquastic Pauli interaction term follows,
\begin{equation*}
    \begin{bmatrix}
        h_{u} + L_{uv} + h_{v} & f_{v}^{^{(\leq 0)}} & f_{u}^{^{(\leq 0)}} & J_{uv}^{^{(\leq 0)}} \\
        f_{v}^{^{(\leq 0)}} & h_{u} - L_{uv} - h_{v} & J_{uv}^{^{(\leq 0)}} & f_{u}^{^{(\leq 0)}} \\
        f_{u}^{^{(\leq 0)}} & J_{uv}^{^{(\leq 0)}} & - h_{u} - L_{uv} + h_{v} & f_{v}^{^{(\leq 0)}} \\
        J_{uv}^{^{(\leq 0)}} & f_{u}^{^{(\leq 0)}} & f_{v}^{^{(\leq 0)}} & - h_{u} + L_{uv} - h_{v}
    \end{bmatrix}.
\end{equation*}
Clearly, this heavily restricts the types of Hamiltonians we can consider.
A general $2$-local stoquastic Hamiltonian term (i.e., over two qubits) will have ten degrees of freedom.
A parent Hamiltonian term has only six degrees of freedom.
Not all possible diagonal terms can be created using combinations of $L$ and $h$.
We aim to prove such Hamiltonians are \clw{StoqMA}{complete}.
Notice that the Ising model and transverse field Ising model are special cases of the parent Hamiltonian.

\begin{theorem}
    The \textnormal{(\textsc{x}\textsuperscript{$-$}--~\textsc{z}/\textsc{x}\textsuperscript{$-$}--~\textsc{z})}-Hamiltonian is \clw{StoqMA}{complete}.
\end{theorem}

\begin{proof}
    The transverse field Ising model is \clw{StoqMA}{complete}~\cite[Theorem 1]{BH16}.
    The set of interactions that make the transverse field Ising model are subset of the (\textsc{x}\textsuperscript{$-$}--~\textsc{z}/\textsc{x}\textsuperscript{$-$}--~\textsc{z})-Hamiltonian interaction set; therefore the (\textsc{x}\textsuperscript{$-$}--~\textsc{z}/\textsc{x}\textsuperscript{$-$}--~\textsc{z})-Hamiltonian is \clw{StoqMA}{hard}.
    To show the Hamiltonian is in \cl{StoqMA}, we simply conclude that we have a $2$-local stoquastic Hamiltonian and, by definition, is in \cl{StoqMA}.
    Thus, the (\textsc{x}\textsuperscript{$-$}--~\textsc{z}/\textsc{x}\textsuperscript{$-$}--~\textsc{z})-Hamiltonian is \clw{StoqMA}{complete}.
\end{proof}

To consider a geometric restriction to the (\textsc{x}\textsuperscript{$-$}--~\textsc{z}/\textsc{x}\textsuperscript{$-$}--~\textsc{z})-Hamiltonian would prove difficult.
Firstly, arguments constructed in prior work~\cite{OT08,CM16} fail due to the inability to simulate the $2$-local terms like $-XX$ and $ZZ$ using only those of $\mathcal{S}_2$.
In the original situation~\cite{BL08}, we are able to construct appropriate reductions given that the stoquastic restriction is not present.
Secondly, the gadgets presented above are insufficient as we must use only terms form $\mathcal{S}_2$.
We leave this as an open problem but conjecture that it is \clw{StoqMA}{complete}.

There are obvious restrictions to \cref{eq:parent_2LSPH} that are stoquastic.
For example, we can consider the following,
\begin{align}
    H(s) &= \sum_{\{u,v\}\in \E(G)} J_{uv}^{^{(\leq 0)}} X_uX_v + L_{uv}(s)Z_uZ_v + \sum_{u\in \V(G)} f_u^{^{(\leq 0)}} X_u + h_u(s)Z_u,\label{eq:parent_pauli_s}\\
    H &= \sum_{\{u,v\}\in \E(G)} J_{uv}^{^{(\leq 0)}} X_uX_v + L_{uv}Z_uZ_v,\notag
\end{align}
where $L(s)$ and $h(s)$ are sign-restricted parameters.
Additionally, motivated by those Pauli interactions that when grouped are stoquastic, we can consider the following,
\begin{align}
    H &= \sum_{\{u,v\}\in \E(G)} J_{uv}^{^{(\geq 0)}} \left(-\alpha X_uX_v + \gamma Z_uZ_v \right),\label{eq:pseduo_pauli_one}\\
    H &= \sum_{\{u,v\}\in \E(G)} J_{uv}^{^{(\leq 0)}} \left(\alpha X_uX_v + \gamma Z_uZ_v \right),\label{eq:pseduo_pauli_two}
\end{align}
where $\gamma \in \mathbb{R}$ and $\alpha \in \mathbb{R}_0^+$.
Since all examples are stoquastic we can determine that they must lie in \cl{StoqMA}.
The difficulty, of course, lies with proving hardness.
The grouped Pauli stoquastic Hamiltonians of \cref{eq:pseduo_pauli_one} and \cref{eq:pseduo_pauli_two} are of a similar format to those Hamiltonians considered by Cubitt, Piddock and Montonaro~\cite{CM16,PM17}.
Using their notation we would say \cref{eq:pseduo_pauli_one} is the $\{-\alpha XX + \gamma ZZ\}^+$-Hamiltonian and \cref{eq:pseduo_pauli_two} is the $\{\alpha XX + \gamma ZZ\}^-$-Hamiltonian.

It is likely some of the above Hamiltonians are \clw{StoqMA}{complete}, however, it would not be surprising if an efficient algorithm could be given to find the ground state.
For \cref{eq:parent_pauli_s} when $\text{sgn}(s)=1$ the Hamiltonian has antiferromagnetic $Z$ interactions.
It is known the antiferromagnetic transverse field Ising model is \clw{StoqMA}{complete}~\cite{PM17}.
Notice that conjugating \cref{eq:parent_2LSPH} by $X$ on all qubits where $h_u < 0$ leaves the sign of the $XX$, $ZZ$ and $X$ terms unchanged but flips the sign of the targeted $Z$ terms.
So we can assume that $h_u \geq 0$ for all $u\in \V(G)$.
A similar argument could be given for the $f_u$ terms via $Z$ conjugation on the associated qubits.
Therefore, since antiferromagnetic transverse field Ising interaction can simulate general version~\cite[Theorem 5]{PM17}, we can conclude that \cref{eq:parent_pauli_s} is \clw{StoqMA}{complete}.
The case where $\text{sgn}(s)=-1$ is more difficult to analyse.
It is known the ferromagnetic transverse field Ising model is in \cl{BPP}; however, $H(-1)$ includes more terms.

Indeed it is known from Bravyi and Gosset that a variation of the ferromagnetic $XY$-Hamiltonian has a \cl{FPRAS} for the partition function~\cite{BG17,TheHamiltonianJungle}.
This result implies that the $\{-XX+\beta YY\}^+$-Hamiltonian lies in \cl{BPP} for a choice of $\beta \in [-1,1]$~\cite{PiddockThesis}.
By the global relabelling property, we can similarly deduce that \cref{eq:pseduo_pauli_one}, for a choice of $\alpha = 1$ and $\gamma \in [-1,1]$, is in \cl{BPP}.

\section{Conclusion}
In this work, we considered the complexity of the \sc{Local Stoquastic Hamiltonian} problem on lattice geometries.
While the general \sc{$2$-Local Hamiltonian} problem has been shown to be \cl{QMA}-complete on a square lattice~\cite{OT08}, the complexity of the $2$-local stoquastic Hamiltonian on a square lattice had not been established as \clw{StoqMA}{complete}.

We demonstrated that the problem of determining the ground-state energy of a $2$-local stoquastic Hamiltonian on a square lattice is indeed \clw{StoqMA}{hard}.
We achieved this by developing perturbative gadget techniques, to reduce general $2$-local stoquastic Hamiltonians to $2$-local ones on a square lattice, analogous to the work by Oliveira and Terhal~\cite{OT08}.
Additionally, we discussed a special case of the adiabatic Hamiltonians proposed by Biamonte and Love~\cite{BL08} and showed that it is also \clw{StoqMA}{complete}.

\begin{result*}[\cref{thrm:2LSH_square_lattice}]
    The \sc{$2$-Local Stoquastic Hamiltonian} problem on a $2$D square qubit lattice is \clw{StoqMA}{complete}.
\end{result*}

Future research could explore Hamiltonians that are known to be stoquastic but whose ground-state energy problems are not yet classified as \clw{StoqMA}{complete}.
The standard example is the antiferromagnetic Heisenberg Hamiltonian~\cite{PM17}.
Further investigation into less restrictive guiding state constructions for stoquastic Hamiltonians, beyond the original formulation by Bravyi~\cite{Bravyi15}, could also provide valuable insights.
For example, it is known that \cl{StoqMA} with an `easy witness' (\cl{eStoqMA}) is equivalent to \cl{MA} and hence the resulting problem --- the \sc{Local Stoquastic Hamiltonian} problem with an easy witness ground state is known to be \clw{MA}{complete}~\cite{Liu21}.
Could it be possible to consider the class \cl{SStoqMA}, where the `\cl{S}' denotes `subset-state witness' analogous to the work of Grilo \emph{et al}.~\cite{GKS15} concerning \cl{SQMA}, and subsequently the complexity of the \emph{Guided} \sc{Local Stoquastic Hamiltonian} problem\footnote{Not to be confused with the variant under a similar name proposed by Bravyi~\cite{Bravyi15}}~\cite{GLG22,CFGHLGMW23}? If it can be shown that \cl{SStoqMA} is equivalent to \cl{StoqMA}, then it might follow that the Guided \sc{Local Stoquastic Hamiltonian} problem is tractable in some sense.

\subparagraph{Merlin's power.} An open problem we briefly discuss concerns understanding the extent of computational power Merlin provides.
Typically, Merlin supplies a witness state $\ket{\xi}$ and Arthur, polynomially bounded in resources, prepares a circuit to verify the input.
\cref{rmk:Merlins_message} gives a condensed tuple representation of the components Merlin and Arthur have for a standard \cl{StoqMA} circuit.
By constraining Arthur's ancillae qubits, we can explore how much computational assistance Merlin must contribute, impacting the power of the verification process. 
Different configurations --- such as Merlin supplying fractions of ancilla qubits or even circuit components --- highlight how Merlin's support affects the verifier's power.
This may have implications for communication complexity and resource-limited quantum verification.

\acknowledgments
GW would like to thank Karl Lin, Ryan Mann and Samuel Elman for helpful discussions and feedback. 
We would also like to thank the anonymous reviewers for their insightful comments that helped improve this manuscript.
GW and MJB were supported by the ARC Centre of Excellence for Quantum Computation and Communication Technology (CQC2T), project number CE170100012. 
GW was supported by a scholarship from the Sydney Quantum Academy. 
This material is also in part based upon work by MJB supported by the Defense Advanced Research Projects Agency under Contract No. HR001122C0074.

\bibliographystyle{quantum}
\bibliography{ref}

\onecolumn
\appendix

\section{Parallel Gadget Application}\label{app:parallel}

When employing perturbation gadgets, it is important to ensure they do not cross-interfere and create unwanted interaction terms.
For each gadget considered in this work, no construction, to the appropriate order, suffers any unwanted effects.
The fourth-order perturbation terms used to reduce a general $3$-local stoquastic Hamiltonian to a special $3$-local one do incur cross-gadget terms, however, such contributions can be ignored in a certain regime~\cite{KKR06,BDOT06}.
We do not discuss that here but instead show that the $3$-to-$2$-local perturbation gadgets do not suffer from cross-gadget terms when applied in parallel.
Additionally, the subdivision and geometric gadgets do not suffer from cross-gadget terms when applied in parallel.
The arguments are essentially universal for all cases and so we start with the simpler case.

\subsection{Subdivision and Geometrical Gadgets}\label{sec:app_subdivision_geometrical_gadgets}
The gadgets, \textsl{Subdivision}, \textsl{Cross} and \textsl{Fork}, use the same unperturbed penalty Hamiltonian.
For a series of edges that are acted upon by these perturbation gadgets in parallel, we must show that there are no cross-gadget terms to second-order.
The unperturbed penalty Hamiltonian is given by
\begin{equation*}
    H_0 = \sum_{i=1}^{p} H_0^{(i)}= \sum_{i=1}^{p} \Delta_i \ketbra{1}_{c_i},
\end{equation*}
where $c_i$ denotes the $i$-th mediator qubit $c$.
Note that we have slightly departed from the original notation in that each $H_0^{(i)}$ contains a parameter $\Delta_i$ that can be chosen independently.
Clearly, the low-energy space of $H_0$ is $\ket{0^p}$ and the high-energy space is given by states with Hamming weight at least $1$.
Clearly, $\Pi_- = \ketbra{0^p}$ and $\Pi_+ =  I - \Pi_-$.
The perturbation term is 
\begin{align*}
    V &= V_{{\rm main}} + V_{{\rm main}}, \\
    V_{{\rm main}} &= \sum_{i=1}^p \sqrt{\Delta_i}\; V_{{\rm main}}^{(i)},\\
    V_{{\rm extra}} &= \sum_{i=1}^p V_{{\rm extra}}^{(i)}.
\end{align*}
We define each $V_{{\rm main}}^{(i)}$ to only act non-trivially on the $i$-th mediator qubits and the relevant edge's target qubits.
Let $V_{{\rm extra}}^{(i)}$ act only on the relevant system qubits, i.e., not the mediator qubits.
The first-order perturbation term shows 
\begin{equation*}
    V_{-} = \left(V_{{\rm extra}}\right)_{-} = \sum_{i=1}^p \left(V_{{\rm extra}}^{(i)}\right)_{-},
\end{equation*} 
assuming that all first-order projection terms of $V_{{\rm main}}$ are zero.
Since $V_{{\rm extra}}$ terms only act on system qubits, we can conclude the cross-projected terms vanish.
We must show $V_{-+}(H_0)^{-1}V_{+-}$ produces no cross gadget terms.
By definition
\begin{equation*}
    V_{-+}(H_0)^{-1}V_{+-} = \sum_{i,j,k} \left(\sqrt{\Delta_i}\; V_{{\rm main}}^{(i)}\right)_{-+} (\Delta_j)^{-1}\ketbra{1}_{c_j} \left(\sqrt{\Delta_k}\; V_{{\rm main}}^{(k)}\right)_{+-}.
\end{equation*}
We know each $V_{{\rm main}}^{(i)}$ acts non-trivially on $c_i$ and the relevant edge's target qubits.
Therefore, the low-energy state $\ket{0}_{c_i}$ can only be excited by a contribution from $\left(V_{{\rm main}}^{(i)}\right)_{+-}$.
Due to orthogonality 
\begin{equation*}
    \sum_{i,j,k} \left(\sqrt{\Delta_i}\; V_{{\rm main}}^{(i)}\right)_{-+} (\Delta_j)^{-1}\ketbra{1}_{c_j} \left(\sqrt{\Delta_k}\; V_{{\rm main}}^{(k)}\right)_{+-} = \sum_{i} \left(V_{{\rm main}}^{(i)}\right)_{-+} \ketbra{1}_{c_i}\left(V_{{\rm main}}^{(i)}\right)_{+-}.
\end{equation*}
Clearly, there are no cross-gadget terms to second-order.

\subsection{3-to-2-Local Gadgets}
When reducing a $3$-local stoquastic Hamiltonian to a $2$-local one, the specific gadget is used in parallel across all interaction edges in $\Omega_3$.
We must show that to second-order there are no cross-gadget terms.
This perturbation technique uses a third-order construction.
The first-order and second-order corrections simply create energy shifts.
The third-order terms are the ones that create the new interaction terms.
The unperturbed Hamiltonian acts on a triple of mediator qubits per interaction edge,
\begin{equation*}
    H_0 = \sum_{i=1}^{q} H_0^{(i)}= \sum_{i=1}^{q} \Delta_i \left(\ketbra{001}_{i_1,i_2,i_3} + \dots + \ketbra{110}_{i_1,i_2,i_3}\right).
\end{equation*}
The low-energy space for a single edge is spanned by $\ket{000}, \ket{111}$, so $\Pi_- = \prod_{i} \Pi_-^{(i)}$ and $\Pi_+ =  I - \Pi_-$.
The perturbation term is
\begin{align*}
    V &= V_{{\rm main}} + V_{{\rm extra}}, \\
    V_{{\rm main}} &= \sum_{i=1}^q \Delta^{2/3}_i\; V_{{\rm main}}^{(i)},\\
    V_{{\rm extra}} &= \sum_{i=1}^q V_{{\rm extra}}^{(i)}.
\end{align*}
We define each $V_{{\rm main}}^{(i)}$ to only act non-trivially on the three $i$-th mediator qubits and the relevant edge's target qubits.
Let $V_{{\rm extra}}^{(i)}$ act only on the relevant system qubits, i.e., not the mediator qubits.
The first-order perturbation term shows
\begin{equation*}
    V_{-} = \left(V_{{\rm extra}}\right)_{-} = \sum_{i=1}^q \left(V_{{\rm extra}}^{(i)}\right)_{-},
\end{equation*}
assuming that all first-order projection terms of $V_{{\rm main}}$ are zero.
Since $V_{{\rm extra}}$ terms only act on system qubits we can conclude the cross projected terms vanish.
The term $V_{-+}(H_0)^{-1}V_{+-}$ must produce no cross-gadget terms.
The term $V_{{\rm main}}^{(i)}$ can only affect the three mediator qubits $i_1,i_2$ and $i_3$ and the relevant edge's target qubits.
It is therefore not hard to see 
\begin{equation*}
    V_{-+}(H_0)^{-1}V_{+-} = \sum_{i} \Delta_i^{1/3}\;\left(V_{{\rm main}}^{(i)}\right)_{-+} (H_0^{(i)})^{-1}\left(V_{{\rm main}}^{(i)}\right)_{+-} = K I,
\end{equation*}
where $K = \sum_i K_i$ is a constant.
Therefore, there are no cross-gadget terms to second-order.
The third-order terms are where there is potential for cross-gadget terms.
The third-order terms are what cause the target interaction terms.
Analogous to the previous subsection, we can show that the third-order terms do not produce cross-gadget terms.
Notice that
\begin{equation*}
    V_{++} = \sum_{i=1}^q \left(V_{{\rm main}}^{(i)}\right)_{++},
\end{equation*}
hence the third-order term is given by
\begin{equation*}
    V_{-+}(H_0)^{-1}V_{++}(H_0)^{-1}V_{+-} = \sum_i \left(V_{{\rm main}}^{(i)}\right)_{-+} (H_0^{(i)})^{-1} \left(V_{{\rm main}}^{(i)}\right)_{++} (H_0^{(i)})^{-1} \left(V_{{\rm main}}^{(i)}\right)_{+-}.
\end{equation*}
Using the definition of the $V_{{\rm main}}^{(i)}$ terms and isolating to a conjugate pair to simplify the calculation, it is not hard to see that $\left(V_{{\rm main}}^{(i)}\right)_{++}$, $(H_0^{(i)})^{-1} \left(V_{{\rm main}}^{(i)}\right)_{+-}$ and $\left(V_{{\rm main}}^{(i)}\right)_{-+} (H_0^{(i)})^{-1}$ each contribute a different operator $O$ to the product resulting in a term like $O_1O_2O_3 \ketbra{111}{000} + O_1O_2O_3 \ketbra{000}{111} = O_1O_2O_3 X_c$.
Therefore, using arguments similar to those in the previous subsection, we can conclude that there are no cross-gadget terms to third-order.

\section{Effective Hamiltonian of the Stoquastic Subdivision Gadget}\label{app:stoquastic_subdivision_gadget}
In this appendix, we prove that the stoquastic subdivision gadget produces the desired effective Hamiltonian using the second-order reduction lemma (\cref{lma:second_order_reduction}).
Recalling the setup, we have:
\begin{align*}
    H &= K - \sum_{a=1}^M \left(C_a \otimes D_a + C_a^\dagger \otimes D_a^\dagger \right),\\
    \Delta H_0 &= \Delta \sum_{a=1}^M \ketbra{1}_a, \\
    V &= \sqrt{\Delta}\; V_{{\rm main}} + V_{{\rm extra}}, \\
    V_{{\rm main}} &= -\sum_{a=1}^M \left(C_a + D_a^\dagger \right)S_a^+ + \left(C_a^\dagger + D_a \right)S_a^-, \\
    V_{{\rm extra}} &= \sum_{a=1}^M \left(C_a^\dagger \otimes C_a + D_a \otimes D_a^\dagger \right).
\end{align*}
The low- and high-energy subspace can be defined using the Hamming weight of a bit string.
Moreover, it is clear that $\Pi_- = \ketbra{0^M}$ and $\Pi_+ =  I - \Pi_-$, hence the high-energy subspace is spanned by states with Hamming weight at least $1$.
The calculation of the effective Hamiltonian is similar to \cref{sec:app_subdivision_geometrical_gadgets}, thus using \cref{lma:second_order_reduction} we have that 
\begin{align*}
        H_{{\rm eff.}} &= \sum_{a=1}^M \left(C_a^\dagger C_a + D_a D_a^\dagger \right)_{-}  \notag\\ 
        &\qquad- \sum_{a,b,c} \left(\sqrt{\Delta}\,  \left(C_a + D_a^\dagger \right)S_a^+ + \left(C_a^\dagger + D_a \right)S_a^-\right)_{-+} \\
        &\qquad\times \frac{1}{\Delta} \ketbra{1}{1}_b \, \left(\sqrt{\Delta}\,  \left(C_c + D_c^\dagger \right)S_c^+ + \left(C_c^\dagger + D_c \right)S_c^-\right)_{+-}.
\notag
\end{align*}
The key terms to analyse are $\left(S_a^\pm\right)_{-+} \ketbra{1}{1}_b \left(S_c^\pm\right)_{+-}$.
Trivially, $\left(S_a^+\right)_{-+} = 0$ and $\left(S_a^-\right)_{+-} = 0$.
It is not hard to show
\begin{align*}
    \left(S_a^+\right)_{+-} &= \bigl(\ket{1}_a \otimes I_{M\setminus a}\bigr)\bra{0^M}, &  
    \left(S_a^-\right)_{-+} &= \ket{0^M}\bigl(\bra{1}_a \otimes I_{M\setminus a}\bigr),
\end{align*}
where the $1$ is in the $a$-th position.
The only terms that survive in the latter half of the above expression are thus given by 
\begin{equation*}
    \left(S_a^-\right)_{-+} \ketbra{1}{1}_b \left(S_c^+\right)_{+-} = \delta_{ab}\delta_{bc}\ketbra{0^M}{0^M},
\end{equation*}
reducing the overall effective Hamiltonian to
\begin{align*}
    H_{{\rm eff.}} &= \sum_{a=1}^M \left(C_a^\dagger C_a + D_a D_a^\dagger\right)\ketbra{0^M}{0^M} - \sum_{a=1}^M \left(C_a + D_a^\dagger \right) \left(C_a^\dagger + D_a \right)\ketbra{0^M}{0^M}, \notag\\
    H_{{\rm eff.}} &= - \sum_{a=1}^M \left(C_aD_a + D_a^\dagger C_a^\dagger\right)\ketbra{0^M}{0^M}.
\end{align*}
This concludes the proof.

\section{Composition Law of Perturbation Gadgets}\label{app:composition}
\cref{def:Hamiltonian_simulation} and \cref{lma:eigenvalue_simulation} can be combined in a certain way to bound the error incurred from a composition of simulations.
Moreover, imagine the following chain of simulations: $H_1$ is an $(\eta_{1},\epsilon_{1})$ simulator of $H$ and $H_2$ is an $(\eta_{2},\epsilon_{2})$ simulator of $H_1$.
It is not hard to see that $H_2$ is an $(\eta,\epsilon)$ simulator of $H$.
But specifically, what the values of $\eta$ and $\epsilon$ are is not immediately clear.
Let $\Delta_j$ represent the spectral gap of Hamiltonian $H_j$.
It is known that these values can be bounded as~\cite{BH16}:
\begin{align*}
    \eta &= \eta_{1} + \eta_{2} + O(\frac{\epsilon_{2}}{\Delta_{1}}), \\
    \epsilon &= \epsilon_{1} + \epsilon_{2} + O(\frac{\epsilon_{2}\norm{H}}{\Delta_{1}}).
\end{align*}
Hence, $H_2$ is an $(\eta,\epsilon)$ simulator of $H$.
This is a general result that can be applied to any number of simulations.
In a general setting where we consider the composition of a constant number of simulations $C = O(1)$, i.e., $H_1$ is a $(\eta_{1},\epsilon_{1})$ simulator of $H$ and $H_2$ is a $(\eta_{2},\epsilon_{2})$ simulator of $H_1$ and so on.
Then we have that
\begin{equation*}
    \epsilon = \sum_{j=1}^{C-1} \epsilon_j \leq \frac{\epsilon}{2} + O(\epsilon)\max_{j}\frac{\norm{H_{j+1}}}{\Delta_j},
\end{equation*}
after choosing $\epsilon_j \leq \frac{\epsilon}{2C}$.
A similar expression can be given for the $\eta$-error, but we do not present it here.
The spectral gaps can be appropriately chosen to ensure that the error is bounded by $\epsilon$.
Therefore, the composition of simulations is valid when a constant number of simulations are considered.
If a polynomial number of simulations is considered, then each error term must be chosen to be appropriately small.
A parallel application of perturbation gadgets has an error upper bound by the maximum error of the individual gadgets.
A series application requires the composition law to be applied.

\section{Toffoli Gate Manipulation}\label{app:toffoli}
Here we discuss the manipulation of a Toffoli gate to nearest-neighbour gates.
Within the \cl{StoqMA} circuits, the Toffoli gates might be long-range and hence a subsequent transition to a spatially sparse graph seemingly requires nearest-neighbour gates.
By using a \Gate{Swap} network, it is possible to achieve this property.
Note that \Gate{Swap} gates are allowed in \cl{StoqMA} circuits since they are just three \Gate{Cnot} gates.

\proptoffdecomp*

\begin{proof}
    Without loss of generality, assume qubit registers $a < b < c$ and define the length $r \coloneqq c - a$.
    We use a \Gate{Swap} network to move qubits $a$ and $c$ to positions $b-1$ and $b+1$, respectively, while keeping qubit $b$ fixed.
    This requires $r$ \Gate{Swap} gates, each decomposable into 3 \Gate{Cnot} gates, totalling $\Theta(r)$ \Gate{Cnot} gates.
    A local \Gate{Toffoli} gate is then applied, followed by reversing the \Gate{Swap} network, adding another $\Theta(r)$ \Gate{Cnot} gates.
    Therefore, the total number of \Gate{Cnot} gates is $\Theta(r)$.
\end{proof}

\corgenericStoqMA*

\begin{proof}
    Consider a long-range \cl{StoqMA} circuit with gates that act on qubits with a maximum separation of $M$.
    Using the \Gate{Swap} network approach from Proposition \ref{prop:toff-decomp}, each long-range gate can be transformed using $\Theta(M)$ local \Gate{Cnot} gates.
    Since each gate in the worst case requires $\Theta(M)$ \Gate{Swap} gates and each \Gate{Swap} gate is composed of 3 \Gate{Cnot} gates, the overall transformation requires $\Theta(M)$ \Gate{Cnot} gates per long-range gate.
    Thus, the entire circuit can be transformed incurring a $\Theta(M)$ increase in the number of gates.
\end{proof}

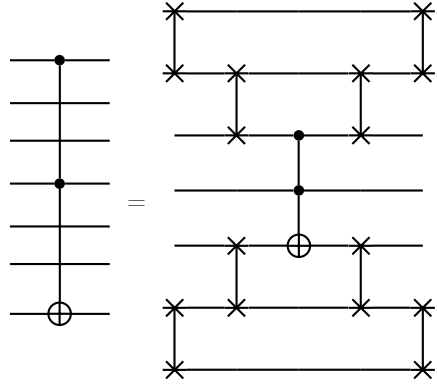
\begin{figure}[!ht]
    \centering
    \begin{quantikz}
        & \ctrl{3} & \qw\\
        & \qw & \qw \\
        & \qw & \qw \\
        & \ctrl{3} & \qw\\
        & \qw & \qw \\
        & \qw & \qw \\
        & \targ{} & \qw\\
    \end{quantikz}~~= 
    \begin{quantikz}
        \swap{1} & \qw      & \qw      & \qw      & \swap{1}\\
        \targX{} & \swap{1} & \qw      & \swap{1} & \targX{}\\
        \qw      & \targX{} & \ctrl{1} & \targX{} & \qw     \\
        \qw      & \qw      & \ctrl{1} & \qw      & \qw     \\
        \qw      & \swap{1} & \targ{}  & \swap{1} & \qw     \\
        \swap{1} & \targX{} & \qw      & \targX{} & \swap{1}\\
        \targX{} & \qw      & \qw      & \qw      & \targX{}\\
    \end{quantikz}
    \caption{Long-range Toffoli gate decomposed to nearest-neighbour gates via a swap network.}
    \label{fig:toff_one}
\end{figure}

\section{Statistics of Circuit Mappings}\label{app:statistics}
The circuit mappings considered in this work first take general \cl{StoqMA} circuits to ones comprised of only nearest neighbour gates.
This requires a number of \Gate{Swap} gates to reduce the range of the gates.
The process of doing so trivially increases the number of gates in the circuit but also retains the original circuit statistics.
This is an artefact of the circuit mapping --- they are equivalent circuits.
The second mapping is from the nearest neighbour circuit to one that is spatially sparse.
This entails an increase in the number of gates and also the number of ancilla qubits required.
Since we are still dealing with \cl{StoqMA} circuits, we know the size of the ancillae registers is characterised by the size of the input $x$.
Moreover, we let $n = \abs{x}$, the proof state be of size $w = \poly{n}$ and also have $m$ and $p$ many $\ket{0}$ and $\ket{+}$ ancillae respectively.
The size of $m$ and $p$ are also bounded by some polynomial in $n$.
This means the circuit space is still polynomial in size.
The number of gates in a circuit is, of course, also bounded by some polynomial in $n$ in order for the verification to be efficient.
Assuming we have some $m$ and $p$ ancillae, it is acceptable to increase both $m$ and $p$ either by a polynomial factor or just to some new polynomial and leave the class definition intact.
The additional ancillae we might want to add need not be acted on by any of the gates in the circuit.
While this would be practically inefficient\footnote{Inefficient here refers to the act of creating the necessary circuit, not the time execution of the circuit.}, it is still a valid \cl{StoqMA} circuit.

In addition to $w$, $m$, $p$ and $L$ being polynomially bounded by $n$, the completeness and soundness statistics, $\alpha$ and $\beta$ are functions of $n$.
What this tells us is that certain circuit modifications will not change the fundamental requirements of the circuit.
If a circuit modification requires a larger input register, then these quantities would be functions of the new input size.
However, the circuit mappings we consider do not change the input size.
The only thing that changes is the number of gates and the number of ancillae.
While this, in general, is not a guarantee that the statistics will be preserved, we will show that the statistics are indeed preserved.

Consider this simple circuit comprised of gates $R_1$, $R_2$ and $R_3$, which are classically reversible gates (see~\cref{fig:circ}).
Let's denote $U = R_3\,R_2\,R_1$.
Our input state is $\ket{x,\xi,0,+}$.
Let us also say that 
\begin{align*}
    R_1\ket{x,\xi,0,+} &= \sum_{x_1,x_2,x_3,x_4 \in \{0,1\}} c_{x}\ket{x_1,x_2,x_3,x_4} \eqqcolon \ket{\psi_1},\\
    R_2\,R_1\ket{x,\xi,0,+} &= \sum_{y_1,y_2,y_3,y_4 \in \{0,1\}} c_{y}\ket{y_1,y_2,y_3,y_4} \eqqcolon R_2\ket{\psi_1} = \ket{\psi_2},\\
    R_3\,R_2\,R_1\ket{x,\xi,0,+} &= \sum_{z_1,z_2,z_3,z_4 \in \{0,1\}} c_{y}\ket{z_1,z_2,z_3,z_4} \eqqcolon R_2\ket{\psi_2} = \ket{\psi_3}.
\end{align*}

This circuit in \cref{fig:circ} has a completeness and soundness characterised by the size of the input state $\ket{x}$.
\begin{align*}
    \textbf{Completeness:} &\quad \exists \ket{\xi} \text{ s.t. } {\rm Pr}_C(1) \coloneqq \bra{x,\xi,0,+} U^\dagger\,\Pi^+_0\,U \ket{x,\xi,0,+} \geq \alpha(n), \\
    \textbf{Soundness:} &\quad \forall \ket{\xi} \text{ s.t. } {\rm Pr}_S(1) \coloneqq \bra{x,\xi,0,+} U^\dagger\,\Pi^+_0\,U \ket{x,\xi,0,+} \leq \beta(n).
\end{align*}

\begin{figure}[!ht]
    \centering
    \begin{quantikz}
        \lstick{$\ket{x}$} & \gate[2]{R_1}\slice[style=quantumpurple]{} & \qw & \qw & \meterD{+\vphantom{0}} \\
        \lstick{$\ket{\xi}$} & \qw & \gate[3]{R_2}\slice[style=quantumpurple]{} & \qw & \qw \\
        \lstick{$\ket{0}$}& \qw & \qw & \gate[2]{R_3}\slice[style=quantumpurple]{} & \qw \\
        \lstick{$\ket{+}$}& \qw & \qw & \qw & \qw
    \end{quantikz}
    \caption{A simple circuit with gates $R_1$, $R_2$ and $R_3$ acting on an input state $\ket{x,\xi,0,+}$.}
    \label{fig:circ}
\end{figure}

Let us define a new circuit that is a modification of the one above (see~\cref{fig:circ_swap}).

\begin{figure}[!ht]
    \centering
    \begin{quantikz}
        \lstick[4]{$\ket{x,\xi,0,+}$} & \gate[2]{R_1} & \swap{4}       & \qw       & \qw       & \qw           & \qw                & \qw          & \qw            & \qw       & \qw       & \qw               & \qw\\
                                       & \qw           & \qw            & \swap{4}  & \qw       & \qw           & \qw               & \qw          & \qw            & \qw       & \qw       & \qw               & \qw\\
                                       & \qw           &\qw             & \qw       & \swap{4}  & \qw           & \qw               & \qw          & \qw            & \qw       & \qw       & \qw               & \qw\\
                                       & \qw           &\qw             & \qw       & \qw       & \swap{4}       & \qw              & \qw          & \qw            & \qw       & \qw       & \qw               & \qw\\
        \lstick[4]{$\ket{0}$}           & \qw          & \targX{}       & \qw       & \qw       & \qw               & \qw           & \swap{4}     & \qw            & \qw       & \qw       & \qw               & \qw\\
                                        & \qw          & \qw            & \targX{}  & \qw       & \qw              & \gate[3]{R_2}  & \qw          & \swap{4}       & \qw       & \qw       & \qw               & \qw\\
                                        & \qw          & \qw            & \qw       & \targX{}  & \qw               & \qw           & \qw          & \qw            & \swap{4}  &\qw        & \qw               & \qw\\
                                        & \qw           &\qw             & \qw       & \qw      & \targX{}          & \qw           & \qw          & \qw            & \qw       & \swap{4}  & \qw               & \qw\\
        \lstick[4]{$\ket{0}$}           & \qw          & \qw            & \qw       & \qw       & \qw               & \qw           & \targX{}     & \qw            & \qw       & \qw       & \qw               & \meterD{+\vphantom{0}}\\
                                        & \qw          & \qw            & \qw       & \qw       & \qw              & \qw            & \qw          & \targX{}       & \qw       & \qw       & \qw               & \qw \\
                                        & \qw          & \qw            & \qw       & \qw       & \qw               & \qw           & \qw          & \qw            & \targX{}  & \qw       & \gate[2]{R_3}     & \qw\\
                                        & \qw           &\qw             & \qw       & \qw      & \qw               & \qw           & \qw          & \qw            & \qw       & \targX{}  & \qw               & \qw\\
    \end{quantikz}
    \caption{A modified circuit with gates $R_1$, $R_2$ and $R_3$ acting on an input state $\ket{x,\xi,0,+}$, but with a swap network in between the gates.}
    \label{fig:circ_swap}
\end{figure}

The point of the circuit in \cref{fig:circ_swap} is to make it so that each qubit is only acted on by a constant number of gates.
We can define $V = R_3\,S_{4 \leftrightarrow 11}\,R_2\,S_{0\leftrightarrow 7}\,R_1$, where $S_{0\leftrightarrow 7}$ and $S_{4\leftrightarrow 11}$ are the \Gate{Swap} gate networks sandwiched between the $R$ gates.
Our input is now a state $\ket{x,\xi,0,+}\ket{0000}\ket{0000}$.
Let's consider the action of $V$ on this state and define $\ket{\boldsymbol{0}} = \ket{0000}$.
\begin{align*}
    R_1\ket{x,\xi,0,+}\ket{\boldsymbol{0}}\ket{\boldsymbol{0}} &= \ket{\psi_1}\ket{\boldsymbol{0}}\ket{\boldsymbol{0}}, \\
    S_{3,7}\,S_{2,6}\,S_{1,5}\,S_{0,4}\ket{\psi_1}\ket{\boldsymbol{0}}\ket{\boldsymbol{0}}, &= \ket{\boldsymbol{0}}\ket{\psi_1}\ket{\boldsymbol{0}},\\
    R_2\ket{\boldsymbol{0}}\ket{\psi_1}\ket{\boldsymbol{0}} &= \ket{\boldsymbol{0}}\ket{\psi_2}\ket{\boldsymbol{0}},\\
    S_{7,11}\,S_{6,10}\,S_{5,9}\,S_{4,8}\ket{\boldsymbol{0}}\ket{\psi_2}\ket{\boldsymbol{0}} &= \ket{\boldsymbol{0}}\ket{\boldsymbol{0}}\ket{\psi_2},\\
    R_3\ket{\boldsymbol{0}}\ket{\boldsymbol{0}}\ket{\psi_2} &= \ket{\boldsymbol{0}}\ket{\boldsymbol{0}}\ket{\psi_3}.
\end{align*}
The output statistics are characterised by 
\begin{align*}
    \Pr(1) &\coloneqq \bra{x,\xi,0,+}\bra{\boldsymbol{0}}\bra{\boldsymbol{0}} V^\dagger\,\Pi^+_8\,V\ket{x,\xi,0,+}\ket{\boldsymbol{0}}\ket{\boldsymbol{0}},\\
    &= \braket{\boldsymbol{0}}\,\braket{\boldsymbol{0}}\,\bra{\psi_3} \Pi^+_8 \ket{\psi_3}.
\end{align*}
Note that by $\braket{\boldsymbol{0}}\,\braket{\boldsymbol{0}}\,\bra{\psi_3} \Pi^+_8 \ket{\psi_3}$ we specifically mean \\$\braket{\boldsymbol{0}}_{0,1,2,3}\,\braket{\boldsymbol{0}}_{4,5,6,7}\,\bra{\psi_3}_{8,9,10,11} \Pi^+_8 \ket{\psi_3}_{8,9,10,11}$.
Let's compare this with the statistics of the first circuit, 
\begin{equation*}
    \bra{x,\xi,0,+}_{0,1,2,3} U^\dagger\,\Pi^+_0\,U \ket{x,\xi,0,+}_{0,1,2,3} = \bra{\psi_3}_{0,1,2,3} \Pi^+_0 \ket{\psi_3}_{0,1,2,3}.
\end{equation*}
Clearly then, $\bra{\psi_3}_{0,1,2,3} \Pi^+_0 \ket{\psi_3}_{0,1,2,3}$ and $\bra{\psi_3}_{8,9,10,11} \Pi^+_8 \ket{\psi_3}_{8,9,10,11}$ are equivalent.
Hence we conclude the statistics are preserved.
Moreover, in the case where $U$ accepts, $V$ also accepts with the same probability.
This idea will generalise and hold when we add more gates and registers.
The important point is that the \Gate{Swap} network we do between each gate of the original circuit effectively transports the bulk of the input to the next array of registers, leaving behind a trail of $\ket{0}$ ancillae.

We address one final and important point --- why can we do this? Well, on Arthur's ``controllable'' side of the \cl{StoqMA} circuit, as mentioned above already, it is possible to increase the number of $\ket{0}$-ancillae and $\ket{+}$-ancillae by a polynomial amount, i.e., $m\mapsto m' = m + \poly{n}$ and $p\mapsto p' = p + \poly{n}$.
The completeness and soundness are functions of $n$, not $m$ and $p$.
Furthermore, the circuit mapping we have considered is exact.
Row one of the construction takes in the input $x$, the proof $\xi$ and a fraction of the ancillae.
The remaining rows are initialised as $\ket{0}$-ancillae.
This fits the definition of the class as per \cref{rmk:Merlins_message}.
The consequence is that Merlin cannot cheat any more than before.
Hence, the statistics are preserved.
\end{document}

%% file: tikz_spspgraph.tex
\tikzset{
    sp-sp-graph/.pic = {
        \foreach \x/\l in {0/0,1.5/1,3/2,4.5/3,6/4,7.5/5,9/6,10.5/7,12/8,13.5/9}{
            \node [draw=black!45, circle, fill=black!45, inner sep=1.5] (a\l) at (\x,0){};
            }
        \foreach \x/\l in {0/0,1.5/1,3/2,4.5/3,6/4,7.5/5,9/6,10.5/7,12/8,13.5/9}{
            \node [draw=black!45, circle, fill=black!45, inner sep=1.5] (b\l) at (\x,2){};
            }
        \foreach \x/\l in {0/0,1.5/1,3/2,4.5/3,6/4,7.5/5,9/6,10.5/7,12/8,13.5/9}{
            \node [draw=black!45, circle, fill=black!45, inner sep=1.5] (c\l) at (\x,4){};
            }
        \node [] (r1) at (15,0) {~Round $1$}; \node [] (r2) at (15,2) {~Round $2$}; \node [] (r3) at (15,4) {~Round $3$};
        
        \node[draw, thick, fit=(a3) (a5), rounded corners, inner sep=7.5, pattern={Lines[distance=2mm,angle=45,line width=0.7mm]}, pattern color= quantumpurple!65] {};
        \foreach \x\l in {4.5/3,6/4,7.5/5}{
            \node [draw=black, circle, fill=black, inner sep=1.5] (a\l) at (\x,0){};
        }
        \node[] (1) at (6,0.65){$R_1$};
    
        \node[draw,thick, fit=(b5) (b7), rounded corners, inner sep=7.5, pattern={Lines[distance=2mm,angle=45,line width=0.7mm]}, pattern color= quantumpurple!65] {};
        \foreach \x\l in {7.5/5,9/6,10.5/7}{
            \node [draw=black, circle, fill=black, inner sep=1.5] (b\l) at (\x,2){};
        }
        \node[] (1) at (9,2.65){$R_2$};
        \node[draw,thick, fit=(c1) (c3), rounded corners, inner sep=7.5, pattern={Lines[distance=2mm,angle=45,line width=0.7mm]}, pattern color= quantumpurple!65] {};
        \foreach \x\l in {1.5/1,3/2,4.5/3}{
            \node [draw=black, circle, fill=black, inner sep=1.5] (c\l) at (\x,4){};
        }
        \node[] (1) at (3,4.65){$R_3$};
        \draw[dashed, very thick, quantumpurple] (11,-0.65)--(12,-0.65); \draw[very thick, quantumpurple, -latex] (12,-0.65)--(14.15,-0.65)--(14.15,0.65)--(10,0.65);
        \draw[very thick, quantumpurple,|-latex] (-0.65,-0.65)--(4.65,-0.65) node[below, midway, yshift=-5pt] {time cursor}; \draw[dashed, very thick, quantumpurple] (2,0.65)--(-0.65,0.65); \draw[very thick, quantumpurple, -latex] (0.5,0.65)--(-0.65,0.65)--(-0.65,1.35)--(2,1.35); 
    
        \draw[line width=3.5pt, white] (13.5,0.25)--(13.5,1.75); \draw[line width=3.5pt, white] (12,0.25)--(12,1.75);
        \draw[very thick, latex-latex, quantumpurple!65] (13.5,0.25)--(13.5,1.75) node[rotate=90, midway, yshift=-7.5pt, fill=white, inner sep=1.5] {\small \textsc{Swap}}; \draw[very thick, latex-latex, quantumpurple!65] (12,0.25)--(12,1.75) node[rotate=90, midway, yshift=-7.5pt, fill=white, inner sep=1.5] {\small \textsc{Swap}};
    },
    interaction/.pic = {
        \draw[fill=black, thick] (0,0)circle(0.1) node[below, yshift=-5pt] {$u$};
        \draw[fill=black, thick] (3,0)circle(0.1) node[below, yshift=-5pt] {$v$};
        \draw[thick] (0,0)--(3,0) node[above, yshift=2.5pt, midway] {$A_uB_v$};
    },
    subdiv-gadget/.pic = {
        \draw[thick](0,0)--(3,0) node[midway, above, yshift=2.55pt] {$P_uP_v + P_u^\dagger P_v^\dagger$}; 
        \draw[fill=black, thick] (0,0)circle(0.1) node[below, yshift=-5pt] {$u$};
        \draw[fill=black, thick] (3,0)circle(0.1) node[below, yshift=-5pt] {$v$};
        \begin{scope}[xshift=5cm]
            \draw[thick](0,0)--(3,0) node[midway, above, yshift=2.55pt] {$P_uS^+_c + P_u^\dagger S^-_c$}; \draw[thick](3,0)--(6,0) node[midway, above, yshift=2.55pt] {$S^-_cP_v + S^+_cP_v^\dagger$}; 
            \draw[fill=black, thick] (0,0)circle(0.1) node[below, yshift=-5pt] {$u$};
            \draw[fill=white, thick] (3,0)circle(0.1) node[below, yshift=-5pt] {$c$};
            \draw[fill=black, thick] (6,0)circle(0.1) node[below, yshift=-5pt] {$v$};
        \end{scope}
    },
    cross-gadget/.pic = {
        \draw[thick](0,0)--(3,3) node[midway, above, rotate=45, yshift=7.5pt, xshift=-40pt, fill=white] {$P_wP_s + P_w^\dagger P_s^\dagger$}; \draw[fill=white, draw=none] (1.5,1.5)circle(0.1);\draw[thick](0,3)--(3,0) node[midway, above, rotate=-45, yshift=7.5pt, xshift=40pt, fill=white] {$P_uP_v + P_u^\dagger P_v^\dagger$}; 
        \draw[fill=black, thick] (0,0)circle(0.1) node[below, yshift=-5pt] {$w$};
        \draw[fill=black, thick] (3,0)circle(0.1) node[below, yshift=-5pt] {$v$};
        \draw[fill=black, thick] (0,3)circle(0.1) node[above, yshift=5pt] {$u$};
        \draw[fill=black, thick] (3,3)circle(0.1) node[above, yshift=5pt] {$s$};
        \begin{scope}[xshift=5cm, yshift=-1.5cm]
            \draw[thick](0,0)--(3,3) node[midway, above, rotate=45, yshift=2.5pt] {$P_wS^+_c + P_w^\dagger S^-_c$}; \draw[thick](3,3)--(6,6) node[midway, below, rotate=45, yshift=-2.5pt] {$S^-_cP_s + S^+_cP_s^\dagger$}; \draw[thick](0,6)--(3,3) node[midway, above, rotate=-45, yshift=2.5pt] {$P_uS^+_c + P_u^\dagger S^-_c$}; \draw[thick](3,3)--(6,0) node[midway, below, rotate=-45, yshift=-2.5pt] {$S^-_cP_v + S^+_cP_v^\dagger$}; 
            \draw[dashed] (0,0)--(6,0)--(6,6)--(0,6)--(0,0);
            \draw[fill=black, thick] (0,0)circle(0.1) node[below, yshift=-5pt] {$w$};
            \draw[fill=black, thick] (6,0)circle(0.1) node[below, yshift=-5pt] {$u$};
            \draw[fill=black, thick] (0,6)circle(0.1) node[above, yshift=5pt] {$v$};
            \draw[fill=black, thick] (6,6)circle(0.1) node[above, yshift=5pt] {$s$};
            \draw[fill=white, thick] (3,3)circle(0.1) node[right, xshift=5pt] {$c$};
        \end{scope}
    },
    fork-gadget/.pic = {
        \draw[thick] (0,0)--(1.5,-{3*sqrt(3)}/2) node[midway, below, rotate=-60, yshift=-2.5pt] {$P_uP_v + P_u^\dagger P_v^\dagger$};
        \draw[thick] (1.5,-{3*sqrt(3)}/2)--(3,0) node[midway, below, rotate=60, yshift=-2.5pt] {$P_vP_w + P_v^\dagger P_w^\dagger$}; 
        \draw[fill=black, thick] (0,0)circle(0.1) node[above, yshift=5pt] {$u$};
        \draw[fill=black, thick] (3,0)circle(0.1) node[above, yshift=5pt] {$w$};
        \draw[fill=black, thick] (1.5,-{3*sqrt(3)}/2)circle(0.1) node[below, yshift=-5pt] {$v$};
        \begin{scope}[xshift=5cm]
            \draw[thick] (0,0)--(1.5,-{3*sqrt(3)}/2) node[midway, below, rotate=-60, yshift=-2.5pt] {$P_uS^+_c + P_u^\dagger S^-_c$}; \draw[thick] (1.5,-{3*sqrt(3)}/2)--(3,0) node[midway, below, rotate=60, yshift=-2.5pt] {$P_wS^+_c + P_w^\dagger S^-_c$}; \draw[thick] (1.5,-{3*sqrt(3)}/2)--(3,-{3*sqrt(3)}) node[midway, below, rotate=-60, yshift=-2.5pt] {$S^-_cP_v + S^+_cP_v^\dagger$};
            \draw[dashed] (0,0)--(3,0);
            \draw[fill=black, thick] (0,0)circle(0.1) node[above, yshift=5pt] {$u$};
            \draw[fill=black, thick] (3,0)circle(0.1) node[above, yshift=5pt] {$w$};
            \draw[fill=white, thick] (1.5,-{3*sqrt(3)}/2)circle(0.1) node[above, yshift=5pt] {$c$};
            \draw[fill=black, thick] (3,-{3*sqrt(3)})circle(0.1) node[below, yshift=-5pt] {$v$};
        \end{scope}
    },
    triangle-gadget/.pic = {
        \draw[thick] (0,0)--(1.5,-{3*sqrt(3)}/2)--(3,0); 
        \draw[dashed] (0,0)--(3,0);
        \draw[dashed] (0,0)--(-0.75,0);\draw[dashed] (0,0)--(-0.75,-0.5);\draw[dashed] (0,0)--(-0.75,0.5);
        \draw[dashed] (3,0)--(3.75,0);\draw[dashed] (3,0)--(3.75,-0.5);\draw[dashed] (3,0)--(3.75,0.5);
        \draw[dashed] (1.5,-{3*sqrt(3)}/2)--(1.5,-{3*sqrt(3)}/2-0.75);\draw[dashed] (1.5,-{3*sqrt(3)}/2)--(2,-{3*sqrt(3)}/2-0.75);\draw[dashed] (1.5,-{3*sqrt(3)}/2)--(1,-{3*sqrt(3)}/2-0.75);
        \draw[fill=black, thick] (0,0)circle(0.1) node[above, yshift=5pt] {$u$};
        \draw[fill=black, thick] (3,0)circle(0.1) node[above, yshift=5pt] {$w$};
        \draw[fill=black, thick] (1.5,-{3*sqrt(3)}/2)circle(0.1) node[below, yshift=-5pt, fill=white] {$v$};
        \begin{scope}[xshift=6cm]
            \draw[thick] (0,0)--(1.5,-{3*sqrt(3)}/2)--(3,0); 
            \draw[dashed] (0,0)--(3,0);
            \draw[dashed] (0,0)--(-0.75,0);\draw[dashed] (0,0)--(-0.75,-0.5);\draw[dashed] (0,0)--(-0.75,0.5);
            \draw[dashed] (3,0)--(3.75,0);\draw[dashed] (3,0)--(3.75,-0.5);\draw[dashed] (3,0)--(3.75,0.5);
            \draw[dashed] (1.5,-{3*sqrt(3)}/2)--(1.5,-{3*sqrt(3)}/2-0.75);\draw[dashed] (1.5,-{3*sqrt(3)}/2)--(2,-{3*sqrt(3)}/2-0.75);\draw[dashed] (1.5,-{3*sqrt(3)}/2)--(1,-{3*sqrt(3)}/2-0.75);
            \draw[fill=black, thick] (0,0)circle(0.1) node[above, yshift=5pt] {$u$};
            \draw[fill=black, thick] (3,0)circle(0.1) node[above, yshift=5pt] {$w$};
            \draw[fill=black, thick] (1.5,-{3*sqrt(3)}/2)circle(0.1) node[below, yshift=-5pt, fill=white] {$v$};
            \draw[fill=white, thick] (0.75,-{3*sqrt(3)}/4)circle(0.1) node[below, yshift=-5pt] {$c_1$};
            \draw[fill=white, thick] (2.25,-{3*sqrt(3)}/4)circle(0.1) node[below, yshift=-5pt] {$c_2$};
        \end{scope}
        \begin{scope}[xshift=12cm]
            \draw[thick] (0,0)--(1.5,-{3*sqrt(3)}/2)--(3,0); 
            \draw[dashed] (0,0)--(3,0);
            \draw[dashed] (0,0)--(-0.75,0);\draw[dashed] (0,0)--(-0.75,-0.5);\draw[dashed] (0,0)--(-0.75,0.5);
            \draw[dashed] (3,0)--(3.75,0);\draw[dashed] (3,0)--(3.75,-0.5);\draw[dashed] (3,0)--(3.75,0.5);
            \draw[dashed] (1.5,-{3*sqrt(3)}/2 -1.5)--(1.5,-{3*sqrt(3)}/2-0.75 -1.5);\draw[dashed] (1.5,-{3*sqrt(3)}/2 -1.5)--(2,-{3*sqrt(3)}/2-0.75 -1.5);\draw[dashed] (1.5,-{3*sqrt(3)}/2 -1.5)--(1,-{3*sqrt(3)}/2-0.75 -1.5);
            \draw[thick] (1.5,-{3*sqrt(3)}/2)--(1.5,-{3*sqrt(3)}/2 -1.5);
            \draw[thick] (0.75,-{3*sqrt(3)}/4)--(2.25,-{3*sqrt(3)}/4);
            \draw[fill=black, thick] (0,0)circle(0.1) node[above, yshift=5pt] {$u$};
            \draw[fill=black, thick] (3,0)circle(0.1) node[above, yshift=5pt] {$w$};
            \draw[fill=white, thick] (1.5,-{3*sqrt(3)}/2)circle(0.1) node[above, yshift=5pt] {$c_3$};
            \draw[fill=black, thick] (1.5,-{3*sqrt(3)}/2  -1.5)circle(0.1) node[below, yshift=-5pt, fill=white] {$v$};
            \draw[fill=white, thick] (0.75,-{3*sqrt(3)}/4)circle(0.1) node[below, yshift=-5pt] {$c_1$};
            \draw[fill=white, thick] (2.25,-{3*sqrt(3)}/4)circle(0.1) node[below, yshift=-5pt] {$c_2$};
        \end{scope}
    },
    localisation/.pic = {
        \draw[thick] (0,0)--(2,1);\draw[thick] (0,0)--(-1.5,1);\draw[thick] (0,0)--(1.5,-0.75);\draw[thick] (0,0)--(0,1.2);
        \draw[dashed] (0,0)--(-1,-1);\draw[dashed] (0,0)--(-1.2,-0.3);\draw[dashed] (0,0)--(0.6,-1.2);
        \draw[fill=black, thick] (0,0)circle(0.1) node[below, yshift=-5pt] {$u$};
        \draw[fill=black, thick] (2,1)circle(0.1);
        \draw[fill=black, thick] (-1.5,1)circle(0.1);
        \draw[fill=black, thick] (1.5,-0.75)circle(0.1);
        \draw[fill=black, thick] (0,1.2)circle(0.1);
        \begin{scope}[xshift=5cm]
            \draw[thick] (0,0)--(2,1);\draw[thick] (0,0)--(-1.5,1);\draw[thick] (0,0)--(1.5,-0.75);\draw[thick] (0,0)--(0,1.2);\draw[dashed] (0,0)--(0.6,-1.2);
            \draw[dashed] (0,0)--(-1,-1);\draw[dashed] (0,0)--(-1.2,-0.3);
            \draw[fill=black, thick] (0,0)circle(0.1) node[below, yshift=-5pt] {$u$};
            \draw[fill=black, thick] (2,1)circle(0.1);\draw[fill=white, thick] (1,0.5)circle(0.1);
            \draw[fill=black, thick] (-1.5,1)circle(0.1);\draw[fill=white, thick] (-0.75,0.5)circle(0.1);
            \draw[fill=black, thick] (1.5,-0.75)circle(0.1);\draw[fill=white, thick] (0.75,-0.375)circle(0.1);
            \draw[fill=black, thick] (0,1.2)circle(0.1);\draw[fill=white, thick] (0,0.6)circle(0.1);
        \end{scope}
    },
    forking-loops/.pic = {
        \draw[thick] (0,1.5)ellipse(1 and 1.5);
        \draw[fill=black,thick] (0,0)circle(0.1) node[below, yshift=-5pt] {$u$};
        \draw[fill=black,thick] (0,3)circle(0.1) node[above, yshift=5pt] {$v$};
        \begin{scope}[xshift=4cm]
            \draw[thick] (0,1.5)ellipse(1 and 1.5);
            \draw[fill=black,thick] (0,0)circle(0.1) node[below, yshift=-5pt] {$u$};
            \draw[fill=black,thick] (0,3)circle(0.1) node[above, yshift=5pt] {$v$};
            \draw[fill=white,thick] (1,1.5)circle(0.1) node[right, xshift=5pt] {$c_1$};
            \draw[fill=white,thick] (-1,1.5)circle(0.1) node[left, xshift=-5pt] {$c_2$};
        \end{scope}
        \begin{scope}[xshift=8cm]
            \draw[thick] (1,1.5)--(0,3)--(-1,1.5)--(0,0.75)--(1,1.5)--(-1,1.5);
            \draw[thick] (0,0)--(0,0.75);
            \draw[fill=black,thick] (0,0)circle(0.1) node[below, yshift=-5pt] {$u$};
            \draw[fill=black,thick] (0,3)circle(0.1) node[above, yshift=5pt] {$v$};
            \draw[fill=white,thick] (1,1.5)circle(0.1) node[right, xshift=5pt] {$c_1$};
            \draw[fill=white,thick] (-1,1.5)circle(0.1) node[left, xshift=-5pt] {$c_2$};
            \draw[fill=white,thick] (0,0.75)circle(0.1) node[above, yshift=5pt] {$c_3$};
        \end{scope}
    },
    sub-cross-edge/.pic = {
        \foreach \x in {1,2,3,4}{
            \draw[dashed](\x,1.25)--(\x,-1.25);
            \draw[thick] (\x,0.5)--(\x,-0.5);
            \draw[fill=white,draw=none] (\x,0)circle(0.1);
        }
        \draw[thick] (0,0)--(5,0);
        \draw[fill=black,thick] (0,0)circle(0.1);
        \draw[fill=black,thick] (5,0)circle(0.1);
        \begin{scope}[xshift=7cm]
            \foreach \x in {1,2,3,4}{
                \draw[dashed](\x,1.25)--(\x,-1.25);
                \draw[thick] (\x,0.5)--(\x,-0.5);
                \draw[fill=white,draw=none] (\x,0)circle(0.1);
            }
            \draw[thick] (0,0)--(5,0);
            \draw[fill=black,thick] (0,0)circle(0.1);
            \draw[fill=black,thick] (5,0)circle(0.1);
            \foreach \x in {1.5,2.5,3.5}{
                \draw[fill=white,thick] (\x,0)circle(0.1);
            }
        \end{scope}
    },
    localise-cross/.pic = {
        \draw[thick] (3,0)--(0,3);\draw[fill=white,draw=none] (1.5,1.5)circle(0.1);
        \draw[thick] (0,0)--(3,3);
        \draw[fill=black,thick] (0,0)circle(0.1);
        \draw[fill=black,thick] (3,0)circle(0.1);
        \draw[fill=black,thick] (0,3)circle(0.1);
        \draw[fill=black,thick] (3,3)circle(0.1);
        \begin{scope}[xshift=5cm]
            \draw[thick] (3,0)--(0,3);\draw[fill=white,draw=none] (1.5,1.5)circle(0.1);
            \draw[thick] (0,0)--(3,3);
            \draw[fill=black,thick] (0,0)circle(0.1);
            \draw[fill=black,thick] (3,0)circle(0.1);
            \draw[fill=black,thick] (0,3)circle(0.1);
            \draw[fill=black,thick] (3,3)circle(0.1);
            \draw[fill=white,thick] (1,1)circle(0.1); \draw[fill=white,thick] (2,2)circle(0.1);
            \draw[fill=white,thick] (1,2)circle(0.1); \draw[fill=white,thick] (2,1)circle(0.1);
        \end{scope}
    },
    rho-decomp/.pic = {
        \draw[thick] (0,0)--(3,0) node[above, midway, yshift=2.5pt] {$H_{uv}$};
        \draw[fill=black,thick] (0,0)circle(0.1) node[below, yshift=-5pt] {$u$};
        \draw[fill=black,thick] (3,0)circle(0.1) node[below, yshift=-5pt] {$v$};
        \begin{scope}[xshift=5cm]
            \draw[thick] (0,0)--(3,0) node[above, midway, yshift=2.5pt] {$\rho^2_u\rho^0_v + \rho^1_u\rho^0_v$};
            \draw[thick] (1.5,0)ellipse(1.5 and 1.25) node[below, yshift=-12.5pt] {$\rho^0_u\rho^0_v$} node[above, yshift=37.5pt] {$\rho^3_u\rho^3_v$};
            \draw[fill=black,thick] (0,0)circle(0.1) node[below, yshift=-5pt, fill=white] {$u$};
            \draw[fill=black,thick] (3,0)circle(0.1) node[below, yshift=-5pt, fill=white] {$v$};
        \end{scope}
    },
    legal-fork/.pic = {
        \draw[thick] (1.5,0)ellipse(1.5 and 1) node[below, yshift=-7.5pt] {$\rho^0_u\rho^0_v$} node[above, yshift=28.5pt] {$\rho^2_u\rho^0_v + \rho^1_u\rho^0_v$};
        \draw[thick] (-1.5,0)ellipse(1.5 and 1) node[below, yshift=-7.5pt] {$\rho^0_u\rho^0_w$} node[above, yshift=28.5pt] {$\rho^2_u\rho^0_w + \rho^1_u\rho^0_w$};
        \draw[fill=black,thick] (0,0)circle(0.1) node[below, yshift=-5pt, fill=white] {$u$};
        \draw[fill=black,thick] (3,0)circle(0.1) node[below, yshift=-5pt, fill=white] {$v$};
        \draw[fill=black,thick] (-3,0)circle(0.1) node[below, yshift=-5pt, fill=white] {$w$};
        \begin{scope}[xshift=8cm]
            \draw[thick] (-3,0)--(-1.5,1.25)--(0,0.75)--(1.5,1.25)--(3,0);
            \draw[thick] (-3,0)--(-1.5,-1.25)--(0,-0.75)--(1.5,-1.25)--(3,0);
            \draw[thick] (0,0.75)--(0,-0.75);
            \draw[dashed] (-1.5,1.25)--(1.5,1.25); \draw[dashed] (-1.5,-1.25)--(1.5,-1.25);
            \draw[fill=black,thick] (0,0)circle(0.1) node[below, yshift=-5pt, fill=white] {$u$};
            \draw[fill=black,thick] (3,0)circle(0.1) node[below, yshift=-5pt, fill=white] {$v$};
            \draw[fill=black,thick] (-3,0)circle(0.1) node[below, yshift=-5pt, fill=white] {$w$};
            \draw[fill=white,thick] (1.5,1.25)circle(0.1) node[above, yshift=5pt] {$c_5$};
            \draw[fill=white,thick] (1.5,-1.25)circle(0.1) node[below, yshift=-5pt] {$c_6$};
            \draw[fill=white,thick] (-1.5,1.25)circle(0.1) node[above, yshift=5pt] {$c_1$};
            \draw[fill=white,thick] (-1.5,-1.25)circle(0.1) node[below, yshift=-5pt] {$c_2$};
            \draw[fill=white,thick] (0,0.75)circle(0.1) node[above, yshift=0pt] {$c_3$};
            \draw[fill=white,thick] (0,-0.75)circle(0.1) node[below, yshift=-0.5pt] {$c_4$};
        \end{scope}
    },
    planar-to-square/.pic = {
        \def\width{10}
        \def\step{0.5}
        \def\term{\width}
        \pgfmathsetmacro\term{\term-\step}
    
        \foreach \x/\y in {1/1.5, 2/8, 3/4, 7/7, 8.5/2, 9.5/9.5}{
            \draw[draw=none, fill=gray!45, draw=none] (\x-0.5,\y-0.5) rectangle (\x+0.5,\y+0.5);
        };
    
        \foreach \i in {0.5, 1, 1.5, ..., 9.5}{
            \draw[quantumpurple!55] (\i,0)--(\i,\width);
            \draw[quantumpurple] (0,\i)--(\width,\i);
        };
    
        \draw[thick, gray, dashed] (1,1.5)--(8.5,2)--(9.5,9.5)--(7,7)--(3,4)--(2,8);
        \draw[thick, gray, dashed] (1,1.5)--(3,4);
        \draw[thick, gray, dashed] (7,7)--(8.5,2);
        \draw[thick, gray, dashed] (3,4)--(8.5,2);
    
        \foreach \x/\y in {1/1.5, 2/8, 3/4, 7/7, 8.5/2, 9.5/9.5}{
            \draw[fill = black] (\x,\y) circle (2pt);
        };
        
        \draw[very thick] (1,1.5)--(5,1.5)--(5,2)--(7,2)--(7,1.5)--(8.5,1.5)--(8.5,2);
        \draw[very thick] (8.5,2)--(9,2)--(9,8)--(9.5,8)--(9.5,9.5);
        \draw[very thick] (8.5,2)--(8.5,2.5)--(8,2.5)--(8,4.5)--(7.5,4.5)--(7.5,6.5)--(7,6.5)--(7,7);
        \draw[very thick] (9.5,9.5)--(9,9.5)--(9,9)--(8.5,9)--(8.5,8.5)--(8,8.5)--(8,8)--(7.5,8)--(7.5,7.5)--(7,7.5)--(7,7);
        \draw[very thick] (7,7)--(6.5,7)--(6.5,6.5)--(6,6.5)--(6,6)--(5,6)--(5,5.5)--(4.5,5.5)--(4.5,5)--(4,5)--(4,4.5)--(3.5,4.5)--(3.5,4)--(3,4);
        \draw[very thick] (3,4)--(2.5,4)--(2.5,3.5)--(2.5,3)--(2,3)--(2,2.5)--(1.5,2.5)--(1.5,2)--(1,2)--(1,1.5);
        \draw[very thick] (3,4)--(3,3.5)--(5,3.5)--(5,3)--(7,3)--(7,2.5)--(7.5,2.5)--(7.5,2)--(8.5,2);
        \draw[very thick] (3,4)--(3,4.5)--(2.5,4.5)--(2.5,7)--(2,7)--(2,8);
        \draw[thick, black] (0,0) rectangle (\width,\width);
    },
}